\PassOptionsToPackage{table,dvipsnames}{xcolor}
\documentclass[12pt,a4paper]{article}

\makeatletter
\def\@seccntformat#1{\@ifundefined{#1@cntformat}%
   {\csname the#1\endcsname\quad}%
   {\csname #1@cntformat\endcsname}%
}
\let\oldappendix\appendix
\renewcommand\appendix{%
    \oldappendix
    \newcommand{\section@cntformat}{\appendixname~\thesection:\,\,}
}
\makeatother

\usepackage[natbibapa]{apacite}

\usepackage[utf8]{inputenc}
\usepackage[T1]{fontenc}
\usepackage{amsmath}
\usepackage{amsfonts}
\usepackage{amssymb}
\usepackage{amsthm}
\usepackage{mathtools}
\usepackage{dsfont}
\usepackage{afterpage}
\usepackage[flushleft]{threeparttable}
\usepackage{pdflscape}
\usepackage{rotating}
\usepackage{changepage}
\usepackage{graphicx}
\graphicspath{{./Figure/}}
\usepackage{relsize}
\usepackage{url}
\usepackage{setspace}
\usepackage{multirow}
\usepackage{colortbl}
\usepackage{tabulary}
\usepackage{etoolbox}
\usepackage{booktabs}
\usepackage{siunitx}
\usepackage{array}
\usepackage{caption}
\usepackage{adjustbox}
\usepackage{tabularx}
\usepackage[section]{placeins}
\usepackage[normalem]{ulem}
\usepackage{xcolor}
\usepackage{geometry}
\usepackage{hyperref}

\newtheorem{proposition}{Proposition}
\hypersetup{
  colorlinks,
  citecolor=violet,
  linkcolor=red,
  urlcolor=blue}
\newtheorem{remark}{Remark}
\newtheorem{theorem}{Theorem}
\newtheorem{definition}[theorem]{Definition}

\newcommand{\R}{\mathbb{R}}
\newcommand{\E}{\mathbb{E}}
\theoremstyle{plain}
\newtheorem{stheorem}{Theorem}
\newtheorem{slemma}{Lemma}
\newtheorem{scorollary}{Corollary}
\theoremstyle{remark}
\newtheorem{sremark}{Remark}
\theoremstyle{plain}

\begin{document}

\title{
\textbf{Asymmetric $q$-Gaussian Distributions from Tsallis Maximum Entropy}}
\author{Francesco Cesarone$^{1,*}$, Massimiliano Corradini$^1$, Nicol\`o Giunta$^1$, Serena Rocci$^1$\\
	{\small $^1$\emph{Roma Tre University - Department of Business Studies}}\\
	{\footnotesize francesco.cesarone@uniroma3.it, massimiliano.corradini@uniroma3.it,}\\
	{\footnotesize nicolo.giunta@uniroma3.it, serena.rocci@uniroma3.it}\\
	{\footnotesize $^*$Corresponding author: francesco.cesarone@uniroma3.it}
}
\date{}
\maketitle

\begin{abstract}
Within the framework of non-extensive statistical mechanics, $q$-Gaussian distributions arise as maximum-entropy solutions of Tsallis entropy under suitable escort moment constraints. Although widely used in physics, biology, geology, finance, and related fields, standard $q$-Gaussians are intrinsically symmetric and cannot capture asymmetric fluctuations often observed in empirical data.

We ask whether asymmetry can arise from the Tsallis maximum-entropy principle, without externally skewing or joining distributions and we show that this is achieved by retaining the Tsallis entropy and imposing a global escort mean together with distinct downside and upside escort semi-variances.

The key step is to reformulate the problem of maximum Tsallis entropy in terms of the escort density. The nonlinear escort constraints then become affine, while maximizing Tsallis entropy becomes equivalent to maximizing a strictly concave functional over a convex feasible set. Consequently, any admissible stationary solution is the unique global maximizer of the corresponding variational problem.
The finite-sample approximation properties of the proposed asymmetric $q$-Gaussian are studied through Monte Carlo experiments using synthetic data generated from Azzalini skew-Student distributions. Compared with the classical symmetric $q$-Gaussian, the asymmetric specification systematically provides a closer approximation to synthetic asymmetric heavy-tailed distributions.
Finally, we fit the model to return series from major equity indices, a commodity, and cryptocurrencies, and compare its in-sample goodness of fit with that of a symmetric $q$-Gaussian. The results show that the asymmetric model offers a closer in-sample description of unconditional return distributions, most clearly in the tails.

\end{abstract}

\section{Introduction}
Since the pioneering work of \cite{tsallis88}, non-extensive statistical mechanics has provided a powerful framework for describing a wide variety of complex systems
characterized, for instance, by long-range interactions, memory effects, multifractality, and non-Gaussian fluctuations, see \cite{Tsallis2009} and references therein. Within this framework, $q$-Gaussian distributions emerge as maximizers of Tsallis entropy under suitable generalized  mean and  variance constraints and naturally generate power-law tails. Moreover, $q$-Gaussians play a distinguished theoretical role through generalized central limit theorems in non-extensive statistical mechanics \citep{umarov2008qclt}. Over the last decades $q$-Gaussians have been successfully applied in a broad variety of contexts, including turbulence \citep{beck2001dynamical,beck2003superstatistics}, anomalous diffusion \citep{tsallis1996anomalous,lutz2003anomalous}, plasmas, granular systems, biological and geophysical phenomena (\citealp{Tsallis2009}, and references therein), and financial markets \citep{borland2002option,queiros2005nonextensive}. In many of these systems, empirical distributions exhibit significant deviations from Gaussian statistics, typically characterized by heavy tails, intermittency, and asymmetric fluctuations.

Despite their broad applicability, standard $q$-Gaussian distributions remain intrinsically symmetric around their location parameter. However, many empirical systems exhibit asymmetric fluctuations, often reflecting distinct mechanisms governing large positive and negative deviations. In financial markets, for example, downside fluctuations are typically stronger and more persistent than upside fluctuations, especially during periods of market stress and volatility amplification. More generally, asymmetry naturally arises in several non-equilibrium systems characterized by intermittent dynamics and heterogeneous fluctuation regimes.

Among these applications, financial markets provide a particularly important example of complex interacting systems exhibiting strong non-Gaussian behaviour. Empirical return distributions are well known to display heavy tails, excess kurtosis, asymmetry, volatility clustering, and intermittent episodes of extreme fluctuations \citep{mantegna1999introduction,sornette03,Cont2001}. Such stylized facts are difficult to reconcile with traditional Gaussian-based approaches, which systematically underestimate the probability of large market movements and extreme events.

Most existing studies have focused on symmetric $q$-Gaussian distributions, typically derived under constraints on the escort mean and escort variance. Various asymmetric heavy-tailed constructions have been proposed in the literature, including skewed Student-t families \citep{fernandez1998fat,jones2003skewt} and related skewing mechanisms. These constructions address the statistical problem of representing skewness, but do not in general address the more specific question considered here: how to generate asymmetry while retaining a single global Tsallis maximum-entropy problem.

Recent contributions have sought to extend the Tsallis framework itself to asymmetric settings. 
In particular, \cite{devi21asymmetric} proposed a piecewise model obtained by joining separately normalized half $q$-Gaussians with side-specific non-extensivity and scale parameters, motivated by separate sign-dependent Tsallis and superstatistical arguments. That construction highlighted the relevance of treating positive and negative fluctuations differently. The objective pursued here is distinct: we retain the standard Tsallis entropy with a single entropic index and introduce asymmetry only through joint escort constraints on one globally normalized density.

In this work, we formulate and solve this asymmetric Tsallis maximum-entropy problem by replacing the single escort-variance constraint of the symmetric construction with distinct downside and upside escort semi-variance constraints. The use of escort semi-variances allows positive and negative fluctuations to be treated separately, thereby providing a natural entropy-based mechanism for generating asymmetry without changing the entropy functional or introducing side-specific entropic indices. The resulting family accommodates both heavy tails and asymmetric fluctuation patterns through distinct left and right dispersion structures.

The central methodological step is to use the escort density $p=f^q/\int f^q\,dx$ as the optimization variable\footnote{Hereafter, the standard Lebesgue integral of $g$ over $\mathbb{R}$ will be denoted by $\int g\,dx$ or $\int g(x)\,dx$.}. In the original variable $f$, each escort constraint depends on $f$ both pointwise and through the normalizing factor $\int f^q\,dx$. Passing to $p$ removes this coupling: the constraints become affine, and maximizing $S_q[f]$ becomes equivalent to maximizing a strictly concave functional $J[p]$ over a convex feasible set. Thus, the reformulation turns the original nonlinear variational problem into a standard concave optimization problem. Strict concavity then provides a variational characterization: any feasible stationary candidate is a global maximizer, and such a maximizer is unique. As a consequence, each admissible member satisfying the escort constraints is the unique Tsallis maximum-entropy density for the corresponding escort mean and semi-variance constraint set.

To assess the finite-sample approximation capabilities of the proposed asymmetric $q$-Gaussian, we perform extensive Monte Carlo experiments using synthetic asymmetric heavy-tailed data generated from Azzalini skew-Student distributions \citep{azzalini2003distributions}. The proposed asymmetric $q$-Gaussian is systematically compared with a symmetric $q$-Gaussian benchmark through likelihood-based criteria and goodness-of-fit distance measures, namely the Kolmogorov--Smirnov and the tail-sensitive Anderson--Darling statistics. The model is then fitted to real financial logarithmic returns, daily, 7-day, and 20-day series of nine markets comprising major equity indices (BIST, FTSE, NASDAQ, DJIA, S\&P\,500, BOVESPA), a commodity (Brent), and two cryptocurrencies (Bitcoin and Ethereum), and compared in sample with the symmetric $q$-Gaussian benchmark, with sampling variability assessed by a nonparametric bootstrap. This empirical analysis is intended as a controlled approximation study for unconditional heavy-tailed asymmetric distributions rather than as a model for conditional volatility dynamics or out-of-sample forecasting.

The remainder of the paper is organized as follows. Section \ref{sec:TheoreticalFramework} recalls the Shannon and Tsallis maximum-entropy constructions and the symmetric $q$-Gaussian. Section \ref{sec:asym_construction} introduces the asymmetric $q$-Gaussian family and proves that any feasible member is the unique global maximizer of the Tsallis maximum-entropy problem under the corresponding escort mean and escort semi-variance constraints. A sharp characterization of the existence of multipliers for a prescribed escort mean and prescribed escort semi-variances is stated in Proposition~\ref{prop:inverse-moment-existence}; its proof is given in Section~S1 of the Supplementary Material. Section \ref{sec:empirical} reports the synthetic and real-market experiments, and Section \ref{sec:conclusions} concludes. The full set of result tables is collected in Appendices~\ref{app:mc_tables} and~\ref{app:real_tables}.

\section{Theoretical framework}\label{sec:TheoreticalFramework}

This section recalls the classical Shannon maximum-entropy construction leading to the Gaussian distribution and the standard non-extensive maximum-entropy framework introduced by Tsallis, in which $q$-Gaussian distributions arise naturally under escort mean--variance constraints. This variational structure provides a natural foundation for extending the $q$-Gaussian family to asymmetric settings through escort semi-variance constraints.

\subsection{Maximum Shannon entropy}\label{sec:MaxShannon}
Let $X$ be a continuous random variable with probability density function (pdf) $g$. The Shannon entropy of  $g$ is defined as
\begin{equation}\label{eq:shannon_entropy}
S[g]=-\int g(x)\ln g(x)\,dx.
\end{equation}
Assuming that $g:\mathbb{R}\rightarrow \mathbb{R}$ 
belongs to $\mathcal{G}$, with
\[
\mathcal{G}=\{ g\geq 0 \,\text{a.e.}, \int g(x)\,dx=1, \int (x-\mu)g(x)\,dx=0, \int (x-\mu)^2 g(x)\, dx=\sigma^2\},
\]
with $\mu\in\mathbb{R}$ and $\sigma^2>0$, the maximizer $g^\star$ of Shannon entropy, as is well known, is the Gaussian density
\begin{equation}\label{eq:gaussian}
g^\star(x)=\frac{1}{\sqrt{2\pi\sigma^2}}
e^{-\frac{(x-\mu)^2}{2\sigma^2}}.
\end{equation}
The Gaussian distribution (\ref{eq:gaussian}) therefore provides the natural maximum-entropy baseline when only first- and second-order moment information is prescribed.

Although the Shannon maximizer provides a rigorous benchmark under mean--variance constraints, its exponentially decaying tails are inconsistent with the power-law--like behaviour commonly observed in many empirical datasets \citep{Tsallis2009}.

\subsection{Maximum Tsallis entropy}

A natural extension of the Shannon maximum-entropy framework is obtained by replacing Shannon entropy with a generalized entropy. In particular, non-extensive Tsallis entropy leads to $q$-Gaussian maximizers,  capable of capturing heavy-tailed behaviour while preserving analytical tractability \citep{tsallis88,Tsallis2009, tsallis2004nonextensive,deSouzaTsallis1997}.

 The Tsallis entropy $S_q[f]$ of a pdf $f$ is defined, for $q \neq 1$, as 
\begin{equation}\label{eq:tsallis-entropy}
    S_q[f]=\frac{1-\int f(x)^q\,dx}{q-1},
\end{equation}
provided $\int f(x)^q \, dx <+\infty$, and   it recovers the standard Shannon entropy in the limit $q\rightarrow 1$.

\begin{definition}[Escort probabilities]
Following the standard escort formalism used in non-extensive Tsallis entropy, see \citep{Tsallis2009}, we define the escort probability $f_q(x)$ (at fixed $q>1$)  associated with the pdf $f(x)$ as
\begin{equation}\label{eq:p_vs_f}
    f_q(x)=\frac{f(x)^q}{\|f^q\|},
\end{equation}
where $\|g\|=\int g(x)\,dx$, and the escort $q$-expectation as
\[
\mathbb{E}_q[h(X)] := \int h(x)\, f_q(x)\,dx.
\]
\end{definition}
\noindent
Note that $f_q(x)>0$ almost everywhere whenever $f(x)$ is and that $f_q(x)$ is a pdf since $\int f_q(x)\, dx =1$.
\begin{remark}\label{rem:escort_inversion}
The escort transformation \eqref{eq:p_vs_f} can be inverted. In fact raising Eq.~\eqref{eq:p_vs_f} to the power $1/q$ gives
\begin{equation}\label{eq:a}
f_q(x)^{1/q}=\frac{f(x)}{\|f^q\|^{1/q}}.    
\end{equation}
Integrating Eq. (\ref{eq:a}), using $\|f\|=1$, yields 
\begin{equation}\label{eq:b}
\|f_q^{1/q}\|=\frac{1}{\|f^q\|^{1/q}}.    
\end{equation}
From Eq. (\ref{eq:a}) we have
\begin{equation}\label{eq:f_vs_p}
f(x)=f_q(x)^{1/q}\|f^q\|^{1/q}.
\end{equation}
\end{remark} 
\noindent
Finally, inserting Eq. (\ref{eq:b}) in Eq. (\ref{eq:f_vs_p}) we obtain
\begin{equation}\label{eq:c}
f(x)=\frac{f_q(x)^{1/q}}{\|f_q^{1/q}\|}.    
\end{equation}

For \(\mu_q\in\mathbb{R}\), \(\sigma_q^2>0\) let
\begin{equation}\label{eq:F}
    \mathcal{F}
=
\left\{
f\geq 0\ \text{a.e.} :
\int f(x)\,dx=1,
\int x f_q(x)\,dx=\mu_q,
\int (x-\mu_q)^2 f_q(x)\,dx=\sigma_q^2
\right\}
\end{equation}
be the admissible set of pdfs $f$ such that $\|f^q\|<+\infty$.

\begin{remark}\label{rem:symmetric_qrange}
As is well known, see \citep{Tsallis2009}, for $1<q<3$, the solution of the non-extensive maximum-entropy problem
\[
\max_{f\in \mathcal{F}} S_q[f],
\]
where $S_q[f]$ is given by Eq.~(\ref{eq:tsallis-entropy}) and $\mathcal{F}$ is defined in Eq.~(\ref{eq:F}), is the symmetric $q$-Gaussian distribution
\begin{equation}\label{eq:q-gaussiana_symm}
f^\star(x)=\frac{1}{Z_q}\Bigl[1+\frac{q-1}{3-q}\frac{(x-\mu_q)^2}{\sigma^2_q}\Bigr]^{-\frac{1}{q-1}},
\end{equation}
where \(Z_q\) is the normalizing constant. On the full real line, the range \(1<q<3\) is the natural admissible range for the normalized $q$-Gaussian with finite escort variance. If one requires the ordinary variance of \(f^\star\) to be finite, the stricter condition \(1<q<5/3\) is needed.
\end{remark}
\begin{remark}
As is well known, see \citep{Tsallis2009}, in the limit \(q\to 1\), the Tsallis entropy converges to the Shannon entropy, and the $q$-Gaussian in Eq.~\eqref{eq:q-gaussiana_symm} converges to the standard Gaussian in Eq.~\eqref{eq:gaussian}, with $\mu_1=\mu$ and $\sigma^2_1=\sigma^2$.
\end{remark}

\section{Generalized Maximum Tsallis entropy approach: asymmetric $q$-Gaussian distributions}
\label{sec:asym_construction}

Despite their ability to reproduce fat tails, symmetric $q$-Gaussian distributions remain structurally limited in capturing asymmetries. For instance, in financial applications, empirical return distributions are typically asymmetric, with heavier downside tails and different dispersion patterns for gains and losses \citep{Cont2001}.

 The question addressed here is not how to construct an arbitrary asymmetric heavy-tailed density, but whether asymmetry can be generated by the standard Tsallis maximum-entropy principle itself. The minimal modification of the symmetric formulation is to retain the same entropy and a common entropic index, while replacing the single escort-variance constraint with separate downside and upside escort semi-variance constraints. Because the asymmetry then enters through the constraints rather than through an external modification of the density, the resulting family captures skewed heavy-tailed behaviour while retaining a maximum-entropy interpretation.

We consider a continuous random variable $X$ with pdf $f$ and the Tsallis non-extensive entropy of $f$, with entropic index $q>1$,\footnote{In this paper, we focus on pdfs with support on $\mathbb{R}$.}
\begin{equation}\label{eq:tsallis_entropy_qgt1}
    S_q[f]
    =
    \frac{1-\int f(x)^q\,dx}{q-1}, \qquad q>1.
\end{equation} 

For $\mu_q \in \mathbb{R}$, we impose the escort mean constraint as in the symmetric case,
\begin{equation}\label{eq:escort_mean_original}
    \int(x-\mu_q) f_q(x)\,dx=0,
\end{equation}
where the escort pdf $f_q$ is given by Eq.~(\ref{eq:p_vs_f}). 
To account for asymmetric distributions, we impose different quadratic dispersion constraints on the two sides of the escort mean $\mu_q$:
\begin{equation}\label{eq:escort_semivariances_original}
    \int[(x-\mu_q)_-]^2 f_q(x)\,dx
    =
    \sigma_{q,-}^2,
    \qquad
    \int[(x-\mu_q)_+]^2 f_q(x)\,dx
    =
    \sigma_{q,+}^2.
\end{equation}
In Eq.~(\ref{eq:escort_semivariances_original})
\[
    (t)_-:=\min\{t,0\}
    \qquad\text{and}\qquad
    (t)_+:=\max\{t,0\},
\]
so that the two constraints correspond to downside and upside quadratic deviations with respect to $\mu_q$, respectively.

Thus, the asymmetric non-extensive maximum-entropy problem is
\begin{equation}\label{eq:max_S_q_in_Fpm}
    \max_{f\in \mathcal{F}_\pm} S_q[f], \qquad q>1
\end{equation}
where $S_q[f]$ is defined in Eq.~(\ref{eq:tsallis_entropy_qgt1}), provided that
$
    \int f(x)^q\,dx < \infty,
$
with admissible set $\mathcal{F}_\pm$  given by
\begin{equation}\label{eq:F_pm}
\begin{aligned}
\mathcal{F}_\pm=\Bigg\{ f(x)\geq 0 \ \text{a.e.} \ :\ 
&\int f(x)\,dx=1, \\
&\int(x-\mu_q)f_q(x)\,dx=0, \\
&\int[(x-\mu_q)_-]^2 f_q(x)\,dx=\sigma^2_{q,-}, \\
&\int[(x-\mu_q)_+]^2 f_q(x)\,dx=\sigma^2_{q,+}
\Bigg\}.
\end{aligned}
\end{equation}

\subsection{Escort reformulation of the non-extensive maximum-entropy problem}

A direct treatment in the original variable $f$ is complicated by the normalized escort ratio $f^q/\int f^q\,dx$, which makes the constraints nonlinear and couples the first-order conditions through the unknown integral $\int f^q\,dx$. Since the escort transformation is invertible, we instead take the escort density itself as the primary optimization variable. This preserves the original Tsallis problem while removing the normalization coupling from the constraints. We therefore reformulate problem~(\ref{eq:max_S_q_in_Fpm}) in terms of
\begin{equation}\label{eq:p_escort}
    p(x)
    \equiv
    f_q(x)
    =
    \frac{f(x)^q}{\int f(u)^q\,du}.
\end{equation}
\begin{proposition}
The asymmetric non-extensive maximum-entropy problem~(\ref{eq:max_S_q_in_Fpm}) can be reformulated in terms of the escort density $p$ as
\begin{equation}\label{eq:max_J_in_Ppm}
    \max_{p \in \mathcal{P}_\pm} J[p], \qquad q>1
\end{equation}
with
\begin{equation}\label{eq:J_functional}
    J[p]=\int p(x)^{1/q}\,dx,
\end{equation}
provided $\int p(x)^{1/q}\,dx < +\infty$,
and
\begin{equation}\label{eq:P_pm}
\begin{aligned}
\mathcal{P}_\pm=\Bigg\{ p(x)\geq 0 \ \text{a.e.} \ :\ 
&\int p(x)\,dx=1, \\
&\int(x-\mu_q)p(x)\,dx=0, \\
&\int[(x-\mu_q)_-]^2 p(x)\,dx=\sigma^2_{q,-}, \\
&\int[(x-\mu_q)_+]^2 p(x)\,dx=\sigma^2_{q,+}
\Bigg\}.
\end{aligned}
\end{equation}
\end{proposition}

\begin{proof}
Let $p$ be the escort density associated with $f$, as defined in Eq.~(\ref{eq:p_escort}). Then, using Eq. (\ref{eq:c}) and Eq. (\ref{eq:J_functional}), we have
\begin{equation}\label{eq:J_functional_1}
    f(x)=\frac{p(x)^{1/q}}{\int p(u)^{1/q}\,du}
    =
    \frac{p(x)^{1/q}}{J[p]},
\end{equation}
and hence
\begin{equation}\label{eq:fq_vs_J}
    f(x)^q=\frac{p(x)}{J[p]^q}.
\end{equation}
Therefore,
\begin{equation}\label{eq:int_fq}
  \int f(x)^q \,dx=\frac{1}{J[p]^q},  
\end{equation}
since $\int p(x)\,dx=1$. 
Substituting Eq.~(\ref{eq:int_fq}) into Eq.~(\ref{eq:tsallis_entropy_qgt1}) gives
\begin{equation}\label{eq:Sq_vs_J}
  S_q[f]=\frac{1-J[p]^{-q}}{q-1}.  
\end{equation}
Since the map
\[
    t\mapsto \frac{1-t^{-q}}{q-1}
\]
is strictly increasing for $t>0$ and $q>1$, maximizing $S_q[f]$ is equivalent to maximizing $J[p]$. Finally, under the escort transformation (\ref{eq:p_escort}), the constraints defining $\mathcal{F}_\pm$ in Eq. (\ref{eq:F_pm}) are mapped exactly into those defining $\mathcal{P}_\pm$ in Eq. (\ref{eq:P_pm}), and the claim follows.
\end{proof}
\begin{remark}
Problem~(\ref{eq:max_J_in_Ppm}) differs substantially from
problem~(\ref{eq:max_S_q_in_Fpm}) in two main respects. First,
problem~(\ref{eq:max_S_q_in_Fpm}) is posed on the non-convex feasible
set \(\mathcal{F}_\pm\), defined in Eq.~(\ref{eq:F_pm}), whereas
problem~(\ref{eq:max_J_in_Ppm}) is posed on the convex feasible set
\(\mathcal{P}_\pm\), defined in Eq.~(\ref{eq:P_pm}). Second, unlike
problem~(\ref{eq:max_S_q_in_Fpm}), the present formulation involves a
single probability density, \(p\), and does not require the introduction
of an associated escort probability. Thus, the escort reparametrization is not merely notational: it is the step that converts the original constrained Tsallis problem into a tractable concave program. The same argument applies, as the symmetric special case, to the standard escort mean--variance formulation.
\end{remark}
\begin{remark}\label{rem:fstar_vs_pstar}
Whenever problem (\ref{eq:max_J_in_Ppm}) admits a maximizer $p=p^{*}$, it follows from 
Eq. (\ref{eq:J_functional_1}) that the corresponding maximizer $f^{*}(x)$ of problem (\ref{eq:max_S_q_in_Fpm}) is given by
\begin{equation}\label{eq:fstar_vs_pstar}
    f^*(x)
    =
    \frac{[p^*(x)]^{1/q}}{\int [p^*(u)]^{1/q}\,du}.
\end{equation}
The corresponding value of the Tsallis non-extensive entropy is given by
\begin{equation}
    S_q[f^*]
    = \frac{1-\int [f^*(x)]^q\,dx}{q-1}.
\end{equation}

\end{remark}

\begin{remark}
Once the problem is written in the escort variable, the first-order
conditions become local in the optimization density and the prescribed
quantities enter linearly. The stationary family can therefore be derived
directly from a single Lagrangian. It follows that problem~(\ref{eq:max_J_in_Ppm})--(\ref{eq:P_pm}) can be associated
with the Lagrangian
\begin{equation}\label{eq:lagrangian_p}
\begin{aligned}
\mathcal L[p]
=&
\int p(x)^{1/q}\,dx
-\lambda_0\left(\int p(x)\,dx-1\right)
\\
&-\lambda_1\int (x-\mu_q)p(x)\,dx
\\
&-\lambda_-
\left(
\int [(x-\mu_q)_-]^2p(x)\,dx
-\sigma_{q,-}^2
\right)
\\
&-\lambda_+
\left(
\int [(x-\mu_q)_+]^2p(x)\,dx
-\sigma_{q,+}^2
\right).
\end{aligned}
\end{equation}
Stationarity of \(\mathcal L\) with respect to \(p\) gives the
stationary candidate $\hat{p}$, whereas stationarity with respect to the Lagrange
multipliers \(\{\lambda_0,\lambda_1,\lambda_-,\lambda_+\}\) imposes the
normalization, centering, and semi-variance constraints.\end{remark}
\begin{proposition}\label{prop:stationary_escort_density}
The stationarity condition with respect to \(p\)  for the Lagrangian in
Eq.~(\ref{eq:lagrangian_p}) yields the following stationary family
\begin{equation}\label{eq:p_hat}
 \hat{p}(x)=\frac{1}{\mathcal{Z}_p}
 \left[
 1+\beta (x-\mu_q)
 +\gamma_- \bigl[(x-\mu_q)_-\bigr]^2
 +\gamma_+ \bigl[(x-\mu_q)_+\bigr]^2
 \right]^{-\frac{q}{q-1}},
\end{equation}
almost everywhere, with 
\begin{equation}\label{eq:Zp}
 \mathcal{Z}_p
 =
 \int
 \left[
 1+\beta (x-\mu_q)
 +\gamma_- \bigl[(x-\mu_q)_-\bigr]^2
 +\gamma_+ \bigl[(x-\mu_q)_+\bigr]^2
 \right]^{-\frac{q}{q-1}}\,dx .
\end{equation}
\end{proposition}
\begin{proof}
 Let \(\hat p \in \mathcal{P}_\pm\) be a stationary solution of the Lagrangian
\(\mathcal L\) in Eq.~(\ref{eq:lagrangian_p}). 
For an arbitrary test function \(h\) with \(|h|\le C\hat p\) for some
constant \(C>0\), so that the first variation is well defined, set
\[
p_\varepsilon=\hat p+\varepsilon h .
\]
We assume \(|\varepsilon|\) small enough so that \(p_\varepsilon \geq 0\) almost everywhere  
and all the integrals appearing in \(\mathcal L[p_\varepsilon]\) are finite.
Since the equality constraints are incorporated through Lagrange
multipliers, the perturbation \(h\) is not required to satisfy the affine
normalization, centering, or semi-variance constraints. Stationarity of
\(\mathcal L\) with respect to \(p\) requires
\[
\left.
\frac{d}{d\varepsilon}
\mathcal L[\hat p+\varepsilon h]
\right|_{\varepsilon=0}
=0
\qquad \text{for all admissible } h .
\]
Using Eq.~(\ref{eq:lagrangian_p}), this gives
\begin{equation}
\int
\left\{
 \frac{1}{q}\hat p(x)^{\frac{1}{q}-1}
 -\lambda_0
 -\lambda_1(x-\mu_q)
 -\lambda_+[(x-\mu_q)_+]^2
 -\lambda_-[(x-\mu_q)_-]^2
\right\}
h(x)\,dx=0 .
\end{equation}
Since \(h\) is arbitrary, the term in braces must vanish almost
everywhere. Hence, almost everywhere,
\begin{equation}\label{eq:kkt_p}
    \frac{1}{q}\hat p(x)^{\frac{1}{q}-1}
    =
    \lambda_0
    +
    \lambda_1(x-\mu_q)
    +
    \lambda_+[(x-\mu_q)_+]^2
    +
    \lambda_-[(x-\mu_q)_-]^2 .
\end{equation}
Equivalently,
\begin{equation}\label{eq:p_kkt}
    \hat p(x)
    =
    \left[
    q\left(
    \lambda_0
    +
    \lambda_1(x-\mu_q)
    +
    \lambda_+[(x-\mu_q)_+]^2
    +
    \lambda_-[(x-\mu_q)_-]^2
    \right)
    \right]^{-\frac{q}{q-1}} .
\end{equation}
Factoring out the constant term \(q\lambda_0\), we obtain
\begin{equation}
\hat p(x)
    =
    C \left[
    1
    +
    \frac{\lambda_1}{\lambda_0}(x-\mu_q)
    +
    \frac{\lambda_+}{\lambda_0}[(x-\mu_q)_+]^2
    +
    \frac{\lambda_-}{\lambda_0}[(x-\mu_q)_-]^2
    \right]^{-\frac{q}{q-1}},
\end{equation}
where
\[
C=(q\lambda_0)^{-\frac{q}{q-1}}.
\]
Introducing the reparametrization
\[
\beta=\frac{\lambda_1}{\lambda_0},\qquad
\gamma_+=\frac{\lambda_+}{\lambda_0},\qquad
\gamma_-=\frac{\lambda_-}{\lambda_0},
\]
and absorbing the multiplicative constant into the normalization constant
\(\mathcal Z_p\), we obtain Eq.~(\ref{eq:p_hat}), with
\(\mathcal Z_p\) given by Eq.~(\ref{eq:Zp}).
\end{proof}

\begin{remark}\label{remark:f_hat}
 Given $\hat{p}(x)$ in Eq. (\ref{eq:p_hat}), the original density $\hat{f}(x)$ is recovered by Eq. (\ref{eq:J_functional_1})  as
 \begin{equation}\label{eq:f_hat}
 \hat{f}(x)= \frac{1}{\mathcal{Z}_f}
 \left[
 1+\beta (x-\mu_q)
 +\gamma_- \bigl[(x-\mu_q)_-\bigr]^2
 +\gamma_+ \bigl[(x-\mu_q)_+\bigr]^2
 \right]^{-\frac{1}{q-1}},
 \end{equation}
almost everywhere, where
\begin{equation}\label{eq:Z_f}
\mathcal{Z}_f= \int \left[
 1+\beta (x-\mu_q)
 +\gamma_- \bigl[(x-\mu_q)_-\bigr]^2
 +\gamma_+ \bigl[(x-\mu_q)_+\bigr]^2
 \right]^{-\frac{1}{q-1}} \, dx.
\end{equation}
Equation~\eqref{eq:f_hat} is the density in the original
Tsallis problem, rather than an auxiliary fit on the escort scale. It is a
single globally normalized law: the common entropic index $q$ determines the
power-law exponent on both sides, while $\gamma_-$ and $\gamma_+$ allow the
directional dispersions to differ and $\beta$ enforces the global
escort-centering constraint. Thus, asymmetry is generated by the joint
constraints, not by joining two separately normalized half-distributions.
Moreover, the two pieces of the kernel have the same value and first
derivative at $x=\mu_q$, so that $\hat f$ is continuously differentiable at
its escort location.
\end{remark}
\begin{remark}
The stationary point \(\hat{p}\) in Eq.~(\ref{eq:p_hat}) is well
defined provided that the kernel 
\begin{equation}\label{eq:H_positive}
H(x)=
1+\beta (x-\mu_q)
+\gamma_- \bigl[(x-\mu_q)_-\bigr]^2
+\gamma_+ \bigl[(x-\mu_q)_+\bigr]^2
>0
\end{equation}
for all \(x\in\mathbb{R}\). It follows that $H(x)>0$ for all \(x\in\mathbb{R}\) if and only if $-2\sqrt{\gamma_+}<\beta<2\sqrt{\gamma_-}.$
 In this case
\(\hat{p}(x)>0\) for all \(x\in\mathbb{R}\).
\end{remark}
\begin{remark}
In order to guarantee the existence of the escort mean and escort semi-variances it
is enough to control the second semi-moments
\[
\int [(x-\mu_q)_\pm]^2 \hat p(x)\,dx .
\]
Under the condition \(\gamma_+>0\) and \(\gamma_->0\), the density
\(\hat p\) has power-law tails 
\[
\hat{p}(x)\sim |x|^{-\frac{2q}{q-1}}
\qquad \text{as } |x|\to +\infty.
\]
Hence
\[
[(x-\mu_q)_\pm]^2 \hat{p}(x)
\sim
|x|^{2-\frac{2q}{q-1}}
\qquad \text{as } |x|\to +\infty.
\]
Therefore the semi-variance moments are finite whenever
\[
2-\frac{2q}{q-1}<-1,
\]
which, for \(q>1\), is equivalent to \(1<q<3\) --- the same admissible
range as in the symmetric case recalled in
Remark~\ref{rem:symmetric_qrange}.
\end{remark}

The calculation above identifies a stationary family, but
stationarity alone would not establish the maximum-entropy interpretation.
The next two propositions use precisely the concavity and affine structure
created by the escort reformulation to upgrade every feasible stationary
member to the unique global solution.

\begin{proposition}\label{prop:global_optimizer}\label{prop:asym}
Let $q>1$. Any feasible stationary solution $\hat{p}\in \mathcal{P}_\pm$ of the form~(\ref{eq:p_hat}) with $\hat{p}>0$ a.e. and
\(J[\hat p]<+\infty\) is a
global maximizer of problem~(\ref{eq:max_J_in_Ppm}).\end{proposition}
\begin{proof}
Let \(\hat p\in\mathcal P_\pm\) be a feasible stationary solution of
the form~(\ref{eq:p_hat}) with $\hat{p}>0$ a.e., and $J[\hat{p}]<+\infty$, and let \(p\in\mathcal P_\pm\) be any admissible density with $J[p]<+\infty$.
Since \(q>1\), the function \(\phi(u)= u^{1/q}\) is concave on
\([0,+\infty)\). Therefore, for every \(u\ge 0\) and every \(v>0\), the tangent inequality holds:
\[
\phi(u)\leq \phi(v)+\phi'(v)(u-v),
\]
or
\begin{equation}\label{eq:tangent_ineq_J}
u^{1/q}
\le
v^{1/q}
+
\frac{1}{q}v^{\frac{1}{q}-1}(u-v).
\end{equation}
Applying \eqref{eq:tangent_ineq_J} pointwise with
\(u=p(x)\geq 0\)  and \(v=\hat p(x)>0\), and integrating, gives
\[
\int p(x)^{1/q}\,dx
\le \int \hat{p}(x)^{1/q}\,dx+
\int
\frac{1}{q}\hat p(x)^{\frac{1}{q}-1}
\bigl(p(x)-\hat p(x)\bigr)\,dx 
\]
or
\begin{equation}\label{eq:J_tangent_bound}
J[p]
\le J[\hat p]+
\int
\frac{1}{q}\hat p(x)^{\frac{1}{q}-1}
\bigl(p(x)-\hat p(x)\bigr)\,dx .
\end{equation}
Since \(\hat p\) is a stationary solution of the Lagrangian,
Eq.~(\ref{eq:kkt_p}) gives
\begin{equation}
\frac{1}{q}\hat p(x)^{\frac{1}{q}-1}
=
\lambda_0
+\lambda_1(x-\mu_q)
+\lambda_+[(x-\mu_q)_+]^2
+\lambda_-[(x-\mu_q)_-]^2 .
\end{equation}
Hence Eq. (\ref{eq:J_tangent_bound}) becomes
\begin{equation}\label{eq:J_optimum}
J[p]-J[\hat{p}]\leq \int  \{\lambda_0
+\lambda_1(x-\mu_q)
+\lambda_+[(x-\mu_q)_+]^2
+\lambda_-[(x-\mu_q)_-]^2 \}(p(x)-\hat{p}(x))\,dx
\end{equation}
But, since both \(p\) and \(\hat p\) belong to
\(\mathcal P_\pm\), their difference satisfies
\begin{equation}
\int (p-\hat p)\,dx=0,
\end{equation}
\begin{equation}
\int (x-\mu_q)(p-\hat p)\,dx=0,
\end{equation}
and
\begin{equation}
\int [(x-\mu_q)_\pm]^2(p-\hat p)\,dx=0 .
\end{equation}
Hence the right-hand side of \eqref{eq:J_optimum} is zero.
Therefore
\begin{equation}\label{eq:J_max}
J[p]\le J[\hat p]
\qquad
\forall p\in\mathcal P_\pm,
\end{equation}
thus \(\hat p\) is a global maximizer $p^*$ of
problem~(\ref{eq:max_J_in_Ppm}).
\end{proof}
\begin{proposition}\label{prop:unique_optimizer}
Problem~(\ref{eq:max_J_in_Ppm}) admits at most one global maximizer.
\end{proposition}

\begin{proof}
Suppose, by contradiction, that there exist two distinct global
maximizers \(p_1^*,p_2^*\in\mathcal P_\pm\), differing on a set of positive
Lebesgue measure. Then
\[
J[p_1^*]=J[p_2^*]=J^*,
\]
where \(J^*\) denotes the maximum value of \(J[p]\) over
\(\mathcal P_\pm\).

For \(0<t<1\), define
\[
\tilde{p}=tp_1^*+(1-t)p_2^* .
\]
Since \(\mathcal P_\pm\) is convex, we have \(\tilde{p}\in\mathcal P_\pm\).
Moreover, since \(q>1\), the function \(\phi(u)= u^{1/q}\) is strictly
concave on \([0,+\infty)\). Hence \(J[p]=\int p(x)^{1/q}\,dx\) is strictly concave, and we obtain
\[
J[tp_1^*+(1-t)p_2^*]
>
tJ[p_1^*]+(1-t)J[p_2^*]
=
tJ^*+(1-t)J^*
=
J^*.
\]
This contradicts the definition of \(J^*\) as the maximum value of
\(J[p]\) over \(\mathcal P_\pm\). Therefore the global maximizer is unique.
\end{proof}
\begin{remark}\label{rem:inverse-moment-problem}
Proposition~\ref{prop:stationary_escort_density} identifies the functional
form of a stationary density $\hat{p}$, whereas
Propositions~\ref{prop:global_optimizer}--\ref{prop:unique_optimizer}
establish global optimality and uniqueness only after feasibility has been
verified. Hence, for prescribed values of $q$, $\mu_q$,
$\sigma^2_{q,-}$, and $\sigma^2_{q,+}$, one must still determine whether
parameters $(\beta,\gamma_-,\gamma_+)$ can be chosen so that the density in
Eq.~\eqref{eq:p_hat} satisfies all the constraints in
Eq.~\eqref{eq:P_pm}. This is the inverse moment problem associated with the
escort mean and escort semi-variances.  In other words,
Propositions~\ref{prop:stationary_escort_density}--
\ref{prop:unique_optimizer} solve the variational problem conditionally on
feasibility; the remaining question is whether the prescribed directional
moments are compatible with a member of the stationary family.
\footnote{The prescribed value of $\mu_q$ creates no additional compatibility
condition because the support is the whole real line:  setting
$y=x-\mu_q$ gives a bijection between densities centered at $\mu_q$ and
densities $g$ satisfying
\[
    \int y\, g(y)\,dy=0,
    \qquad
    \int y_\pm^2\,g(y)\,dy=\sigma^2_{q,\pm}.
\]}.

\end{remark}

\begin{proposition}\label{prop:inverse-moment-existence}
Let $\mu_q\in\mathbb{R}$, $1<q<3$, and
$\sigma^2_{q,-}>0$ and $\sigma^2_{q,+}>0$, and set
\[
    \rho=\frac{\sigma^2_{q,+}}{\sigma^2_{q,-}}.
\]
Consider the stationary escort density $\hat{p} $ in
Eq.~\eqref{eq:p_hat}, with kernel $H$ given by
Eq.~\eqref{eq:H_positive}.

A non-degenerate triple
\[
    \beta\in\mathbb{R},
    \qquad \gamma_->0,
    \qquad \gamma_+>0,
    \qquad -2\sqrt{\gamma_+}<\beta<2\sqrt{\gamma_-},
\]
for which $\hat{p} \in\mathcal{P}_{\pm}$ satisfies the prescribed
centering and escort semi-variance constraints exists precisely as follows:
\begin{enumerate}
    \item if $\frac{3}{2}\le q<3$, it exists for every pair
    $\sigma^2_{q,-}>0$ and $\sigma^2_{q,+}>0$;

    \item if $1<q<\frac{3}{2}$, it exists if and only if
    \[
        \frac{3-2q}{3-q}
        <\rho
        <\frac{3-q}{3-2q}.
    \]
\end{enumerate}

In the second case, the endpoint ratios yield unique degenerate boundary
solutions. More precisely,
\[
    \rho=\frac{3-2q}{3-q}
    \quad\Longrightarrow\quad
    \gamma_-=0,\ \gamma_+>0,
\]
whereas
\[
    \rho=\frac{3-q}{3-2q}
    \quad\Longrightarrow\quad
    \gamma_+=0,\ \gamma_->0.
\]
Outside the corresponding closed interval, no stationary density of the
form~\eqref{eq:p_hat}, with a globally positive kernel of the form
\eqref{eq:H_positive} and $\gamma_-\geq0$ and $\gamma_+\geq0$, satisfies all
three prescribed escort-moment constraints.

\end{proposition}

\begin{proof}
The proof is given in Section~S1 of the Supplementary Material, where the statement appears as Corollary~S1. 

\begin{remark}
When \(q=1\), from Eq. \ref{eq:p_escort}, the escort density $p(x)$ coincides with the original density $f(x)=f_1(x)$ and
Tsallis entropy reduces to Shannon entropy. Therefore, the asymmetric
maximum-entropy problem becomes
\[
\max_{f_1\in\mathcal{F}_{\pm}^{(1)}}
S[f_1],
\qquad
S[f_1]=-\int f_1(x)\ln f_1(x)\,dx,
\]
where
\begin{equation}\label{eq:F_pm^1}
\begin{aligned}
\mathcal{F}^{(1)}_\pm=\Bigg\{ f_1(x)\geq 0 \ \text{a.e.} \ :\ 
&\int f_1(x)\,dx=1, \\
&\int(x-\mu)f_1(x)\,dx=0, \\
&\int[(x-\mu)_-]^2 f_1(x)\,dx=\sigma^2_{-}, \\
&\int[(x-\mu)_+]^2 f_1(x)\,dx=\sigma^2_{+}
\Bigg\}.
\end{aligned}
\end{equation}The stationarity conditions associated with this problem yield the
piecewise-quadratic exponential family
\begin{equation}\label{eq:f_1}
f_1(x)
=
\frac{1}{Z_{1}}
\exp\left\{
-\beta^{'}(x-\mu)
-\gamma_{-}^{'}\bigl[(x-\mu)_{-}\bigr]^2
-\gamma_{+}^{'}\bigl[(x-\mu)_{+}\bigr]^2
\right\},
\end{equation}
where
\begin{equation}\label{eq:Z_1}
Z_{1}
=
\int
\exp\left\{
-\beta^{'}(x-\mu)
-\gamma_{-}^{'}\bigl[(x-\mu)_{-}\bigr]^2
-\gamma_{+}^{'}\bigl[(x-\mu)_{+}\bigr]^2
\right\}dx.
\end{equation}
In particular, \(\gamma_{-}^{'}>0\) and \(\gamma_{+}^{'}>0\) guarantee
integrability on the whole real line. Since Shannon entropy is strictly
concave and the constraints are affine, any feasible stationary density
of the above form is the unique global maximizer of the corresponding
Shannon maximum-entropy problem.
\end{remark}
\begin{remark}
The asymmetric Shannon density in Eq. (\ref{eq:f_1}) and (\ref{eq:Z_1})  can be recovered as a
non-degenerate limit of the asymmetric \(q\)-Gaussian family in
Eq. (\ref{eq:f_hat}) and (\ref{eq:Z_f}). In fact, with
\[
\begin{aligned}
\beta
   &= (q-1)\beta^{'}+o(q-1),\\
\gamma_{-}
   &= (q-1)\gamma_{-}^{'}+o(q-1),\\
\gamma_{+}
   &= (q-1)\gamma_{+}^{'}+o(q-1),
\end{aligned}
\]
with \(\mu_q=\mu\), the kernel $H(x)$ in Eq. (\ref{eq:H_positive}) becomes
\[
H(x)
=
1+(q-1)
\left\{
\beta^{'}(x-\mu)
+\gamma_{-}^{'}[(x-\mu)_{-}]^{2}
+\gamma_{+}^{'}[(x-\mu)_{+}]^{2}
\right\}
+o(q-1).
\]
Consequently,
\[
\begin{aligned}
\lim_{q\to1^{+}}
H(x)^{-\frac{1}{q-1}}
&=
\exp\left\{
-\beta^{'}(x-\mu)
-\gamma_{-}^{'}[(x-\mu)_{-}]^{2}
-\gamma_{+}^{'}[(x-\mu)_{+}]^{2}
\right\}
\end{aligned}
\]
by the standard exponential limit
\[
\lim_{\varepsilon\to0}
\bigl(1+\varepsilon a\bigr)^{-1/\varepsilon}
=e^{-a}.
\]
Hence, it
follows that
\[
f_q(x)\longrightarrow f_{1}(x),
\]
where \(f_{1}\) is the asymmetric Shannon maximum-entropy density given
in Eq. (\ref{eq:f_1}) and (\ref{eq:Z_1}).

\end{remark}
\end{proof}

\begin{remark}\label{rem:forward-parametrization}
Proposition~\ref{prop:inverse-moment-existence} gives the sharp answer to the
inverse moment problem when $q$ and the escort constraints are prescribed.
Its role is to complete the theoretical characterization of the
maximum-entropy construction, rather than to introduce an additional stage
in statistical estimation. In the statistical implementation, the
construction is used in the opposite direction. The parameter vector
\[
    \theta=(q,\mu_q,\beta,\gamma_-,\gamma_+)
\]
is estimated from the data, subject to the positivity and integrability
conditions and to the escort-centering condition
\[
    \int_{\mathbb{R}}(x-\mu_q)p_\theta(x)\,dx=0.
\]
The fitted escort density then induces its own escort semi-variances,
\[
    \sigma^2_{q,\pm}(\theta)
    =\int_{\mathbb{R}}\bigl[(x-\mu_q)_{\pm}\bigr]^2
      p_\theta(x)\,dx.
\]
Thus, likelihood estimation does not require the numerical inversion of the
moment map. In particular, when $1<q<3/2$, $\gamma_->0$, and
$\gamma_+>0$, the ratio of the induced escort semi-variances automatically
lies in the open compatibility interval of
Proposition~\ref{prop:inverse-moment-existence}. For these induced
constraints, Propositions~\ref{prop:global_optimizer}--
\ref{prop:unique_optimizer} show that $p_\theta$ is the unique maximizer of
the escort problem, while Eq.~\eqref{eq:J_functional_1} yields the associated
Tsallis maximum-entropy density $f_\theta$, which coincides with the density
in Eq.~\eqref{eq:f_hat}.
\end{remark}

\section{Empirical analysis}\label{sec:empirical}

We assess the asymmetric Tsallis model introduced in
Section~\ref{sec:asym_construction} on synthetic data, to study how it
responds to controlled departures from symmetry, and on real
financial returns, to evaluate its descriptive fit on empirical series.

Throughout this section the density \(f_\theta\), equivalently the asymmetric
$q$-Gaussian of Section~\ref{sec:asym_construction}, is the asymmetric
Tsallis density, with parameter vector
\[
\theta=(q,\mu_q,\beta,\gamma_-,\gamma_+).
\]
Its symmetric restriction, obtained by setting \(\beta=0\) and
\(\gamma_-=\gamma_+\), is the symmetric Tsallis benchmark, and the two models
are fitted by maximum likelihood on the same samples.

\subsection{Data}\label{subsec:data}

\paragraph{Synthetic data.}
Synthetic samples are generated from the Azzalini skew-\(t\) distribution
\citep{azzalini2003distributions}, with location \(\xi=0\) and scale
\(\omega=1\), over the grid
\[
\alpha\in\{-2,-1,0,1,2\},
\qquad
\nu\in\{3.1,3.5,3.8,4.1,5,10,15,30\}.
\]
The parameter \(\alpha\) controls asymmetry, with \(\alpha=0\) giving the
symmetric case, whereas \(\nu\) controls tail thickness in the
data-generating distribution, so that the grid ranges from nearly symmetric
to skewed samples and from heavy-tailed to lighter-tailed ones.

The asymmetric Tsallis model is not nested in the Azzalini family, so the
experiment is an approximation exercise, which tests how well the
entropy-based family adapts to controlled skewness and tail thickness
generated outside the Tsallis class. Random variates are generated using the
stochastic representation in \citet{AzzaliniSN}.

\paragraph{Real-market data.}
The data are daily closing prices, adjusted for dividends and stock splits,
downloaded from LSEG Data \& Analytics (formerly Refinitiv--Datastream,
see \citealp{LSEG}); the dataset is publicly available in \citet{Cesarone2025DS}.
We consider non-overlapping logarithmic returns
\(r_t=\log(P_{t+\Delta}/P_t)\), at three
horizons, \(\Delta=1\), \(\Delta=7\) and \(\Delta=20\), where \(\Delta\)
counts \emph{trading} days throughout and the horizon-\(\Delta\) series is
obtained by subsampling the daily closing-price series every \(\Delta\) rows,
so that no observation is used twice. In calendar terms \(\Delta=7\)
corresponds to about a week and a half and \(\Delta=20\) to about one month.
The sample includes nine markets, six equity indices (BIST, FTSE, NASDAQ,
DJIA, S\&P~500, and BOVESPA), one commodity (BRENT), and two cryptocurrencies
(Bitcoin and Ethereum).

The sample periods differ across series, depending on data availability.
BRENT starts in 1990, the equity indices at the end of 2009, Bitcoin in 2014
and Ethereum in 2018, and all series end in late 2025 or early 2026. The first
and last dates, together with the number of observations at each frequency,
are reported in Table~\ref{tab:descriptive_markets}.

Before estimation, returns are standardised to zero mean and unit variance,
so that the analysis is focused on the shape of the unconditional return
distribution: tail thickness, summarised by the raw (non-excess) kurtosis,
whose Gaussian reference value is \(3\), and the balance between downside and
upside dispersion, measured by the semi-variances
\(\sigma^2_{\pm}=\mathbb E[(r-\mathbb E[r])^2\,\mathbf 1_{r\gtrless\mathbb E[r]}]\)
reported in the same table.

\subsection{Parameter estimation via constrained maximum likelihood}
\label{subsec:estimation}

The observed returns are treated as an i.i.d.\ sample from \(f_\theta\). For
\(z=x-\mu_q\), let
\[
H_\theta(x)
=
1+\beta z+\gamma_-(z)_-^2+\gamma_+(z)_+^2,
\]
denote the kernel of Eq.~\eqref{eq:H_positive}, so that in the original
density scale
\begin{equation}\label{eq:f_density_q}
    f_\theta(x)
    =
    \frac{H_\theta(x)^{-\frac{1}{q-1}}}{Z_f(\theta)},
    \qquad
    Z_f(\theta)
    =
    \int_{\mathbb R}H_\theta(u)^{-\frac{1}{q-1}}\,du .
\end{equation}
The entropic index is left free over the whole admissible range \(1<q<3\) in
which \(Z_f(\theta)\) is finite, as specified in
Eq.~\eqref{eq:theta_admissible} below. The stricter restriction
\(q<5/3\), which would in addition guarantee a finite ordinary variance
(Remark~\ref{rem:symmetric_qrange}), is not imposed.

In the numerical implementation, positivity of \(H_\theta\) on the whole of
\(\mathbb R\) is enforced through the sufficient conditions
\[
\beta^2<4\gamma_-,
\qquad
\beta^2<4\gamma_+ .
\]
On each half-line the kernel is the quadratic
\(H_\theta=1+\beta z+\gamma_\pm z^2\) in \(z=x-\mu_q\): the inequalities make
both discriminants negative and \(H_\theta=1\) at \(z=0\), so both pieces are
positive on all of \(\mathbb R\). Positivity alone would constrain only the
piece whose vertex lies inside its own half-line; the symmetric pair is kept
because it is smooth in \(\theta\), hence suited to the interior-point solver,
and because the closed-form partition function, moments and CDF require
\(4\gamma_\pm-\beta^2>0\) on both branches. Normalization is carried by
\(Z_f(\theta)\), continuity at \(x=\mu_q\) by the kernel itself.

Let $\varepsilon=10^{-8}$, $q_{\min}=1.001$, $q_{\max}=2.999$, and
$\gamma_0=10^{-10}$, and define the (closed) admissible set
\begin{equation}\label{eq:theta_admissible}
\Theta_\varepsilon
=
\left\{
\theta=(q,\mu_q,\beta,\gamma_-,\gamma_+)\in\mathbb{R}^5
\;\middle|\;
\begin{aligned}
&q_{\min}\le q\le q_{\max},\quad
\gamma_-\ge\gamma_0,\ \gamma_+\ge\gamma_0,\\
&\beta^2\le 4(1-\varepsilon)\gamma_-,\quad
\beta^2\le 4(1-\varepsilon)\gamma_+,\\
&\int_{\mathbb R}(x-\mu_q)\,H_\theta(x)^{-\frac{q}{q-1}}\,dx=0
\end{aligned}
\right\} .
\end{equation}

The parameter vector is estimated by constrained maximum likelihood over
\(\Theta_\varepsilon\),
\begin{equation}\label{eq:loglik_f_q}
    \widehat\theta
    \in
    \arg\max_{\theta\in\Theta_\varepsilon}\ \ell_f(\theta),
    \qquad
    \ell_f(\theta)
    =
    -\frac{1}{q-1}\sum_{i=1}^N\log H_\theta(x_i)
    -
    N\log Z_f(\theta),
\end{equation}
where the escort-centering constraint makes \(\mu_q\) the escort mean of the
fitted density.

The escort semi-variances are not estimated as free parameters, but recovered
a posteriori from the fitted escort density
\[
p_\theta(x)
=
\frac{H_\theta(x)^{-\frac{q}{q-1}}}{Z_p(\theta)},
\qquad
Z_p(\theta)
=
\int_{\mathbb R}H_\theta(u)^{-\frac{q}{q-1}}\,du,
\]
as
\begin{equation}\label{eq:escort_semivariances_fitted}
\widehat\sigma^2_{q,-}
=
\int_{\mathbb R}
(x-\widehat\mu_q)_-^2\,p_{\widehat\theta}(x)\,dx,
\qquad
\widehat\sigma^2_{q,+}
=
\int_{\mathbb R}
(x-\widehat\mu_q)_+^2\,p_{\widehat\theta}(x)\,dx .
\end{equation}
For comparison with the sample moments, we also report the \emph{ordinary}
semi-variances
\[
\sigma_{f,-}^2
=
\int_{\mathbb R}(x-\mu_f)_-^2\,f_\theta(x)\,dx,
\qquad
\sigma_{f,+}^2
=
\int_{\mathbb R}(x-\mu_f)_+^2\,f_\theta(x)\,dx,
\]
where \(\mu_f=\mathbb E_f[X]\). These are ordinary moments of the fitted
density, not the escort semi-variances appearing in the maximum-entropy
constraints.

This parametrization reverses the point of view of the maximum-entropy
derivation, where the escort semi-variances are fixed and
\((\beta,\gamma_-,\gamma_+)\) are Lagrange multipliers. By
Remark~\ref{rem:forward-parametrization}, each admissible fitted density
therefore retains its Tsallis maximum-entropy interpretation for the
constraints that it induces.

The likelihood is not a concave function of \(\theta\). This does not
contradict the variational result of Section~\ref{sec:asym_construction},
which concerns the entropy functional at fixed escort constraints. Estimation
is therefore carried out by constrained numerical optimization, in
\textsc{MATLAB} with \texttt{fmincon} (interior-point algorithm), fitting
first the two-parameter symmetric restriction and then using that estimate as
the starting point for the full asymmetric problem. The symmetric benchmark is
fitted by the built-in location--scale $t$ routine of the same environment.

For BOVESPA the fitted asymmetry is weakest at all three horizons, and the
likelihood is correspondingly flattest along the asymmetric directions. As a
numerical check we ran a multistart procedure from dispersed asymmetric
starting values, with \(\beta\) of both signs and \(\gamma_-\neq\gamma_+\). No
asymmetric start improved on the warm-start solution, which was recovered in
every run up to optimizer tolerance; the corresponding estimates are reported
in Section~S2 of the Supplementary Material (Table~S1). Over the
whole Monte Carlo design the solver returned a feasible stationary point in
all but one of the $400\,000$ asymmetric fits, and that replication was
discarded, with no effect on the reported medians.

The distribution-based goodness-of-fit statistics are computed from the
cumulative distribution function associated with
Eq.~\eqref{eq:f_density_q}. Under the positivity conditions above, this CDF
admits a real-valued closed form in terms of incomplete beta functions. Set
\[
a=\frac{1}{q-1},
\qquad
z_y=y-\mu_q,
\qquad
D_\pm=4\gamma_\pm-\beta^2>0.
\]
For \(s\in\{-,+\}\), define
\[
t_s(z)=\frac{2\gamma_s z+\beta}{\sqrt{D_s}},
\qquad
C_s=
\left(\frac{D_s}{4\gamma_s}\right)^{-a}
\frac{\sqrt{D_s}}{2\gamma_s}.
\]
Moreover, let
\[
B_a=B\!\left(\frac12,a-\frac12\right),
\]
and define
\[
G_a(t)
=
\frac12 B_a
+
\frac12\,\operatorname{sgn}(t)\,
B_{\frac{t^2}{1+t^2}}
\!\left(\frac12,a-\frac12\right),
\]
where \(B_x(\cdot,\cdot)\) denotes the incomplete beta function and
\(\operatorname{sgn}(0)=0\). Equivalently,
\[
G_a(t)=\int_{-\infty}^{t}(1+u^2)^{-a}\,du .
\]
Using
\[
1+\beta z+\gamma_s z^2
=
\frac{D_s}{4\gamma_s}\bigl(1+t_s(z)^2\bigr),
\]
the normalizing constant can be written as
\begin{equation}\label{eq:Zf_closed_form}
Z_f(\theta)
=
C_-G_a(t_-(0))
+
C_+\bigl[B_a-G_a(t_+(0))\bigr].
\end{equation}
Therefore,
\begin{equation}\label{eq:cdf_closed_form_empirical}
F_\theta(y)
=
\begin{cases}
\displaystyle
\frac{C_-G_a(t_-(z_y))}{Z_f(\theta)},
& y\le \mu_q,
\\[14pt]
\displaystyle
\frac{
C_-G_a(t_-(0))
+
C_+\bigl[G_a(t_+(z_y))-G_a(t_+(0))\bigr]
}{Z_f(\theta)},
& y>\mu_q .
\end{cases}
\end{equation}

\subsection{Goodness-of-fit and benchmark model}\label{subsec:gof}

Goodness-of-fit at the maximum-likelihood estimate \(\widehat{\theta}\) is
assessed through two empirical-distribution statistics, the
Kolmogorov--Smirnov statistic \(D_n\) \citep{an1933sulla} and the
Anderson--Darling statistic \(A_n^2\)
\citep{anderson1952asymptotic,anderson1954test,stephens1974edf}, both
computed from the CDF in Eq.~\eqref{eq:cdf_closed_form_empirical}. The first
measures the largest uniform distance between empirical and fitted
distribution functions, whereas the second gives more weight to tail
discrepancies. Since the parameters are estimated
on the same sample, the standard distribution-free critical values do not
apply \citep{lilliefors1967kolmogorov,lehmann2005testing}, and we therefore
use \(D_n\) and \(A_n^2\) as in-sample diagnostic measures.

We report the log-likelihood difference \(\Delta\ell\) and the AIC and BIC
differences, throughout taken as asymmetric minus symmetric Tsallis: under the
escort-centering constraint the two models have four and three free
parameters, so the criteria charge the gain from relaxing symmetry against a
single degree of freedom. The nesting supports a formal test,
\(\mathrm{LR}=2(\widehat\ell_{\mathrm{asym}}-\widehat\ell_{\mathrm{sym}})
\sim\chi^2_1\) \citep{wilks1938}, whose outcome is marked in the tables by
one, two and three stars on \(\Delta\ell\) at the \(5\%\), \(1\%\) and
\(0.1\%\) levels. The test uses the point estimates, whereas the
\(\Delta\ell\) displayed alongside is the bootstrap median.

Monte Carlo statistics are summarised across the \(R=10000\) replications,
whereas for the real-market series sampling variability is assessed by a
nonparametric bootstrap \citep{efron1993introduction,davison1997bootstrap}:
each of the \(B=999\) samples is drawn with replacement from the \(n\)
standardised returns, both models are re-estimated on it and the statistics
recomputed. In either case we report the median and the \([p_5,p_{95}]\)
interval, the value \(B=999\) following the usual convention of choosing
\(B\) so that \(\alpha(B+1)\) is an integer at the nominal levels of interest
\citep{davison1997bootstrap}.

Resampling individual observations preserves the unconditional marginal law
but, under volatility clustering \citep{Cont2001}, understates sampling
variability: the intervals are descriptive, not exact confidence intervals.
Block bootstrap schemes \citep{politis1994stationary,ledoit2008robust} are
left for future work.

\subsection{Synthetic experiments: results}\label{subsec:synthetic}

Each of the \(40\) scenarios of the grid in Section~\ref{subsec:data} is
replicated \(R=10000\) times with samples of size \(n=3000\).

Table~\ref{tab:mc_excerpt} crosses all five values of \(\alpha\) with
\(\nu\in\{3.1,5,30\}\). Each scenario is described by the entropic index \(\widehat q\), the
asymmetry parameter \(\widehat\beta\), the two ordinary semi-variances
\(\widehat\sigma^2_{f,-}\) and \(\widehat\sigma^2_{f,+}\), and the
log-likelihood difference \(\widetilde{\Delta\ell}\) with the KS and AD
statistics of each fit.
All entries are medians across replications.

Table~\ref{tab:mc_excerpt} exhibits four regularities. First, for every
\(\alpha\) the estimate \(\widehat q\) decreases monotonically in \(\nu\), from
\(1.31\)--\(1.49\) at \(\nu=3.1\) to \(1.01\)--\(1.06\) at \(\nu=30\), the
level at each \(\nu\) depending on \(\alpha\), and therefore tracks the tail
weight of the generating process. Second, the estimated asymmetry follows the
shape parameter in sign and in magnitude: \(\widehat\beta\) is \(0.000\) at
\(\alpha=0\), grows to \(0.254\) at \(\alpha=1\) and to \(0.350\) at
\(\alpha=2\) for \(\nu=3.1\), and reverses exactly under
\(\alpha\mapsto-\alpha\), as the design requires. The fitted semi-variances
carry the same information, the heavier side taking the larger component,
\(\widehat\sigma^2_{f,+}=1.115\) against
\(\widehat\sigma^2_{f,-}=0.439\) at \(\alpha=2\), \(\nu=3.1\), a ratio of about
\(2.5\), the two exchanging roles with the sign of \(\alpha\) and agreeing to
within one unit in the third decimal under symmetry, \(1.408\) against
\(1.407\) at \(\nu=3.1\). Both the imbalance and the overall level shrink as
\(\nu\) grows and the tails become lighter.

\begin{table}[!htbp]
  \centering
  \footnotesize
  \setlength{\tabcolsep}{4.5pt}
  \begin{tabular}{rr rrrr r rr rr}
    \toprule
    \multicolumn{2}{c}{True param.} & \multicolumn{4}{c}{Fitted asymmetric Tsallis} & & \multicolumn{2}{c}{KS} & \multicolumn{2}{c}{AD} \\
    \cmidrule(lr){1-2}\cmidrule(lr){3-6}\cmidrule(lr){8-9}\cmidrule(lr){10-11}
    $\alpha$ & $\nu$ & $\widehat{q}$ & $\widehat{\beta}$ & $\widehat{\sigma}^2_{f,-}$ & $\widehat{\sigma}^2_{f,+}$ & $\widetilde{\Delta\ell}$ & $\widetilde{D}_n^{\,\mathrm{asym}}$ & $\widetilde{D}_n^{\,\mathrm{sym}}$ & $\widetilde{A}_n^{2,\mathrm{asym}}$ & $\widetilde{A}_n^{2,\mathrm{sym}}$ \\
    \midrule
    \multirow{3}{*}{$2$} & 3.1 & 1.310 & $0.350$ & 0.439 & 1.115 & 154.8 & 0.0161 & 0.0606 & 0.829 & 30.613 \\
     & 5 & 1.197 & $0.193$ & 0.322 & 0.590 & 114.7 & 0.0138 & 0.0451 & 0.569 & 17.725 \\
     & 30 & 1.009 & $0.006$ & 0.228 & 0.307 & 57.7 & 0.0113 & 0.0321 & 0.343 & 7.738 \\
    \addlinespace
    \multirow{3}{*}{$1$} & 3.1 & 1.437 & $0.254$ & 0.709 & 1.307 & 53.8 & 0.0110 & 0.0347 & 0.327 & 10.749 \\
     & 5 & 1.299 & $0.134$ & 0.497 & 0.697 & 33.4 & 0.0105 & 0.0240 & 0.287 & 5.203 \\
     & 30 & 1.062 & $0.012$ & 0.348 & 0.388 & 7.9 & 0.0100 & 0.0170 & 0.252 & 1.278 \\
    \addlinespace
    \multirow{3}{*}{$0$} & 3.1 & 1.487 & $0.000$ & 1.408 & 1.407 & 0.2 & 0.0091 & 0.0098 & 0.208 & 0.283 \\
     & 5 & 1.331 & $0.000$ & 0.832 & 0.832 & 0.2 & 0.0093 & 0.0098 & 0.216 & 0.279 \\
     & 30 & 1.062 & $0.000$ & 0.536 & 0.536 & 0.2 & 0.0096 & 0.0105 & 0.233 & 0.295 \\
    \addlinespace
    \multirow{3}{*}{$-1$} & 3.1 & 1.437 & $-0.254$ & 1.309 & 0.710 & 53.8 & 0.0110 & 0.0347 & 0.327 & 10.745 \\
     & 5 & 1.299 & $-0.134$ & 0.697 & 0.497 & 33.3 & 0.0105 & 0.0241 & 0.287 & 5.197 \\
     & 30 & 1.062 & $-0.012$ & 0.388 & 0.348 & 7.9 & 0.0100 & 0.0171 & 0.253 & 1.284 \\
    \addlinespace
    \multirow{3}{*}{$-2$} & 3.1 & 1.310 & $-0.350$ & 1.115 & 0.440 & 154.7 & 0.0161 & 0.0606 & 0.818 & 30.632 \\
     & 5 & 1.198 & $-0.193$ & 0.589 & 0.322 & 114.8 & 0.0138 & 0.0450 & 0.562 & 17.721 \\
     & 30 & 1.008 & $-0.006$ & 0.307 & 0.228 & 57.8 & 0.0112 & 0.0323 & 0.341 & 7.753 \\
    \bottomrule
  \end{tabular}
  \caption{Monte Carlo results: representative scenarios, medians across replications.}
  \label{tab:mc_excerpt}
\end{table}

Third, the likelihood gain is essentially nil under symmetry and large
otherwise. At \(\alpha=0\) the median \(\widetilde{\Delta\ell}\) equals \(0.2\)
at every \(\nu\), whereas at \(|\alpha|=2\) it ranges from \(57.7\) to
\(154.8\), the values at \(\alpha\) and \(-\alpha\) agreeing to within the
last reported digit. Fourth, the two goodness-of-fit statistics separate the
models differently. Under symmetry the fits are indistinguishable on both,
whereas
under asymmetry the KS discrepancy of the benchmark grows up to three or four
times that of the asymmetric model, \(0.0606\) against \(0.0161\) at
\(\alpha=2\), \(\nu=3.1\), and the tail-weighted AD discrepancy up to more
than an order of magnitude, \(30.613\) against \(0.829\) in the same cell.
Since KS is governed by the centre and AD weights the tails, the gap between
the two ratios localises the advantage of the asymmetric specification in the
tail region, not in the bulk.

The same pattern holds over the whole grid, where the ordering of the two
goodness-of-fit statistics never reverses: the asymmetric escort semi-variance
constraints capture systematic skewness that the symmetric benchmark cannot
represent, on data generated outside the Tsallis class. The complete grid,
with the Monte Carlo intervals, the remaining structural parameters and the
exact moments of the generating process, is reported in
Appendix~\ref{app:mc_tables}.

The last rows of each block call for a caveat. At \(|\alpha|=2\) and
\(\nu=30\) the estimate \(\widehat q=1.009\) sits at the Gaussian limit of the
family and \(\widehat\beta=0.006\) looks negligible, yet the fitted
semi-variances still differ, \(0.228\) against \(0.307\). As \(q\to1\) the
kernel converges to a two-piece Gaussian,
\(H^{-1/(q-1)}\to\exp\{-\beta'z-\gamma'_\pm z^2\}\), so the asymmetry is
carried by the rescaled \(\beta/(q-1)\) and \(\gamma_\pm/(q-1)\), which remain
of order one. Therefore \(\widehat\beta\) is small in absolute terms, not in
the terms that shape the density.

\paragraph{Graphical evidence.}
Figure~\ref{fig:skewt_semilog_grid} shows four representative fits, for
\(\alpha\in\{0,2\}\) and \(\nu\in\{3.5,5\}\), each panel reporting the binned
empirical density of one simulated sample with the two fitted densities on a
semi-logarithmic scale. In the symmetric panels the two curves are
indistinguishable. For \(\alpha=2\)
the asymmetric fit follows the empirical shape in the bulk and in both tails,
with different decay on the two sides, whereas the benchmark misses the
directional tail behaviour.

\begin{figure}[!htbp]
  \centering
  \includegraphics[width=0.85\textwidth]{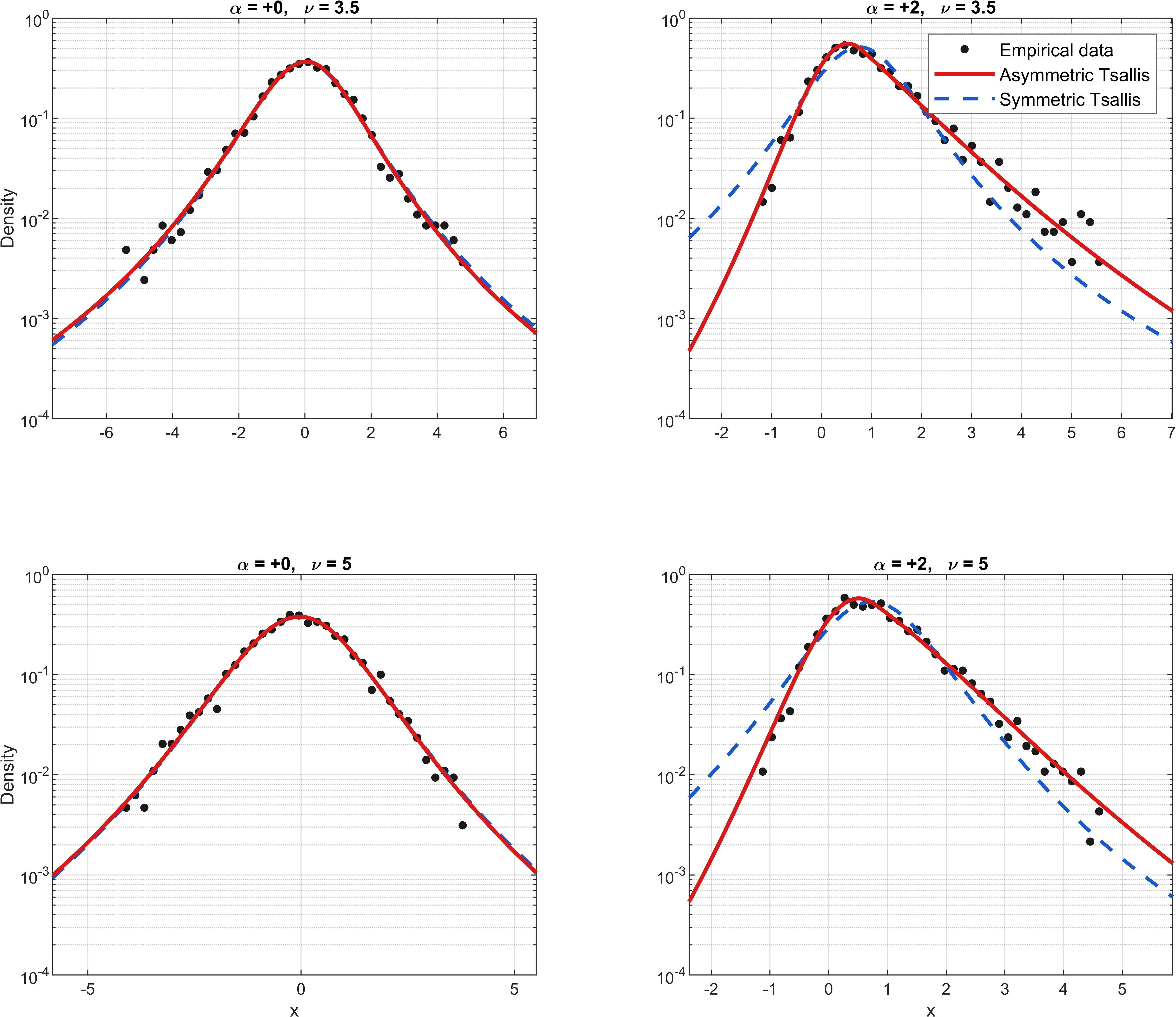}
  \caption{Four representative Monte Carlo fits. Rows: $\nu\in\{3.5,5\}$; columns: $\alpha\in\{0,2\}$.}
  \label{fig:skewt_semilog_grid}
\end{figure}

\subsection{Real-market experiments: results}\label{subsec:real}

Table~\ref{tab:descriptive_markets} reports the descriptive statistics of the
nine markets at the three frequencies. The series depart from Gaussianity in
two respects, asymmetry and tail weight. The skewness is negative for every
equity index and for the commodity at all three horizons, the only exception
being BIST at $\Delta=20$, and correspondingly $\sigma^2_->\sigma^2_+$ in all
those cases, whereas the two cryptocurrencies are close to symmetric
($|\text{Skew}|\le0.5$). The kurtosis is largest at the daily frequency for
most series, reaching $22.5$ for DJIA and $24.7$ for BRENT, and decreases under
aggregation, falling close to the Gaussian value $3$ at the lowest frequency
(NASDAQ $3.51$, BIST $3.58$, ETH $3.96$); it stays high only where a few
isolated extremes dominate (FTSE at $\Delta=20$, $22.4$, and DJIA at
$\Delta=7$, $22.6$). At $\Delta=1$ the cryptocurrencies are not the
heaviest-tailed series by sample kurtosis (BTC $9.78$ and ETH $10.90$, against
$22.5$ for DJIA), yet they carry the two largest fitted $\widehat q$ ($1.62$
and $1.54$). The two descriptors therefore respond to different features of the
distribution, a point taken up below.

\begin{table}[!htbp]
  \centering
  \adjustbox{max width=\textwidth}{\scriptsize
    \begin{tabular}{ccccrrrrrr}
      \toprule
      Market & First date & Last date & $\Delta$ & Obs.\ & $\bar r$ & $\sigma_{-}^{2}$ & $\sigma_{+}^{2}$ & Skew & Kurt \\
      \midrule
      \multirow{3}{*}{BIST} & \multirow{3}{*}{31/12/2009} & \multirow{3}{*}{31/12/2025} & 1 & 4017 & 0.076\% & 1.3328 & 1.0823 & -0.62 & 7.79 \\
       &  &  & 7 & 573 & 0.534\% & 9.8334 & 7.8150 & -0.49 & 4.37 \\
       &  &  & 20 & 200 & 1.527\% & 25.2391 & 26.0167 & 0.06 & 3.58 \\
      \midrule
      \multirow{3}{*}{FTSE} & \multirow{3}{*}{31/12/2009} & \multirow{3}{*}{31/12/2025} & 1 & 4041 & 0.015\% & 0.5255 & 0.4298 & -0.69 & 12.93 \\
       &  &  & 7 & 577 & 0.104\% & 3.3343 & 2.5798 & -0.70 & 7.20 \\
       &  &  & 20 & 202 & 0.301\% & 11.6568 & 5.7035 & -2.73 & 22.41 \\
      \midrule
      \multirow{3}{*}{NASDAQ} & \multirow{3}{*}{31/12/2009} & \multirow{3}{*}{31/12/2025} & 1 & 4024 & 0.065\% & 0.9398 & 0.7871 & -0.37 & 10.67 \\
       &  &  & 7 & 574 & 0.456\% & 5.2984 & 3.9060 & -0.64 & 4.88 \\
       &  &  & 20 & 201 & 1.305\% & 17.5698 & 11.0869 & -0.69 & 3.51 \\
      \midrule
      \multirow{3}{*}{DJIA} & \multirow{3}{*}{01/01/2010} & \multirow{3}{*}{31/12/2025} & 1 & 4048 & 0.038\% & 0.6003 & 0.4935 & -0.76 & 22.51 \\
       &  &  & 7 & 578 & 0.266\% & 4.0362 & 2.4209 & -2.14 & 22.59 \\
       &  &  & 20 & 202 & 0.755\% & 8.4666 & 5.1300 & -1.07 & 6.04 \\
      \midrule
      \multirow{3}{*}{S\&P500} & \multirow{3}{*}{31/12/2009} & \multirow{3}{*}{31/12/2025} & 1 & 4034 & 0.045\% & 0.6565 & 0.5379 & -0.62 & 16.59 \\
       &  &  & 7 & 576 & 0.317\% & 4.1410 & 2.4808 & -1.22 & 7.68 \\
       &  &  & 20 & 201 & 0.906\% & 11.4279 & 6.5890 & -0.97 & 5.13 \\
      \midrule
      \multirow{3}{*}{BOVESPA} & \multirow{3}{*}{30/12/2009} & \multirow{3}{*}{30/12/2025} & 1 & 3968 & 0.022\% & 1.1526 & 0.9833 & -0.79 & 15.08 \\
       &  &  & 7 & 566 & 0.147\% & 7.5130 & 5.6855 & -1.14 & 10.98 \\
       &  &  & 20 & 198 & 0.423\% & 20.5811 & 17.3381 & -0.39 & 4.06 \\
      \midrule
      \multirow{3}{*}{BRENT} & \multirow{3}{*}{02/01/1990} & \multirow{3}{*}{09/01/2026} & 1 & 9212 & 0.012\% & 2.8897 & 2.3736 & -1.18 & 24.71 \\
       &  &  & 7 & 1316 & 0.086\% & 19.4652 & 14.1694 & -1.03 & 10.11 \\
       &  &  & 20 & 460 & 0.241\% & 59.0907 & 44.9922 & -1.13 & 12.83 \\
      \midrule
      \multirow{3}{*}{BTC} & \multirow{3}{*}{17/07/2014} & \multirow{3}{*}{31/12/2025} & 1 & 4152 & 0.119\% & 6.5976 & 5.9036 & -0.47 & 9.78 \\
       &  &  & 7 & 593 & 0.835\% & 43.0459 & 44.8476 & 0.09 & 6.48 \\
       &  &  & 20 & 207 & 2.390\% & 157.1413 & 154.0742 & -0.04 & 5.50 \\
      \midrule
      \multirow{3}{*}{ETH} & \multirow{3}{*}{12/02/2018} & \multirow{3}{*}{31/12/2025} & 1 & 2877 & 0.043\% & 10.6772 & 9.5902 & -0.34 & 10.90 \\
       &  &  & 7 & 411 & 0.301\% & 76.9579 & 68.7903 & -0.36 & 5.40 \\
       &  &  & 20 & 143 & 0.890\% & 239.5047 & 239.7775 & -0.07 & 3.96 \\
      \bottomrule
    \end{tabular}}
  \caption{Descriptive statistics of the logarithmic returns.}
  \label{tab:descriptive_markets}
\end{table}

\begin{table}[!htbp]
  \centering
  \small
  \begin{tabular}{l r rr rr rr}
    \toprule
    Market & $\Delta\ell$ & $\widetilde{D}_n^{\,\mathrm{asym}}$ & $\widetilde{D}_n^{\,\mathrm{sym}}$ & $\widetilde{A}_n^{2,\mathrm{asym}}$ & $\widetilde{A}_n^{2,\mathrm{sym}}$ & $\widehat{\beta}$ & $\widehat{q}$ \\
    \midrule
    BIST    & $4.1^{**}$ & 0.0303 & 0.0380 & 0.440 & 1.101 & $-0.062$ & 1.19 \\
    FTSE    & $1.3$ & 0.0296 & 0.0306 & 0.491 & 0.657 & $-0.079$ & 1.35 \\
    NASDAQ  & $7.1^{***}$ & 0.0302 & 0.0449 & 0.413 & 1.543 & $-0.106$ & 1.24 \\
    DJIA    & $3.0^{*}$ & 0.0297 & 0.0317 & 0.558 & 0.944 & $-0.127$ & 1.36 \\
    S\&P500 & $11.0^{***}$ & 0.0328 & 0.0472 & 0.435 & 2.483 & $-0.273$ & 1.38 \\
    BOVESPA & $1.0$ & 0.0293 & 0.0297 & 0.430 & 0.546 & $-0.025$ & 1.19 \\
    BRENT   & $5.4^{**}$ & 0.0186 & 0.0259 & 0.393 & 1.273 & $-0.109$ & 1.36 \\
    BTC     & $0.7$ & 0.0322 & 0.0342 & 0.477 & 0.695 & $+0.067$ & 1.40 \\
    ETH     & $0.6$ & 0.0419 & 0.0471 & 0.815 & 1.052 & $+0.108$ & 1.51 \\
    \bottomrule
  \end{tabular}
  \caption{Real markets, weekly horizon ($\Delta=7$): bootstrap medians.}
  \label{tab:real_excerpt}
\end{table}

Table~\ref{tab:real_excerpt} summarises the results at the weekly horizon, one
row per market. The first column reports the log-likelihood difference
$\Delta\ell$, asymmetric minus symmetric Tsallis, as a bootstrap median, with
stars marking the likelihood-ratio test of the symmetric restriction; the
median KS and AD statistics of the two fitted models are then paired market by
market, followed by the asymmetry parameter $\widehat\beta$ and the entropic
index $\widehat q$. The same statistics for all three horizons, with bootstrap
intervals, fitted moments and the remaining structural parameters, are reported
in Appendix~\ref{app:real_tables}.

Since the symmetric benchmark is nested, $\Delta\ell\ge0$ by construction in
all $27$ market--horizon combinations: the informative quantity is its
magnitude, not its sign. The largest gains arise where the fitted
\emph{structural} asymmetry is strongest, which need not be where the
descriptive skewness is largest. S\&P500 at $\Delta=7$ attains
$\Delta\ell=11.0$, the highest value of the panel together with NASDAQ at
$\Delta=1$, and the largest $|\widehat\beta|$ of the horizon ($0.27$), whereas
DJIA, whose sample skewness at $\Delta=7$ is by far the largest in absolute
value ($-2.14$ against $-1.22$), reaches only $3.0$ with $|\widehat\beta|=0.13$.
Besides S\&P500 and NASDAQ ($7.1$), the symmetric restriction is rejected at
the $0.1\%$ level in five further configurations: NASDAQ at $\Delta=1$ ($11.0$)
and $\Delta=20$ ($9.9$), S\&P500 at $\Delta=20$ ($9.4$), DJIA at $\Delta=20$
($7.2$) and BRENT at $\Delta=1$ ($8.2$).

The improvement concentrates in the tail-sensitive statistic: at $\Delta=7$ the
asymmetric model reduces $\widetilde A_n^2$ from $2.48$ to $0.44$ for S\&P500
and from $1.54$ to $0.41$ for NASDAQ, and
$\widetilde A_n^{2,\mathrm{asym}}<\widetilde A_n^{2,\mathrm{sym}}$ in all $27$
combinations, whereas the KS comparison is mixed across the panel and its
bootstrap intervals overlap.
The penalised criteria order the evidence accordingly. $\Delta$AIC favours the
asymmetric model in $16$ of the $27$ combinations, at all three horizons for
FTSE, NASDAQ, DJIA and S\&P500; the stricter $\Delta$BIC, which charges the
extra parameter a $\log n$ penalty, in the $11$ where the gain clears it
($2\Delta\ell\gtrsim\log n$). The symmetric model remains preferred for
BOVESPA and the two cryptocurrencies at every horizon.

\paragraph{Structural reading.}
The structural parameters of Table~\ref{tab:real_excerpt} separate the two
groups of markets. Every equity index and the commodity carry a negative
$\widehat\beta$, S\&P500 $-0.273$, DJIA $-0.127$, BRENT $-0.109$, NASDAQ
$-0.106$, FTSE $-0.079$, BIST $-0.062$ and BOVESPA $-0.025$, that is a heavier
downside, together with an entropic index between $1.19$ and $1.38$. The two
cryptocurrencies instead show a small positive $\widehat\beta$ (ETH $+0.108$,
BTC $+0.067$), the largest $\widehat q$ of the sample (ETH $1.51$ and BTC
$1.40$, both above the range spanned by the indices) and a negligible
$\Delta\ell$ ($0.6$ and $0.7$), that is heavy tails without directional
imbalance.

This resolves the contrast with the sample kurtosis. Kurtosis is a
fourth-moment functional, dominated by a handful of extreme observations,
whereas $\widehat q$ summarises the decay of the whole tail. At
$\widehat q=1.509$ the fitted ETH density has tail index
$\nu_f=(3-\widehat q)/(\widehat q-1)\approx2.9$, hence infinite theoretical
kurtosis, although its sample kurtosis at this horizon is only $5.40$.

The bootstrap intervals, reported in Appendix~\ref{app:real_tables}, are
consistent with the distinction. At $\Delta=1$ they exclude zero for every
equity index except BOVESPA, and for the commodity, whereas they cover it for
both cryptocurrencies. The fitted semi-variances there reproduce the
directional imbalance of the sample ones wherever the asymmetry is
significant, that is in $21$ of the $27$ market--horizon combinations. The
model therefore recovers the directional information where it is present and
collapses onto its symmetric restriction where it is not, rather than imposing
asymmetry a priori.

The entropic index also declines under temporal aggregation in eight of the
nine series, as partial central-limit behaviour would suggest, FTSE
$1.47\to1.35\to1.29$ and NASDAQ $1.48\to1.24\to1.18$ from $\Delta=1$ to
$\Delta=20$. The exception is S\&P500, whose $\widehat q$ is flat between
$\Delta=7$ and $\Delta=20$, $1.38$ against $1.39$. At the lowest frequency,
however, the sample size ($n\approx143$--$460$) widens the bootstrap intervals
to the point where the tail estimate is no longer sharp.

\paragraph{Descriptive skewness versus structural asymmetry.}
BOVESPA separates descriptive skewness from structural asymmetry. Its sample
skewness is $-0.79$ at $\Delta=1$ and $-1.14$ at $\Delta=7$, whereas the
fitted $\widehat\beta$ lies between $-0.02$ and $-0.03$ at every frequency,
with bootstrap intervals covering zero and a likelihood gain that is never
significant ($\Delta\ell\le1.0$); the multistart check reported in Section~S2
of the Supplementary Material excludes an optimizer artifact. The likelihood
therefore detects no systematic shape asymmetry in BOVESPA, and both criteria
favour the symmetric model, which suggests that its sample skewness is driven
by a small number of extreme observations rather than by a directional
imbalance of the whole distribution.

\paragraph{The tail mechanism.}
Figure~\ref{fig:pdf_qq} shows the mechanism on two series at the same horizon,
S\&P\,500 and NASDAQ at $\Delta=7$ ($n=576$ and $n=574$), both rejecting the
symmetric restriction at the $0.1\%$ level ($\Delta\ell=11.0$ and $7.1$) but
with asymmetries of different size ($\widehat\beta=-0.273$ and $-0.106$). In
both rows the asymmetric fit lies above the symmetric benchmark in the left
tail and below it in the right one, and follows the empirical bins through the
bulk and the left tail. With a single scale parameter the benchmark cannot do
both: it is too light where the empirical tail is heavier and too heavy on the
opposite side. In the right tail of the S\&P\,500 panel the few available bins
lie above both fits, where the sample carries little information.

The QQ-plots read the same discrepancy in the quantiles. The asymmetric model
lies on the $45^{\circ}$ line over the bulk and most of the tails, the residual
departure being confined to the few most extreme observations, whereas the
symmetric benchmark departs in the same direction but by more and over a wider
range of quantiles. Hence the Anderson--Darling gaps of
Table~\ref{tab:bootstrap_markets}, from $2.48$ to $0.44$ and from $1.54$ to
$0.41$, the tail-sensitive statistic weighting exactly the region where a
single scale fails.

\begin{figure}[!htbp]
  \centering
  \includegraphics[width=\textwidth]{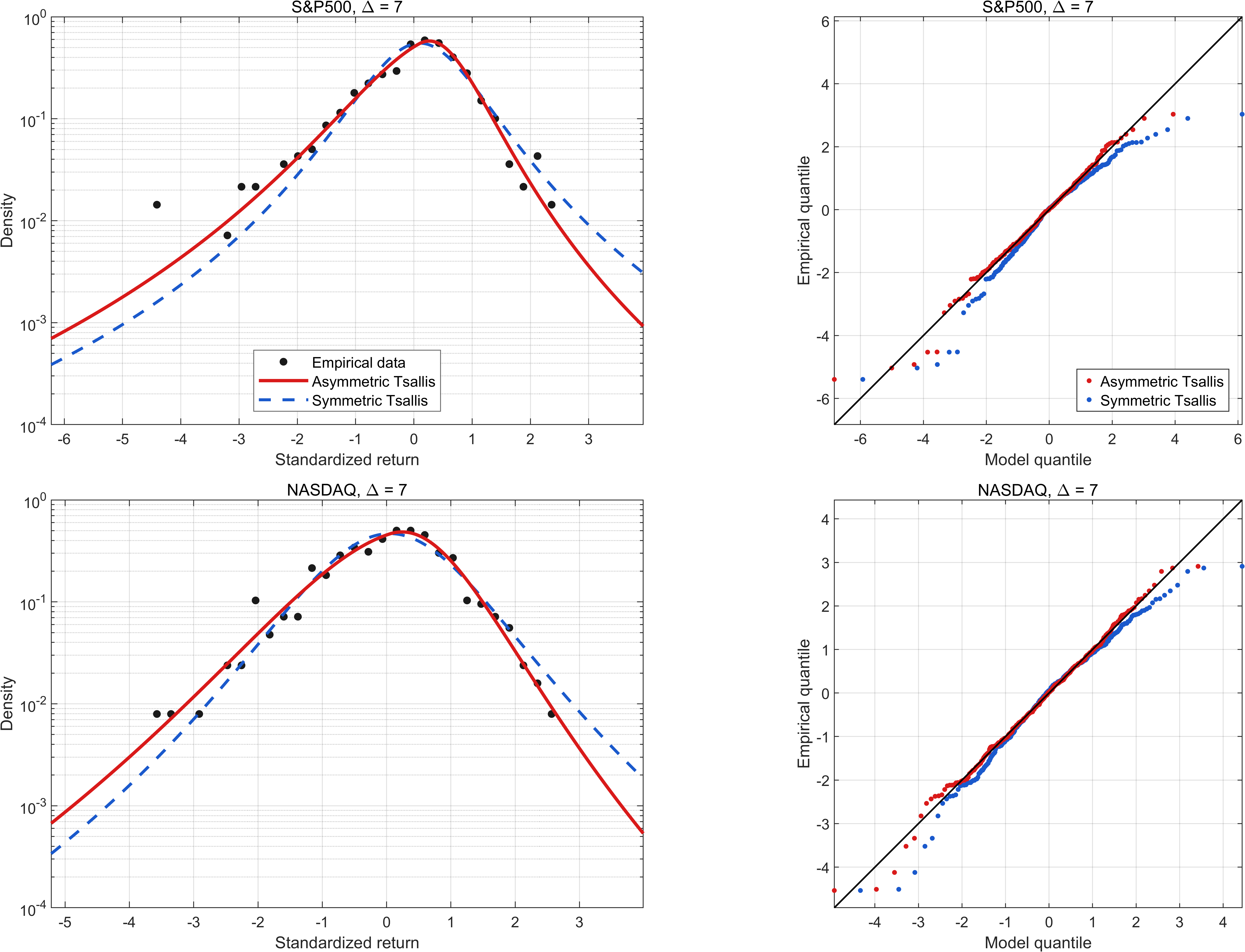}
  \caption{S\&P\,500 (top) and NASDAQ (bottom) at $\Delta=7$: fitted densities on a semi-logarithmic scale (left) and QQ-plots (right).}
  \label{fig:pdf_qq}
\end{figure}

The same pair of diagnostics, the fitted density on a semi-logarithmic scale
and the QQ-plot, is reported for all nine markets at the three horizons in
Section~S3 of the Supplementary Material (Figures~S1--S27).

\section{Conclusions}\label{sec:conclusions}

This paper develops a class of asymmetric $q$-Gaussian distributions obtained from the stationarity conditions of a Tsallis maximum-entropy problem under escort mean and escort semi-variance constraints. The proposed construction extends the standard symmetric entropy-based framework by assigning distinct directional scale parameters \(\gamma_-\) and \(\gamma_+\) to negative and positive deviations while retaining a single entropic index \(q\). In this way, tail heaviness, governed by the common asymptotic exponent, and directional tail balance, governed by the relative scale of the two sides, are controlled separately. From a theoretical perspective, the model preserves the entropy-based foundation of the $q$-Gaussian family: through the escort-density reformulation, any admissible stationary density in the proposed two-piece quadratic-kernel family is the unique global maximizer of Tsallis entropy for the escort constraints that it satisfies. At the same time, the parametrization remains parsimonious: the CDF is available in closed form through incomplete beta functions, and likelihood estimation reduces to a low-dimensional constrained optimization problem.

The analysis combines the theoretical derivation with an empirical assessment. First, a Monte Carlo experiment (\(R=10000\) replications, \(n=3000\), and a grid of \(40\) skew-\(t\) configurations), based on synthetic data generated from Azzalini skew-\(t\) distributions, shows that the proposed asymmetric specification adapts to skewed and heavy-tailed features across the whole grid. In every asymmetric scenario (\(|\alpha|\ge1\)) the 5th percentile of \(\Delta\ell\) across replications stays above the \(\chi^2_1\) threshold, its smallest value being \(2.8\) at \(|\alpha|=1\) and \(\nu=30\) against \(1.92\), so the likelihood-ratio test rejects the symmetric restriction in at least \(95\%\) of the replications. Median KS and AD discrepancies are several-fold smaller, by up to \(3.8\times\) and \(37\times\) at \(|\alpha|=2\) and \(\nu=3.1\). Under symmetry (\(\alpha=0\)), the two models are empirically indistinguishable, and the information criteria correctly favour the more parsimonious symmetric Tsallis specification. Since the data-generating process is itself a skew-\(t\), this evidence speaks to the gain from relaxing symmetry within the entropic family, not to competitiveness against the generating skew-\(t\) family itself.

Second, the application to nine real logarithmic return series, equity indices (BIST, FTSE, NASDAQ, DJIA, S\&P\,500, BOVESPA), Brent crude oil, and cryptocurrencies (BTC, ETH), at daily, \(7\)-day, and \(20\)-day horizons, illustrates the framework in practice. The series display substantial deviations from Gaussianity, with high kurtosis and non-negligible differences between negative and positive semi-variances. As in the synthetic study, the asymmetric Tsallis model improves on the symmetric benchmark chiefly in the tail-sensitive Anderson--Darling statistic, which is smaller in all \(27\) market--horizon combinations, and most clearly where empirical asymmetry is salient. The Kolmogorov--Smirnov comparison is instead mixed. On the penalised criteria the evidence is more selective: $\Delta$AIC favours the asymmetric specification in $16$ of the $27$ combinations and the stricter $\Delta$BIC in $11$, the symmetric restriction being preferred for BOVESPA and for the two cryptocurrencies at every horizon. The evidence should therefore be interpreted as an in-sample assessment of descriptive fit, not as evidence of out-of-sample or predictive performance.

Overall, the results show that, within the Tsallis maximum-entropy construction, asymmetric escort semi-variance constraints improve the descriptive fit when downside and upside dispersion differ markedly, as is often the case for financial returns.

Several directions for further research remain open. On the theoretical side, it would be useful to investigate multivariate extensions of the framework and their implications for dependence modelling and portfolio construction. On the empirical side, future work may extend the benchmark set to alternative asymmetric heavy-tailed distributions, such as skew-\(t\) and generalized hyperbolic families, and move beyond in-sample diagnostics to formal out-of-sample and predictive assessments. A particularly important step, given volatility clustering, is to replace the iid bootstrap with dependence-robust resampling and to re-estimate the model on volatility-filtered, for example GARCH-standardised, residuals, so that uncertainty bands remain valid under serial dependence; see Section~\ref{subsec:gof}. Out-of-sample evaluation via predictive log scores and value-at-risk/expected-shortfall backtesting would then assess the distribution directly in risk-measurement and forecasting applications.

\section*{Data and code availability}

The market price data used in Section~\ref{sec:empirical} are daily closing
prices from LSEG Data \& Analytics \citep{LSEG}, publicly available in
\citet{Cesarone2025DS}. Data and code will be made available
on request.

\appendix
\setcounter{table}{0}
\renewcommand{\thetable}{\thesection.\arabic{table}}
\setcounter{figure}{0}
\renewcommand{\thefigure}{\thesection.\arabic{figure}}

\section{Full result tables: artificial data}\label{app:mc_tables}

In this appendix we report the complete versions of the Monte Carlo tables
summarised by Table~\ref{tab:mc_excerpt} in Section~\ref{subsec:synthetic}.
Each table covers the full $5\times8$ scenario grid,
$\alpha\in\{2,1,0,-1,-2\}$ and $\nu\in\{3.1,3.5,3.8,4.1,5,10,15,30\}$, for a
total of $40$ configurations, each replicated $R=10000$ times with samples of
size $n=3000$. Except where stated below, every entry is the Monte Carlo median
across replications, followed in square brackets by the $[p_5,p_{95}]$ interval
of the same quantity.
Table~\ref{tab:mc_comparison_f} reports, after the two columns identifying
the scenario, the log-likelihood difference $\Delta\ell$, the AIC and BIC
differences, and the KS and AD statistics of the asymmetric fit and of the
symmetric benchmark, in this order.
Table~\ref{tab:mc_params_f} reports the fitted entropic index $\widehat q$,
the fitted mean $\widehat\mu_f$, the downside and upside semi-variances
$\widehat\sigma^2_{f,-}$ and $\widehat\sigma^2_{f,+}$ and the total variance
$\widehat\sigma^2_{f,\mathrm{tot}}$; its last three columns are not estimates
but the exact ordinary moments $\mu_{\mathrm{DGP}}$ and
$\sigma^2_{\mathrm{DGP},\pm}$ of the Azzalini skew-$t$ generating process,
reported as benchmark values.
Table~\ref{tab:mc_structural_f} reports the five structural
maximum-likelihood parameters, in the order
$\widehat q$, $\widehat\mu_q$, $\widehat\beta$, $\widehat\gamma_-$ and
$\widehat\gamma_+$, of the kernel $H_\theta$.

\begin{table}[!h]
  \centering
  \makebox[\textwidth]{%
    \resizebox{\textwidth}{!}{%
    \begin{tabular}{rrrrrrrrr}
\toprule
\multicolumn{2}{c}{True param.} & \multicolumn{3}{c}{Model comparison: median $[p_5,p_{95}]$} & \multicolumn{2}{c}{KS: median $[p_5,p_{95}]$} & \multicolumn{2}{c}{AD: median $[p_5,p_{95}]$} \\
\cmidrule(lr){1-2}\cmidrule(lr){3-5}\cmidrule(lr){6-7}\cmidrule(lr){8-9}
$\alpha$ & $\nu$ & $\widetilde{\Delta\ell}$ & $\widetilde{\Delta\mathrm{AIC}}$ & $\widetilde{\Delta\mathrm{BIC}}$ & $\widetilde{D}_n^{\,\mathrm{asym}}$ & $\widetilde{D}_n^{\,\mathrm{sym}}$ & $\widetilde{A}_n^{2,\mathrm{asym}}$ & $\widetilde{A}_n^{2,\mathrm{sym}}$ \\
\midrule
\multirow{8}{*}{$2$} & 3.1 & 154.8 [129.5, 181.6] & -307.6 [-361.2, -257.0] & -301.6 [-355.2, -251.0] & 0.0161 [0.0111, 0.0228] & 0.0606 [0.0543, 0.0669] & 0.829 [0.380, 1.510] & 30.613 [25.334, 36.508] \\
 & 3.5 & 143.3 [119.1, 168.9] & -284.5 [-335.8, -236.2] & -278.5 [-329.8, -230.2] & 0.0154 [0.0105, 0.0217] & 0.0559 [0.0499, 0.0621] & 0.731 [0.331, 1.390] & 26.247 [21.576, 31.465] \\
 & 3.8 & 135.8 [112.6, 161.6] & -269.6 [-321.2, -223.2] & -263.6 [-315.2, -217.2] & 0.0150 [0.0102, 0.0211] & 0.0531 [0.0472, 0.0593] & 0.692 [0.306, 1.332] & 23.765 [19.439, 28.713] \\
 & 4.1 & 129.6 [106.0, 154.7] & -257.2 [-307.5, -209.9] & -251.2 [-301.5, -203.9] & 0.0146 [0.0100, 0.0205] & 0.0507 [0.0448, 0.0567] & 0.652 [0.294, 1.252] & 21.818 [17.644, 26.462] \\
 & 5 & 114.7 [92.8, 138.9] & -227.3 [-275.9, -183.6] & -221.3 [-269.9, -177.6] & 0.0138 [0.0094, 0.0197] & 0.0451 [0.0393, 0.0510] & 0.569 [0.244, 1.139] & 17.725 [14.173, 21.801] \\
 & 10 & 80.4 [62.1, 100.5] & -158.9 [-199.1, -122.2] & -152.9 [-193.1, -116.2] & 0.0121 [0.0082, 0.0173] & 0.0333 [0.0284, 0.0386] & 0.408 [0.179, 0.848] & 10.746 [8.188, 13.699] \\
 & 15 & 69.1 [52.1, 87.9] & -136.2 [-173.9, -102.2] & -130.2 [-167.9, -96.2] & 0.0116 [0.0079, 0.0165] & 0.0317 [0.0262, 0.0387] & 0.364 [0.165, 0.778] & 9.135 [6.699, 11.878] \\
 & 30 & 57.7 [42.4, 74.5] & -113.5 [-147.1, -82.8] & -107.5 [-141.1, -76.8] & 0.0113 [0.0078, 0.0161] & 0.0321 [0.0254, 0.0400] & 0.343 [0.162, 0.737] & 7.738 [5.562, 10.288] \\
\addlinespace
\multirow{8}{*}{$1$} & 3.1 & 53.8 [38.8, 71.8] & -105.6 [-141.6, -75.7] & -99.6 [-135.6, -69.7] & 0.0110 [0.0076, 0.0159] & 0.0347 [0.0290, 0.0407] & 0.327 [0.157, 0.672] & 10.749 [7.737, 14.376] \\
 & 3.5 & 47.7 [33.3, 64.2] & -93.4 [-126.3, -64.6] & -87.4 [-120.3, -58.6] & 0.0109 [0.0076, 0.0157] & 0.0314 [0.0260, 0.0372] & 0.316 [0.151, 0.664] & 8.807 [6.170, 11.918] \\
 & 3.8 & 43.9 [30.2, 60.0] & -85.7 [-118.0, -58.4] & -79.7 [-112.0, -52.4] & 0.0108 [0.0075, 0.0156] & 0.0294 [0.0241, 0.0351] & 0.306 [0.150, 0.636] & 7.727 [5.329, 10.646] \\
 & 4.1 & 40.9 [27.6, 56.5] & -79.7 [-111.0, -53.2] & -73.7 [-105.0, -47.2] & 0.0107 [0.0075, 0.0153] & 0.0278 [0.0227, 0.0333] & 0.302 [0.147, 0.617] & 6.903 [4.653, 9.581] \\
 & 5 & 33.4 [21.6, 47.6] & -64.7 [-93.2, -41.2] & -58.7 [-87.2, -35.2] & 0.0105 [0.0074, 0.0151] & 0.0240 [0.0192, 0.0292] & 0.287 [0.144, 0.607] & 5.203 [3.383, 7.472] \\
 & 10 & 17.0 [9.0, 27.8] & -32.0 [-53.6, -16.0] & -26.0 [-47.6, -10.0] & 0.0102 [0.0072, 0.0145] & 0.0185 [0.0136, 0.0248] & 0.264 [0.136, 0.552] & 2.431 [1.323, 3.969] \\
 & 15 & 12.2 [5.6, 21.5] & -22.4 [-41.0, -9.2] & -16.4 [-35.0, -3.2] & 0.0101 [0.0072, 0.0144] & 0.0179 [0.0125, 0.0248] & 0.260 [0.135, 0.533] & 1.803 [0.878, 3.164] \\
 & 30 & 7.9 [2.8, 15.8] & -13.8 [-29.5, -3.6] & -7.8 [-23.5, 2.4] & 0.0100 [0.0071, 0.0143] & 0.0170 [0.0115, 0.0239] & 0.252 [0.134, 0.526] & 1.278 [0.549, 2.460] \\
\addlinespace
\multirow{8}{*}{$0$} & 3.1 & 0.2 [0.0, 1.9] & 1.6 [-1.9, 2.0] & 7.6 [4.1, 8.0] & 0.0091 [0.0067, 0.0127] & 0.0098 [0.0071, 0.0136] & 0.208 [0.118, 0.388] & 0.283 [0.144, 0.633] \\
 & 3.5 & 0.2 [0.0, 2.0] & 1.6 [-1.9, 2.0] & 7.6 [4.1, 8.0] & 0.0091 [0.0067, 0.0128] & 0.0097 [0.0071, 0.0136] & 0.209 [0.118, 0.393] & 0.280 [0.142, 0.623] \\
 & 3.8 & 0.2 [0.0, 1.9] & 1.5 [-1.9, 2.0] & 7.6 [4.1, 8.0] & 0.0092 [0.0067, 0.0129] & 0.0098 [0.0071, 0.0136] & 0.211 [0.118, 0.398] & 0.277 [0.145, 0.619] \\
 & 4.1 & 0.2 [0.0, 2.0] & 1.5 [-1.9, 2.0] & 7.6 [4.1, 8.0] & 0.0092 [0.0068, 0.0130] & 0.0097 [0.0071, 0.0137] & 0.212 [0.118, 0.402] & 0.276 [0.144, 0.611] \\
 & 5 & 0.2 [0.0, 1.9] & 1.5 [-1.8, 2.0] & 7.5 [4.2, 8.0] & 0.0093 [0.0068, 0.0130] & 0.0098 [0.0071, 0.0138] & 0.216 [0.120, 0.408] & 0.279 [0.144, 0.598] \\
 & 10 & 0.2 [0.0, 1.9] & 1.6 [-1.9, 2.0] & 7.6 [4.2, 8.0] & 0.0095 [0.0069, 0.0134] & 0.0101 [0.0072, 0.0147] & 0.224 [0.124, 0.437] & 0.282 [0.144, 0.614] \\
 & 15 & 0.2 [0.0, 2.0] & 1.6 [-1.9, 2.0] & 7.6 [4.1, 8.0] & 0.0096 [0.0069, 0.0135] & 0.0103 [0.0073, 0.0151] & 0.228 [0.125, 0.448] & 0.286 [0.146, 0.622] \\
 & 30 & 0.2 [0.0, 2.0] & 1.6 [-2.0, 2.0] & 7.6 [4.0, 8.0] & 0.0096 [0.0070, 0.0136] & 0.0105 [0.0074, 0.0156] & 0.233 [0.127, 0.452] & 0.295 [0.148, 0.656] \\
\addlinespace
\multirow{8}{*}{$-1$} & 3.1 & 53.8 [38.8, 71.6] & -105.6 [-141.2, -75.7] & -99.6 [-135.2, -69.6] & 0.0110 [0.0076, 0.0158] & 0.0347 [0.0291, 0.0406] & 0.327 [0.155, 0.665] & 10.745 [7.735, 14.381] \\
 & 3.5 & 47.6 [33.5, 64.4] & -93.2 [-126.8, -64.9] & -87.2 [-120.8, -58.9] & 0.0108 [0.0075, 0.0156] & 0.0314 [0.0260, 0.0372] & 0.311 [0.151, 0.640] & 8.785 [6.186, 11.976] \\
 & 3.8 & 44.1 [30.2, 60.4] & -86.1 [-118.8, -58.4] & -80.1 [-112.8, -52.4] & 0.0107 [0.0075, 0.0155] & 0.0295 [0.0243, 0.0351] & 0.302 [0.145, 0.634] & 7.783 [5.348, 10.700] \\
 & 4.1 & 40.6 [27.6, 56.3] & -79.3 [-110.5, -53.3] & -73.3 [-104.5, -47.3] & 0.0106 [0.0074, 0.0153] & 0.0278 [0.0228, 0.0332] & 0.299 [0.145, 0.629] & 6.894 [4.678, 9.583] \\
 & 5 & 33.3 [21.6, 47.4] & -64.6 [-92.8, -41.1] & -58.6 [-86.8, -35.1] & 0.0105 [0.0074, 0.0150] & 0.0241 [0.0193, 0.0292] & 0.287 [0.145, 0.598] & 5.197 [3.374, 7.445] \\
 & 10 & 17.1 [9.1, 27.8] & -32.2 [-53.6, -16.2] & -26.2 [-47.6, -10.2] & 0.0102 [0.0072, 0.0144] & 0.0185 [0.0136, 0.0246] & 0.265 [0.136, 0.545] & 2.431 [1.327, 3.958] \\
 & 15 & 12.3 [5.6, 21.5] & -22.6 [-41.1, -9.1] & -16.6 [-35.1, -3.1] & 0.0101 [0.0072, 0.0143] & 0.0179 [0.0124, 0.0247] & 0.260 [0.135, 0.529] & 1.814 [0.879, 3.160] \\
 & 30 & 7.9 [2.8, 15.6] & -13.8 [-29.3, -3.7] & -7.8 [-23.3, 2.3] & 0.0100 [0.0072, 0.0141] & 0.0171 [0.0114, 0.0239] & 0.253 [0.132, 0.515] & 1.284 [0.549, 2.439] \\
\addlinespace
\multirow{8}{*}{$-2$} & 3.1 & 154.7 [129.6, 182.3] & -307.4 [-362.6, -257.3] & -301.4 [-356.6, -251.2] & 0.0161 [0.0109, 0.0226] & 0.0606 [0.0543, 0.0670] & 0.818 [0.381, 1.520] & 30.632 [25.353, 36.666] \\
 & 3.5 & 142.8 [118.5, 169.3] & -283.7 [-336.6, -235.0] & -277.6 [-330.6, -229.0] & 0.0154 [0.0105, 0.0217] & 0.0560 [0.0498, 0.0621] & 0.737 [0.330, 1.399] & 26.207 [21.463, 31.463] \\
 & 3.8 & 135.7 [112.2, 162.2] & -269.5 [-322.5, -222.3] & -263.5 [-316.5, -216.3] & 0.0150 [0.0102, 0.0211] & 0.0531 [0.0472, 0.0593] & 0.687 [0.306, 1.322] & 23.784 [19.406, 28.774] \\
 & 4.1 & 129.5 [106.5, 154.7] & -256.9 [-307.5, -210.9] & -250.9 [-301.5, -204.9] & 0.0147 [0.0099, 0.0208] & 0.0506 [0.0448, 0.0567] & 0.649 [0.284, 1.264] & 21.777 [17.676, 26.401] \\
 & 5 & 114.8 [93.2, 138.5] & -227.6 [-275.0, -184.5] & -221.6 [-269.0, -178.5] & 0.0138 [0.0094, 0.0196] & 0.0450 [0.0394, 0.0509] & 0.562 [0.250, 1.121] & 17.721 [14.222, 21.731] \\
 & 10 & 80.5 [62.3, 100.5] & -159.0 [-199.0, -122.5] & -153.0 [-193.0, -116.5] & 0.0120 [0.0082, 0.0173] & 0.0332 [0.0284, 0.0386] & 0.404 [0.179, 0.841] & 10.768 [8.224, 13.700] \\
 & 15 & 69.1 [52.2, 88.1] & -136.1 [-174.1, -102.3] & -130.1 [-168.1, -96.3] & 0.0116 [0.0079, 0.0166] & 0.0318 [0.0264, 0.0388] & 0.363 [0.165, 0.772] & 9.120 [6.749, 11.871] \\
 & 30 & 57.8 [42.5, 74.8] & -113.5 [-147.6, -82.9] & -107.5 [-141.6, -76.9] & 0.0112 [0.0078, 0.0161] & 0.0323 [0.0255, 0.0399] & 0.341 [0.160, 0.732] & 7.753 [5.563, 10.311] \\
\bottomrule
    \end{tabular}%
    }%
  }%
  \caption{Monte Carlo model comparison, full grid: median $[p_5,p_{95}]$.}
  \label{tab:mc_comparison_f}
\end{table}

\begin{table}[h]
  \centering
  \makebox[\textwidth]{%
    \resizebox{\textwidth}{!}{%
    \begin{tabular}{rrrrrrrrrr}
\toprule
\multicolumn{2}{c}{True param.} & \multicolumn{5}{c}{Asymmetric Tsallis ($f_\theta$)} & \multicolumn{3}{c}{DGP (Azzalini)} \\
\cmidrule(lr){1-2}\cmidrule(lr){3-7}\cmidrule(lr){8-10}
$\alpha$ & $\nu$ & $\widehat{q}$ & $\widehat{\mu}_f$ & $\widehat{\sigma}^2_{f,-}$ & $\widehat{\sigma}^2_{f,+}$ & $\widehat{\sigma}^2_{f,\mathrm{tot}}$ & $\mu_{\mathrm{DGP}}$ & $\sigma^2_{\mathrm{DGP},-}$ & $\sigma^2_{\mathrm{DGP},+}$ \\
\midrule
\multirow{8}{*}{$2$} & 3.1 & 1.310 [1.234, 1.370] & $0.957$ [$0.919$, $0.998$] & 0.439 [0.409, 0.473] & 1.115 [0.964, 1.296] & 1.555 [1.380, 1.764] & $0.973$ & $0.452$ & $1.412$ \\
 & 3.5 & 1.277 [1.199, 1.337] & $0.920$ [$0.885$, $0.957$] & 0.398 [0.372, 0.426] & 0.906 [0.796, 1.040] & 1.303 [1.173, 1.460] & $0.931$ & $0.404$ & $1.062$ \\
 & 3.8 & 1.257 [1.180, 1.315] & $0.898$ [$0.865$, $0.934$] & 0.375 [0.351, 0.401] & 0.804 [0.713, 0.913] & 1.180 [1.069, 1.310] & $0.907$ & $0.379$ & $0.909$ \\
 & 4.1 & 1.240 [1.163, 1.297] & $0.881$ [$0.848$, $0.915$] & 0.358 [0.336, 0.381] & 0.729 [0.650, 0.823] & 1.087 [0.990, 1.197] & $0.889$ & $0.360$ & $0.803$ \\
 & 5 & 1.197 [1.128, 1.253] & $0.844$ [$0.815$, $0.874$] & 0.322 [0.304, 0.340] & 0.590 [0.535, 0.654] & 0.912 [0.843, 0.988] & $0.849$ & $0.322$ & $0.624$ \\
 & 10 & 1.090 [1.028, 1.145] & $0.772$ [$0.748$, $0.796$] & 0.259 [0.247, 0.272] & 0.389 [0.361, 0.418] & 0.648 [0.610, 0.687] & $0.773$ & $0.259$ & $0.393$ \\
 & 15 & 1.051 [1.001, 1.105] & $0.751$ [$0.728$, $0.774$] & 0.243 [0.231, 0.254] & 0.344 [0.322, 0.367] & 0.587 [0.556, 0.619] & $0.752$ & $0.242$ & $0.346$ \\
 & 30 & 1.009 [1.001, 1.062] & $0.732$ [$0.710$, $0.753$] & 0.228 [0.218, 0.239] & 0.307 [0.289, 0.325] & 0.535 [0.509, 0.561] & $0.732$ & $0.228$ & $0.307$ \\
\addlinespace
\multirow{8}{*}{$1$} & 3.1 & 1.437 [1.391, 1.478] & $0.757$ [$0.717$, $0.799$] & 0.709 [0.644, 0.791] & 1.307 [1.124, 1.549] & 2.016 [1.781, 2.327] & $0.769$ & $0.711$ & $1.508$ \\
 & 3.5 & 1.397 [1.352, 1.440] & $0.727$ [$0.689$, $0.765$] & 0.629 [0.579, 0.690] & 1.056 [0.930, 1.212] & 1.685 [1.519, 1.890] & $0.736$ & $0.627$ & $1.164$ \\
 & 3.8 & 1.373 [1.328, 1.416] & $0.710$ [$0.674$, $0.747$] & 0.589 [0.545, 0.640] & 0.937 [0.834, 1.061] & 1.526 [1.387, 1.687] & $0.717$ & $0.585$ & $1.012$ \\
 & 4.1 & 1.351 [1.304, 1.394] & $0.696$ [$0.661$, $0.732$] & 0.558 [0.518, 0.602] & 0.852 [0.764, 0.957] & 1.411 [1.291, 1.548] & $0.702$ & $0.553$ & $0.906$ \\
 & 5 & 1.299 [1.252, 1.341] & $0.667$ [$0.635$, $0.700$] & 0.497 [0.466, 0.531] & 0.697 [0.638, 0.766] & 1.194 [1.111, 1.288] & $0.671$ & $0.493$ & $0.724$ \\
 & 10 & 1.165 [1.116, 1.211] & $0.610$ [$0.583$, $0.639$] & 0.396 [0.376, 0.417] & 0.477 [0.448, 0.510] & 0.874 [0.828, 0.921] & $0.611$ & $0.394$ & $0.482$ \\
 & 15 & 1.115 [1.064, 1.161] & $0.594$ [$0.567$, $0.621$] & 0.371 [0.353, 0.390] & 0.429 [0.405, 0.455] & 0.799 [0.762, 0.840] & $0.595$ & $0.370$ & $0.431$ \\
 & 30 & 1.062 [1.010, 1.108] & $0.579$ [$0.553$, $0.605$] & 0.348 [0.332, 0.365] & 0.388 [0.368, 0.409] & 0.736 [0.703, 0.770] & $0.579$ & $0.348$ & $0.389$ \\
\addlinespace
\multirow{8}{*}{$0$} & 3.1 & 1.487 [1.448, 1.524] & $0.000$ [$-0.047$, $0.047$] & 1.408 [1.216, 1.675] & 1.407 [1.217, 1.680] & 2.817 [2.452, 3.326] & $0.000$ & $1.405$ & $1.405$ \\
 & 3.5 & 1.443 [1.403, 1.481] & $0.000$ [$-0.043$, $0.043$] & 1.164 [1.041, 1.329] & 1.165 [1.038, 1.330] & 2.330 [2.094, 2.635] & $0.000$ & $1.166$ & $1.166$ \\
 & 3.8 & 1.415 [1.375, 1.453] & $0.000$ [$-0.043$, $0.042$] & 1.054 [0.953, 1.184] & 1.054 [0.949, 1.183] & 2.109 [1.920, 2.350] & $0.000$ & $1.055$ & $1.055$ \\
 & 4.1 & 1.391 [1.350, 1.429] & $0.000$ [$-0.041$, $0.041$] & 0.976 [0.887, 1.080] & 0.974 [0.886, 1.080] & 1.950 [1.788, 2.144] & $0.000$ & $0.976$ & $0.976$ \\
 & 5 & 1.331 [1.289, 1.371] & $0.000$ [$-0.039$, $0.038$] & 0.832 [0.770, 0.903] & 0.832 [0.770, 0.905] & 1.665 [1.554, 1.795] & $0.000$ & $0.833$ & $0.833$ \\
 & 10 & 1.180 [1.133, 1.222] & $0.000$ [$-0.033$, $0.034$] & 0.625 [0.590, 0.663] & 0.625 [0.590, 0.662] & 1.250 [1.188, 1.316] & $0.000$ & $0.625$ & $0.625$ \\
 & 15 & 1.123 [1.075, 1.167] & $0.000$ [$-0.033$, $0.032$] & 0.577 [0.547, 0.608] & 0.577 [0.547, 0.608] & 1.154 [1.100, 1.210] & $0.000$ & $0.577$ & $0.577$ \\
 & 30 & 1.062 [1.012, 1.107] & $0.000$ [$-0.031$, $0.031$] & 0.536 [0.510, 0.563] & 0.536 [0.509, 0.563] & 1.072 [1.024, 1.120] & $0.000$ & $0.536$ & $0.536$ \\
\addlinespace
\multirow{8}{*}{$-1$} & 3.1 & 1.437 [1.391, 1.479] & $-0.757$ [$-0.799$, $-0.716$] & 1.309 [1.125, 1.552] & 0.710 [0.644, 0.794] & 2.019 [1.781, 2.332] & $-0.769$ & $1.508$ & $0.711$ \\
 & 3.5 & 1.398 [1.352, 1.440] & $-0.727$ [$-0.765$, $-0.689$] & 1.056 [0.929, 1.213] & 0.630 [0.579, 0.690] & 1.684 [1.520, 1.890] & $-0.736$ & $1.164$ & $0.627$ \\
 & 3.8 & 1.373 [1.327, 1.416] & $-0.710$ [$-0.747$, $-0.674$] & 0.938 [0.836, 1.062] & 0.589 [0.544, 0.640] & 1.527 [1.391, 1.692] & $-0.717$ & $1.012$ & $0.585$ \\
 & 4.1 & 1.351 [1.306, 1.394] & $-0.696$ [$-0.732$, $-0.661$] & 0.853 [0.766, 0.954] & 0.558 [0.518, 0.602] & 1.412 [1.292, 1.546] & $-0.702$ & $0.906$ & $0.553$ \\
 & 5 & 1.299 [1.252, 1.342] & $-0.667$ [$-0.700$, $-0.634$] & 0.697 [0.638, 0.765] & 0.497 [0.465, 0.531] & 1.193 [1.110, 1.289] & $-0.671$ & $0.724$ & $0.493$ \\
 & 10 & 1.166 [1.117, 1.210] & $-0.610$ [$-0.639$, $-0.582$] & 0.478 [0.448, 0.509] & 0.396 [0.376, 0.417] & 0.874 [0.828, 0.922] & $-0.611$ & $0.482$ & $0.394$ \\
 & 15 & 1.115 [1.065, 1.160] & $-0.594$ [$-0.621$, $-0.567$] & 0.429 [0.404, 0.454] & 0.371 [0.353, 0.389] & 0.799 [0.761, 0.839] & $-0.595$ & $0.431$ & $0.370$ \\
 & 30 & 1.062 [1.009, 1.108] & $-0.579$ [$-0.605$, $-0.553$] & 0.388 [0.368, 0.409] & 0.348 [0.332, 0.365] & 0.736 [0.703, 0.770] & $-0.579$ & $0.389$ & $0.348$ \\
\addlinespace
\multirow{8}{*}{$-2$} & 3.1 & 1.310 [1.235, 1.372] & $-0.957$ [$-0.997$, $-0.919$] & 1.115 [0.965, 1.297] & 0.440 [0.409, 0.474] & 1.556 [1.381, 1.763] & $-0.973$ & $1.412$ & $0.452$ \\
 & 3.5 & 1.278 [1.202, 1.337] & $-0.919$ [$-0.957$, $-0.884$] & 0.907 [0.795, 1.037] & 0.398 [0.372, 0.426] & 1.304 [1.172, 1.455] & $-0.931$ & $1.062$ & $0.404$ \\
 & 3.8 & 1.258 [1.181, 1.315] & $-0.898$ [$-0.933$, $-0.864$] & 0.805 [0.713, 0.913] & 0.375 [0.352, 0.401] & 1.180 [1.070, 1.307] & $-0.907$ & $0.909$ & $0.379$ \\
 & 4.1 & 1.240 [1.165, 1.298] & $-0.880$ [$-0.915$, $-0.848$] & 0.729 [0.650, 0.821] & 0.357 [0.336, 0.381] & 1.087 [0.991, 1.196] & $-0.889$ & $0.803$ & $0.360$ \\
 & 5 & 1.198 [1.129, 1.254] & $-0.843$ [$-0.874$, $-0.814$] & 0.589 [0.535, 0.653] & 0.322 [0.303, 0.341] & 0.911 [0.842, 0.990] & $-0.849$ & $0.624$ & $0.322$ \\
 & 10 & 1.090 [1.028, 1.144] & $-0.772$ [$-0.797$, $-0.747$] & 0.389 [0.362, 0.418] & 0.259 [0.246, 0.272] & 0.647 [0.611, 0.687] & $-0.773$ & $0.393$ & $0.259$ \\
 & 15 & 1.051 [1.001, 1.103] & $-0.751$ [$-0.775$, $-0.728$] & 0.344 [0.322, 0.367] & 0.243 [0.231, 0.254] & 0.586 [0.557, 0.618] & $-0.752$ & $0.346$ & $0.242$ \\
 & 30 & 1.008 [1.001, 1.061] & $-0.731$ [$-0.754$, $-0.709$] & 0.307 [0.289, 0.325] & 0.228 [0.217, 0.239] & 0.535 [0.509, 0.562] & $-0.732$ & $0.307$ & $0.228$ \\
\bottomrule
    \end{tabular}%
    }%
  }%
  \caption{Monte Carlo fitted moments and skew-$t$ benchmark, full grid: median $[p_5,p_{95}]$.}
  \label{tab:mc_params_f}
\end{table}

\begin{table}[h]
  \centering
  \adjustbox{max width=\textwidth}{\scriptsize
    \begin{tabular}{rrrrrrr}
\toprule
\multicolumn{2}{c}{True param.} & \multicolumn{5}{c}{Asymmetric Tsallis ($f_\theta$)} \\
\cmidrule(lr){1-2}\cmidrule(lr){3-7}
$\alpha$ & $\nu$ & $\widehat{q}$ & $\widehat{\mu}_q$ & $\widehat{\beta}$ & $\widehat{\gamma}_-$ & $\widehat{\gamma}_+$ \\
\midrule
\multirow{8}{*}{$2$} & 3.1 & 1.310 [1.234, 1.370] & $0.792$ [$0.754$, $0.838$] & $0.350$ [$0.269$, $0.419$] & 0.626 [0.430, 0.816] & 0.057 [0.021, 0.104] \\
 & 3.5 & 1.277 [1.199, 1.337] & $0.788$ [$0.750$, $0.832$] & $0.301$ [$0.224$, $0.367$] & 0.548 [0.362, 0.726] & 0.055 [0.018, 0.100] \\
 & 3.8 & 1.257 [1.180, 1.315] & $0.784$ [$0.749$, $0.828$] & $0.273$ [$0.197$, $0.336$] & 0.504 [0.322, 0.671] & 0.054 [0.018, 0.097] \\
 & 4.1 & 1.240 [1.163, 1.297] & $0.781$ [$0.746$, $0.823$] & $0.249$ [$0.174$, $0.309$] & 0.465 [0.291, 0.624] & 0.053 [0.019, 0.094] \\
 & 5 & 1.197 [1.128, 1.253] & $0.773$ [$0.741$, $0.810$] & $0.193$ [$0.129$, $0.247$] & 0.372 [0.226, 0.514] & 0.049 [0.020, 0.087] \\
 & 10 & 1.090 [1.028, 1.145] & $0.750$ [$0.722$, $0.779$] & $0.074$ [$0.023$, $0.118$] & 0.159 [0.045, 0.271] & 0.031 [0.007, 0.063] \\
 & 15 & 1.051 [1.001, 1.105] & $0.740$ [$0.715$, $0.767$] & $0.039$ [$0.001$, $0.079$] & 0.087 [0.002, 0.189] & 0.019 [0.000, 0.050] \\
 & 30 & 1.009 [1.001, 1.062] & $0.728$ [$0.705$, $0.751$] & $0.006$ [$0.001$, $0.043$] & 0.015 [0.002, 0.109] & 0.004 [0.000, 0.032] \\
\addlinespace
\multirow{8}{*}{$1$} & 3.1 & 1.437 [1.391, 1.478] & $0.616$ [$0.583$, $0.650$] & $0.254$ [$0.212$, $0.299$] & 0.587 [0.487, 0.698] & 0.195 [0.144, 0.252] \\
 & 3.5 & 1.397 [1.352, 1.440] & $0.614$ [$0.582$, $0.647$] & $0.216$ [$0.178$, $0.257$] & 0.515 [0.423, 0.616] & 0.182 [0.134, 0.237] \\
 & 3.8 & 1.373 [1.328, 1.416] & $0.612$ [$0.580$, $0.645$] & $0.194$ [$0.158$, $0.233$] & 0.472 [0.386, 0.570] & 0.173 [0.127, 0.227] \\
 & 4.1 & 1.351 [1.304, 1.394] & $0.610$ [$0.579$, $0.642$] & $0.175$ [$0.141$, $0.213$] & 0.436 [0.352, 0.528] & 0.166 [0.120, 0.217] \\
 & 5 & 1.299 [1.252, 1.341] & $0.606$ [$0.576$, $0.637$] & $0.134$ [$0.104$, $0.167$] & 0.353 [0.278, 0.435] & 0.146 [0.105, 0.194] \\
 & 10 & 1.165 [1.116, 1.211] & $0.590$ [$0.562$, $0.618$] & $0.052$ [$0.033$, $0.073$] & 0.171 [0.112, 0.233] & 0.090 [0.055, 0.130] \\
 & 15 & 1.115 [1.064, 1.161] & $0.583$ [$0.555$, $0.610$] & $0.030$ [$0.016$, $0.047$] & 0.112 [0.059, 0.170] & 0.066 [0.033, 0.102] \\
 & 30 & 1.062 [1.010, 1.108] & $0.574$ [$0.548$, $0.600$] & $0.012$ [$0.002$, $0.025$] & 0.057 [0.009, 0.106] & 0.037 [0.006, 0.071] \\
\addlinespace
\multirow{8}{*}{$0$} & 3.1 & 1.487 [1.448, 1.524] & $0.000$ [$-0.037$, $0.037$] & $0.000$ [$-0.039$, $0.038$] & 0.321 [0.265, 0.384] & 0.321 [0.266, 0.384] \\
 & 3.5 & 1.443 [1.403, 1.481] & $0.000$ [$-0.036$, $0.036$] & $0.000$ [$-0.035$, $0.034$] & 0.284 [0.232, 0.340] & 0.283 [0.233, 0.341] \\
 & 3.8 & 1.415 [1.375, 1.453] & $0.000$ [$-0.036$, $0.036$] & $-0.001$ [$-0.033$, $0.032$] & 0.261 [0.212, 0.316] & 0.261 [0.213, 0.315] \\
 & 4.1 & 1.391 [1.350, 1.429] & $0.000$ [$-0.036$, $0.036$] & $0.000$ [$-0.030$, $0.030$] & 0.242 [0.196, 0.295] & 0.242 [0.196, 0.294] \\
 & 5 & 1.331 [1.289, 1.371] & $0.000$ [$-0.035$, $0.035$] & $0.000$ [$-0.025$, $0.025$] & 0.198 [0.157, 0.244] & 0.198 [0.157, 0.244] \\
 & 10 & 1.180 [1.133, 1.222] & $0.000$ [$-0.033$, $0.033$] & $0.000$ [$-0.013$, $0.013$] & 0.099 [0.066, 0.133] & 0.099 [0.067, 0.133] \\
 & 15 & 1.123 [1.075, 1.167] & $0.000$ [$-0.032$, $0.032$] & $0.000$ [$-0.009$, $0.009$] & 0.065 [0.037, 0.096] & 0.065 [0.036, 0.096] \\
 & 30 & 1.062 [1.012, 1.107] & $0.000$ [$-0.031$, $0.031$] & $0.000$ [$-0.005$, $0.005$] & 0.032 [0.006, 0.060] & 0.032 [0.006, 0.060] \\
\addlinespace
\multirow{8}{*}{$-1$} & 3.1 & 1.437 [1.391, 1.479] & $-0.616$ [$-0.650$, $-0.583$] & $-0.254$ [$-0.299$, $-0.213$] & 0.196 [0.144, 0.255] & 0.588 [0.488, 0.699] \\
 & 3.5 & 1.398 [1.352, 1.440] & $-0.614$ [$-0.647$, $-0.581$] & $-0.216$ [$-0.257$, $-0.178$] & 0.182 [0.134, 0.239] & 0.516 [0.422, 0.618] \\
 & 3.8 & 1.373 [1.327, 1.416] & $-0.612$ [$-0.645$, $-0.580$] & $-0.194$ [$-0.233$, $-0.159$] & 0.173 [0.126, 0.227] & 0.473 [0.386, 0.568] \\
 & 4.1 & 1.351 [1.306, 1.394] & $-0.610$ [$-0.643$, $-0.579$] & $-0.176$ [$-0.212$, $-0.141$] & 0.165 [0.121, 0.218] & 0.436 [0.353, 0.528] \\
 & 5 & 1.299 [1.252, 1.342] & $-0.606$ [$-0.637$, $-0.575$] & $-0.134$ [$-0.167$, $-0.104$] & 0.147 [0.105, 0.195] & 0.353 [0.276, 0.435] \\
 & 10 & 1.166 [1.117, 1.210] & $-0.590$ [$-0.618$, $-0.562$] & $-0.052$ [$-0.073$, $-0.034$] & 0.091 [0.056, 0.131] & 0.171 [0.113, 0.233] \\
 & 15 & 1.115 [1.065, 1.160] & $-0.582$ [$-0.610$, $-0.556$] & $-0.030$ [$-0.047$, $-0.016$] & 0.066 [0.033, 0.103] & 0.112 [0.059, 0.169] \\
 & 30 & 1.062 [1.009, 1.108] & $-0.574$ [$-0.601$, $-0.548$] & $-0.012$ [$-0.025$, $-0.002$] & 0.037 [0.005, 0.071] & 0.057 [0.008, 0.106] \\
\addlinespace
\multirow{8}{*}{$-2$} & 3.1 & 1.310 [1.235, 1.372] & $-0.791$ [$-0.838$, $-0.753$] & $-0.350$ [$-0.420$, $-0.271$] & 0.057 [0.021, 0.104] & 0.627 [0.434, 0.817] \\
 & 3.5 & 1.278 [1.202, 1.337] & $-0.787$ [$-0.831$, $-0.750$] & $-0.302$ [$-0.368$, $-0.226$] & 0.056 [0.019, 0.100] & 0.551 [0.367, 0.726] \\
 & 3.8 & 1.258 [1.181, 1.315] & $-0.783$ [$-0.826$, $-0.748$] & $-0.273$ [$-0.336$, $-0.198$] & 0.054 [0.018, 0.097] & 0.505 [0.324, 0.671] \\
 & 4.1 & 1.240 [1.165, 1.298] & $-0.780$ [$-0.822$, $-0.745$] & $-0.249$ [$-0.309$, $-0.178$] & 0.053 [0.020, 0.095] & 0.466 [0.297, 0.624] \\
 & 5 & 1.198 [1.129, 1.254] & $-0.773$ [$-0.809$, $-0.740$] & $-0.193$ [$-0.248$, $-0.130$] & 0.049 [0.019, 0.088] & 0.374 [0.227, 0.515] \\
 & 10 & 1.090 [1.028, 1.144] & $-0.750$ [$-0.779$, $-0.722$] & $-0.074$ [$-0.118$, $-0.024$] & 0.031 [0.007, 0.063] & 0.158 [0.047, 0.270] \\
 & 15 & 1.051 [1.001, 1.103] & $-0.740$ [$-0.767$, $-0.714$] & $-0.039$ [$-0.078$, $-0.001$] & 0.019 [0.000, 0.048] & 0.087 [0.002, 0.187] \\
 & 30 & 1.008 [1.001, 1.061] & $-0.728$ [$-0.752$, $-0.705$] & $-0.006$ [$-0.042$, $-0.001$] & 0.004 [0.000, 0.031] & 0.014 [0.002, 0.107] \\
\bottomrule
    \end{tabular}}
  \caption{Monte Carlo structural estimates, full grid: median $[p_5,p_{95}]$.}
  \label{tab:mc_structural_f}
\end{table}

\FloatBarrier
\setcounter{table}{0}
\setcounter{figure}{0}
\section{Full result tables: real markets}\label{app:real_tables}

In this appendix we report the complete versions of the real-market tables
summarised by Table~\ref{tab:real_excerpt} in Section~\ref{subsec:real}. Each
table covers the nine markets at the three sampling horizons
$\Delta\in\{1,7,20\}$ trading days, for a total of $27$ market--horizon
combinations. Every entry is the median over $B=999$ iid bootstrap resamples,
followed in square brackets by the $[p_5,p_{95}]$ interval, with the exception
of the $\Delta$AIC and $\Delta$BIC columns of
Table~\ref{tab:bootstrap_markets} and of the sample semi-variances of
Table~\ref{tab:params_markets}, which are computed once on the original
sample. Table~\ref{tab:bootstrap_markets} reports, after the market, the horizon and
the sample size $n$, the log-likelihood difference $\Delta\ell$ with the stars
of the likelihood-ratio test, the AIC and BIC differences, and the KS and AD
statistics of the asymmetric fit and of the symmetric benchmark.
Table~\ref{tab:params_markets} reports the fitted entropic index
$\widehat q$, the fitted mean $\widehat\mu_f$, the semi-variances
$\widehat\sigma^2_{f,\pm}$ and the total variance
$\widehat\sigma^2_{f,\mathrm{tot}}$; its last two columns are not estimates
but the sample semi-variances $s^2_-$ and $s^2_+$ of the standardised
returns, which satisfy $s^2_-+s^2_+=1$. The fitted total variance need not
reproduce that value, since the ordinary second moment of $f_\theta$ diverges
as $q\to5/3$ and lies well above one in the series with the largest
$\widehat q$, so $\widehat\sigma^2_{f,\pm}$ and $s^2_\pm$ should be compared
in relative rather than in absolute terms.
Table~\ref{tab:structural_markets} reports the structural parameters
$\widehat q$, $\widehat\mu_q$, $\widehat\beta$, $\widehat\gamma_-$ and
$\widehat\gamma_+$.

\begin{table}[!ht]
  \centering
  \adjustbox{max width=\textwidth}{\scriptsize
    \begin{tabular}{llr rrr rrrr}
      \toprule
      & & & \multicolumn{3}{c}{Log-likelihood \& info criteria} & \multicolumn{4}{c}{Bootstrap GoF, median $[p_5,p_{95}]$} \\
      \cmidrule(lr){4-6}\cmidrule(lr){7-10}
      Market & $\Delta$ & $n$ & $\Delta\ell$ & $\Delta$AIC & $\Delta$BIC & $\widetilde{D}_n^{\,\mathrm{asym}}$ & $\widetilde{D}_n^{\,\mathrm{sym}}$ & $\widetilde{A}_n^{2,\mathrm{asym}}$ & $\widetilde{A}_n^{2,\mathrm{sym}}$ \\
      \midrule
      \multirow{3}{*}{BIST} & 1 & 4017 & $4.4^{**}$ $[0.9,10.9]$ & -7.3 & -1.0 & 0.0121 $[0.0090,0.0158]$ & 0.0149 $[0.0109,0.0195]$ & 0.559 $[0.356,0.867]$ & 1.225 $[0.618,2.268]$ \\
       & 7 & 573 & $4.1^{**}$ $[0.7,9.5]$ & -6.1 & -1.7 & 0.0303 $[0.0225,0.0399]$ & 0.0380 $[0.0281,0.0511]$ & 0.440 $[0.259,0.748]$ & 1.101 $[0.508,2.163]$ \\
       & 20 & 200 & $0.2$ $[0.0,2.2]$ & 2.0 & 5.3 & 0.0628 $[0.0433,0.0850]$ & 0.0627 $[0.0440,0.0851]$ & 0.527 $[0.295,0.931]$ & 0.583 $[0.335,1.011]$ \\
      \midrule
      \multirow{3}{*}{FTSE} & 1 & 4041 & $5.1^{**}$ $[1.2,11.7]$ & -8.0 & -1.7 & 0.0133 $[0.0098,0.0176]$ & 0.0153 $[0.0119,0.0198]$ & 0.641 $[0.358,1.042]$ & 1.566 $[0.759,2.876]$ \\
       & 7 & 577 & $1.3$ $[0.0,5.3]$ & -0.5 & 3.9 & 0.0296 $[0.0223,0.0391]$ & 0.0306 $[0.0235,0.0409]$ & 0.491 $[0.281,0.833]$ & 0.657 $[0.355,1.287]$ \\
       & 20 & 202 & $2.1^{*}$ $[0.2,6.3]$ & -1.9 & 1.4 & 0.0525 $[0.0389,0.0698]$ & 0.0542 $[0.0414,0.0739]$ & 0.450 $[0.258,0.817]$ & 0.752 $[0.370,1.469]$ \\
      \midrule
      \multirow{3}{*}{NASDAQ} & 1 & 4024 & $11.0^{***}$ $[4.5,19.9]$ & -20.6 & -14.3 & 0.0229 $[0.0185,0.0273]$ & 0.0236 $[0.0193,0.0290]$ & 2.276 $[1.477,3.145]$ & 4.367 $[2.885,6.296]$ \\
       & 7 & 574 & $7.1^{***}$ $[2.3,14.0]$ & -11.9 & -7.6 & 0.0302 $[0.0227,0.0405]$ & 0.0449 $[0.0318,0.0607]$ & 0.413 $[0.246,0.712]$ & 1.543 $[0.716,2.744]$ \\
       & 20 & 201 & $9.9^{***}$ $[4.3,16.7]$ & -18.3 & -15.0 & 0.0603 $[0.0460,0.0805]$ & 0.0877 $[0.0660,0.1193]$ & 0.738 $[0.443,1.291]$ & 2.528 $[1.386,4.309]$ \\
      \midrule
      \multirow{3}{*}{DJIA} & 1 & 4048 & $2.4^{*}$ $[0.2,7.4]$ & -2.9 & 3.4 & 0.0213 $[0.0167,0.0260]$ & 0.0213 $[0.0169,0.0262]$ & 1.880 $[1.174,2.712]$ & 2.402 $[1.508,3.651]$ \\
       & 7 & 578 & $3.0^{*}$ $[0.3,7.7]$ & -4.1 & 0.3 & 0.0297 $[0.0225,0.0395]$ & 0.0317 $[0.0242,0.0420]$ & 0.558 $[0.326,1.020]$ & 0.944 $[0.488,1.720]$ \\
       & 20 & 202 & $7.2^{***}$ $[2.2,13.6]$ & -12.1 & -8.8 & 0.0500 $[0.0352,0.0757]$ & 0.0711 $[0.0514,0.0967]$ & 0.491 $[0.255,1.140]$ & 1.632 $[0.792,2.885]$ \\
      \midrule
      \multirow{3}{*}{S\&P500} & 1 & 4034 & $4.0^{**}$ $[0.7,9.6]$ & -6.0 & 0.3 & 0.0228 $[0.0188,0.0271]$ & 0.0221 $[0.0182,0.0263]$ & 2.551 $[1.706,3.492]$ & 3.408 $[2.246,4.832]$ \\
       & 7 & 576 & $11.0^{***}$ $[4.7,19.2]$ & -19.9 & -15.6 & 0.0328 $[0.0235,0.0467]$ & 0.0472 $[0.0364,0.0590]$ & 0.435 $[0.262,0.751]$ & 2.483 $[1.304,4.148]$ \\
       & 20 & 201 & $9.4^{***}$ $[3.8,16.2]$ & -17.4 & -14.1 & 0.0573 $[0.0431,0.0751]$ & 0.0854 $[0.0649,0.1097]$ & 0.662 $[0.378,1.107]$ & 2.591 $[1.285,4.492]$ \\
      \midrule
      \multirow{3}{*}{BOVESPA} & 1 & 3968 & $0.6$ $[0.0,3.6]$ & 0.7 & 7.0 & 0.0127 $[0.0095,0.0164]$ & 0.0120 $[0.0092,0.0156]$ & 0.599 $[0.344,0.995]$ & 0.710 $[0.403,1.207]$ \\
       & 7 & 566 & $1.0$ $[0.0,4.1]$ & 0.4 & 4.7 & 0.0293 $[0.0210,0.0412]$ & 0.0297 $[0.0217,0.0404]$ & 0.430 $[0.228,0.874]$ & 0.546 $[0.267,1.052]$ \\
       & 20 & 198 & $0.6$ $[0.0,3.6]$ & 0.9 & 4.2 & 0.0515 $[0.0375,0.0708]$ & 0.0545 $[0.0388,0.0777]$ & 0.448 $[0.274,0.815]$ & 0.549 $[0.307,1.082]$ \\
      \midrule
      \multirow{3}{*}{BRENT} & 1 & 9212 & $8.2^{***}$ $[2.9,16.1]$ & -14.9 & -7.7 & 0.0093 $[0.0071,0.0121]$ & 0.0119 $[0.0094,0.0149]$ & 0.728 $[0.387,1.236]$ & 2.277 $[1.223,3.858]$ \\
       & 7 & 1316 & $5.4^{**}$ $[1.4,11.5]$ & -8.7 & -3.5 & 0.0186 $[0.0141,0.0255]$ & 0.0259 $[0.0195,0.0338]$ & 0.393 $[0.246,0.675]$ & 1.273 $[0.597,2.461]$ \\
       & 20 & 460 & $0.5$ $[0.0,4.1]$ & 0.9 & 5.0 & 0.0301 $[0.0227,0.0405]$ & 0.0317 $[0.0237,0.0439]$ & 0.364 $[0.224,0.628]$ & 0.470 $[0.264,1.042]$ \\
      \midrule
      \multirow{3}{*}{BTC} & 1 & 4152 & $0.9$ $[0.0,4.3]$ & 0.2 & 6.5 & 0.0259 $[0.0214,0.0310]$ & 0.0258 $[0.0215,0.0303]$ & 3.814 $[2.794,4.980]$ & 4.079 $[2.941,5.416]$ \\
       & 7 & 593 & $0.7$ $[0.0,4.0]$ & 0.7 & 5.1 & 0.0322 $[0.0231,0.0446]$ & 0.0342 $[0.0245,0.0467]$ & 0.477 $[0.274,0.895]$ & 0.695 $[0.339,1.370]$ \\
       & 20 & 207 & $0.3$ $[0.0,2.3]$ & 1.9 & 5.2 & 0.0540 $[0.0402,0.0736]$ & 0.0561 $[0.0405,0.0772]$ & 0.583 $[0.332,1.044]$ & 0.661 $[0.355,1.239]$ \\
      \midrule
      \multirow{3}{*}{ETH} & 1 & 2877 & $0.2$ $[0.0,2.0]$ & 2.0 & 7.9 & 0.0344 $[0.0287,0.0404]$ & 0.0342 $[0.0286,0.0401]$ & 3.473 $[2.501,4.604]$ & 3.583 $[2.592,4.755]$ \\
       & 7 & 411 & $0.6$ $[0.0,3.6]$ & 0.8 & 4.8 & 0.0419 $[0.0335,0.0524]$ & 0.0471 $[0.0367,0.0611]$ & 0.815 $[0.522,1.227]$ & 1.052 $[0.636,1.842]$ \\
       & 20 & 143 & $0.4$ $[0.0,2.8]$ & 1.6 & 4.5 & 0.0672 $[0.0481,0.0935]$ & 0.0707 $[0.0500,0.1001]$ & 0.586 $[0.289,1.071]$ & 0.719 $[0.356,1.353]$ \\
      \bottomrule
    \end{tabular}}
  \caption{Real markets, all horizons: model comparison, bootstrap median $[p_5,p_{95}]$.}
  \label{tab:bootstrap_markets}
\end{table}

\begin{table}[!ht]
  \centering
  \makebox[\textwidth]{%
    \resizebox{\textwidth}{!}{%
    \begin{tabular}{llr rrrrr rr}
      \toprule
      & & & \multicolumn{5}{c}{Asymmetric Tsallis ($f_\theta$)} & \multicolumn{2}{c}{Sample} \\
      \cmidrule(lr){4-8}\cmidrule(lr){9-10}
      Market & $\Delta$ & $n$ & $\widehat{q}$ & $\widehat{\mu}_f$ & $\widehat{\sigma}^2_{f,-}$ & $\widehat{\sigma}^2_{f,+}$ & $\widehat{\sigma}^2_{f,\mathrm{tot}}$ & $s^2_{-}$ & $s^2_{+}$ \\
      \midrule
      \multirow{3}{*}{BIST} & 1 & 4017 & 1.387 $[1.353,1.420]$ & 0.007 $[-0.018,0.033]$ & 0.5402 $[0.4973,0.5926]$ & 0.4769 $[0.4424,0.5137]$ & 1.0176 $[0.9462,1.0972]$ & 0.5519 & 0.4481 \\
       & 7 & 573 & 1.191 $[1.077,1.282]$ & -0.005 $[-0.069,0.068]$ & 0.5568 $[0.4740,0.6518]$ & 0.4465 $[0.3939,0.5080]$ & 1.0063 $[0.8780,1.1405]$ & 0.5572 & 0.4428 \\
       & 20 & 200 & 1.169 $[1.001,1.344]$ & 0.005 $[-0.113,0.118]$ & 0.4980 $[0.4034,0.6105]$ & 0.5118 $[0.4075,0.6299]$ & 1.0134 $[0.8307,1.2163]$ & 0.4924 & 0.5076 \\
      \midrule
      \multirow{3}{*}{FTSE} & 1 & 4041 & 1.472 $[1.441,1.502]$ & -0.000 $[-0.026,0.025]$ & 0.6056 $[0.5374,0.6850]$ & 0.5075 $[0.4616,0.5737]$ & 1.1122 $[1.0081,1.2500]$ & 0.5501 & 0.4499 \\
       & 7 & 577 & 1.353 $[1.252,1.434]$ & 0.010 $[-0.058,0.078]$ & 0.5306 $[0.4303,0.6751]$ & 0.4489 $[0.3809,0.5440]$ & 0.9839 $[0.8288,1.1840]$ & 0.5638 & 0.4362 \\
       & 20 & 202 & 1.292 $[1.001,1.460]$ & 0.022 $[-0.083,0.122]$ & 0.4744 $[0.3219,0.8163]$ & 0.3410 $[0.2639,0.4661]$ & 0.8206 $[0.5940,1.2800]$ & 0.6715 & 0.3285 \\
      \midrule
      \multirow{3}{*}{NASDAQ} & 1 & 4024 & 1.482 $[1.446,1.518]$ & -0.005 $[-0.033,0.024]$ & 0.6596 $[0.5865,0.7612]$ & 0.5066 $[0.4556,0.5806]$ & 1.1661 $[1.0538,1.3294]$ & 0.5442 & 0.4558 \\
       & 7 & 574 & 1.235 $[1.129,1.318]$ & -0.001 $[-0.068,0.066]$ & 0.5779 $[0.4918,0.6947]$ & 0.4256 $[0.3727,0.4877]$ & 1.0058 $[0.8774,1.1646]$ & 0.5756 & 0.4244 \\
       & 20 & 201 & 1.179 $[1.001,1.350]$ & 0.001 $[-0.106,0.107]$ & 0.6677 $[0.5381,0.8954]$ & 0.3815 $[0.3020,0.4690]$ & 1.0572 $[0.8636,1.3213]$ & 0.6131 & 0.3869 \\
      \midrule
      \multirow{3}{*}{DJIA} & 1 & 4048 & 1.531 $[1.496,1.565]$ & 0.007 $[-0.020,0.033]$ & 0.6415 $[0.5415,0.7978]$ & 0.5553 $[0.4703,0.6914]$ & 1.1955 $[1.0257,1.4872]$ & 0.5488 & 0.4512 \\
       & 7 & 578 & 1.358 $[1.217,1.457]$ & 0.013 $[-0.057,0.077]$ & 0.5063 $[0.4012,0.6697]$ & 0.3901 $[0.3273,0.4871]$ & 0.8968 $[0.7415,1.1378]$ & 0.6251 & 0.3749 \\
       & 20 & 202 & 1.129 $[1.001,1.300]$ & 0.014 $[-0.105,0.119]$ & 0.6117 $[0.4456,0.8388]$ & 0.3726 $[0.2950,0.4541]$ & 0.9876 $[0.7608,1.2774]$ & 0.6227 & 0.3773 \\
      \midrule
      \multirow{3}{*}{S\&P500} & 1 & 4034 & 1.529 $[1.495,1.562]$ & 0.006 $[-0.019,0.036]$ & 0.6779 $[0.5743,0.8311]$ & 0.5624 $[0.4801,0.6832]$ & 1.2396 $[1.0669,1.5014]$ & 0.5497 & 0.4503 \\
       & 7 & 576 & 1.382 $[1.276,1.471]$ & -0.008 $[-0.076,0.061]$ & 0.6732 $[0.5109,0.8839]$ & 0.3903 $[0.3281,0.4670]$ & 1.0634 $[0.8561,1.3429]$ & 0.6254 & 0.3746 \\
       & 20 & 201 & 1.393 $[1.170,1.565]$ & -0.027 $[-0.130,0.076]$ & 0.8893 $[0.5735,2.2095]$ & 0.3888 $[0.2974,0.5475]$ & 1.2737 $[0.8855,2.7399]$ & 0.6343 & 0.3657 \\
      \midrule
      \multirow{3}{*}{BOVESPA} & 1 & 3968 & 1.338 $[1.295,1.375]$ & 0.007 $[-0.019,0.032]$ & 0.4813 $[0.4443,0.5208]$ & 0.4620 $[0.4288,0.4981]$ & 0.9430 $[0.8784,1.0108]$ & 0.5396 & 0.4604 \\
       & 7 & 566 & 1.194 $[1.001,1.305]$ & 0.010 $[-0.054,0.078]$ & 0.4893 $[0.4163,0.5875]$ & 0.4430 $[0.3856,0.5084]$ & 0.9323 $[0.8145,1.0861]$ & 0.5692 & 0.4308 \\
       & 20 & 198 & 1.185 $[1.026,1.330]$ & 0.003 $[-0.114,0.124]$ & 0.5345 $[0.4064,0.6889]$ & 0.4661 $[0.3766,0.5798]$ & 1.0040 $[0.7994,1.2364]$ & 0.5428 & 0.4572 \\
      \midrule
      \multirow{3}{*}{BRENT} & 1 & 9212 & 1.449 $[1.426,1.471]$ & 0.000 $[-0.017,0.017]$ & 0.5494 $[0.5092,0.5919]$ & 0.4797 $[0.4504,0.5139]$ & 1.0277 $[0.9666,1.1006]$ & 0.5490 & 0.4510 \\
       & 7 & 1316 & 1.355 $[1.285,1.410]$ & 0.006 $[-0.034,0.047]$ & 0.5415 $[0.4651,0.6317]$ & 0.4330 $[0.3885,0.4928]$ & 0.9748 $[0.8566,1.1143]$ & 0.5787 & 0.4213 \\
       & 20 & 460 & 1.331 $[1.183,1.443]$ & 0.007 $[-0.068,0.077]$ & 0.4936 $[0.3947,0.6560]$ & 0.4404 $[0.3629,0.5566]$ & 0.9400 $[0.7667,1.1767]$ & 0.5677 & 0.4323 \\
      \midrule
      \multirow{3}{*}{BTC} & 1 & 4152 & 1.616 $[1.586,1.648]$ & 0.022 $[-0.009,0.052]$ & 1.2480 $[0.8359,2.9660]$ & 1.3966 $[0.9372,3.4512]$ & 2.6321 $[1.7886,6.4509]$ & 0.5278 & 0.4722 \\
       & 7 & 593 & 1.403 $[1.303,1.498]$ & 0.014 $[-0.058,0.085]$ & 0.4984 $[0.4079,0.6441]$ & 0.5670 $[0.4574,0.7529]$ & 1.0714 $[0.8814,1.3869]$ & 0.4898 & 0.5102 \\
       & 20 & 207 & 1.305 $[1.049,1.449]$ & 0.009 $[-0.098,0.120]$ & 0.4840 $[0.3668,0.6696]$ & 0.5069 $[0.3832,0.7197]$ & 0.9941 $[0.7784,1.3361]$ & 0.5049 & 0.4951 \\
      \midrule
      \multirow{3}{*}{ETH} & 1 & 2877 & 1.543 $[1.504,1.586]$ & 0.015 $[-0.019,0.050]$ & 0.6867 $[0.5666,0.9467]$ & 0.7000 $[0.5640,0.9783]$ & 1.3812 $[1.1432,1.9079]$ & 0.5268 & 0.4732 \\
       & 7 & 411 & 1.509 $[1.394,1.600]$ & 0.038 $[-0.054,0.116]$ & 0.6044 $[0.4237,1.1336]$ & 0.7353 $[0.5072,1.4555]$ & 1.3416 $[0.9549,2.5933]$ & 0.5280 & 0.4720 \\
       & 20 & 143 & 1.364 $[1.088,1.575]$ & 0.030 $[-0.140,0.164]$ & 0.5279 $[0.3719,0.7880]$ & 0.5978 $[0.4102,1.1630]$ & 1.1385 $[0.8256,1.9026]$ & 0.4997 & 0.5003 \\
      \bottomrule
    \end{tabular}%
    }%
  }%
  \caption{Real markets, all horizons: fitted moments and sample semi-variances, bootstrap median $[p_5,p_{95}]$.}
  \label{tab:params_markets}
\end{table}

\begin{table}[!ht]
  \centering
  \adjustbox{max width=\textwidth}{\scriptsize
    \begin{tabular}{llr rrrrr}
      \toprule
      & & & \multicolumn{5}{c}{Asymmetric Tsallis ($f_\theta$)} \\
      \cmidrule(lr){4-8}
      Market & $\Delta$ & $n$ & $\widehat{q}$ & $\widehat{\mu}_q$ & $\widehat{\beta}$ & $\widehat{\gamma}_-$ & $\widehat{\gamma}_+$ \\
      \midrule
      \multirow{3}{*}{BIST} & 1 & 4017 & 1.387 $[1.353,1.420]$ & 0.029 $[0.007,0.052]$ & -0.064 $[-0.103,-0.029]$ & 0.4021 $[0.3340,0.4782]$ & 0.5210 $[0.4346,0.6165]$ \\
       & 7 & 573 & 1.191 $[1.077,1.282]$ & 0.023 $[-0.038,0.093]$ & -0.062 $[-0.119,-0.015]$ & 0.0900 $[0.0300,0.1792]$ & 0.1874 $[0.0634,0.3326]$ \\
       & 20 & 200 & 1.169 $[1.001,1.344]$ & 0.003 $[-0.111,0.116]$ & 0.000 $[-0.061,0.059]$ & 0.1110 $[0.0006,0.3595]$ & 0.1068 $[0.0004,0.3547]$ \\
      \midrule
      \multirow{3}{*}{FTSE} & 1 & 4041 & 1.472 $[1.441,1.502]$ & 0.030 $[0.009,0.050]$ & -0.096 $[-0.145,-0.048]$ & 0.6382 $[0.5359,0.7491]$ & 0.8416 $[0.7207,0.9667]$ \\
       & 7 & 577 & 1.353 $[1.252,1.434]$ & 0.039 $[-0.021,0.099]$ & -0.079 $[-0.168,-0.002]$ & 0.3190 $[0.1633,0.4915]$ & 0.4661 $[0.2773,0.6852]$ \\
       & 20 & 202 & 1.292 $[1.001,1.460]$ & 0.067 $[-0.041,0.162]$ & -0.136 $[-0.337,-0.001]$ & 0.1858 $[0.0004,0.5432]$ & 0.4805 $[0.0013,1.0635]$ \\
      \midrule
      \multirow{3}{*}{NASDAQ} & 1 & 4024 & 1.482 $[1.446,1.518]$ & 0.040 $[0.020,0.060]$ & -0.147 $[-0.196,-0.095]$ & 0.6180 $[0.5022,0.7547]$ & 0.9247 $[0.7968,1.0813]$ \\
       & 7 & 574 & 1.235 $[1.129,1.318]$ & 0.043 $[-0.019,0.109]$ & -0.106 $[-0.179,-0.047]$ & 0.1119 $[0.0461,0.1993]$ & 0.2793 $[0.1307,0.4648]$ \\
       & 20 & 201 & 1.179 $[1.001,1.350]$ & 0.066 $[-0.079,0.211]$ & -0.153 $[-0.354,-0.001]$ & 0.0371 $[0.0001,0.1496]$ & 0.2774 $[0.0011,0.7895]$ \\
      \midrule
      \multirow{3}{*}{DJIA} & 1 & 4048 & 1.531 $[1.496,1.565]$ & 0.029 $[0.011,0.048]$ & -0.085 $[-0.148,-0.022]$ & 1.0070 $[0.8276,1.2010]$ & 1.2040 $[1.0352,1.4048]$ \\
       & 7 & 578 & 1.358 $[1.217,1.457]$ & 0.056 $[-0.004,0.111]$ & -0.127 $[-0.234,-0.043]$ & 0.3210 $[0.1323,0.5852]$ & 0.5787 $[0.2797,0.9203]$ \\
       & 20 & 202 & 1.129 $[1.001,1.300]$ & 0.058 $[-0.079,0.179]$ & -0.091 $[-0.277,-0.001]$ & 0.0290 $[0.0000,0.1341]$ & 0.1749 $[0.0011,0.5718]$ \\
      \midrule
      \multirow{3}{*}{S\&P500} & 1 & 4034 & 1.529 $[1.495,1.562]$ & 0.036 $[0.017,0.055]$ & -0.109 $[-0.165,-0.048]$ & 0.9304 $[0.7629,1.1085]$ & 1.1850 $[1.0190,1.3795]$ \\
       & 7 & 576 & 1.382 $[1.276,1.471]$ & 0.086 $[0.025,0.145]$ & -0.273 $[-0.400,-0.166]$ & 0.2533 $[0.1200,0.4350]$ & 0.7836 $[0.4665,1.1881]$ \\
       & 20 & 201 & 1.393 $[1.170,1.565]$ & 0.126 $[-0.002,0.227]$ & -0.415 $[-0.881,-0.147]$ & 0.1704 $[0.0271,0.5430]$ & 0.9949 $[0.2624,2.6670]$ \\
      \midrule
      \multirow{3}{*}{BOVESPA} & 1 & 3968 & 1.338 $[1.295,1.375]$ & 0.014 $[-0.010,0.038]$ & -0.019 $[-0.049,0.010]$ & 0.3456 $[0.2798,0.4173]$ & 0.3830 $[0.3044,0.4572]$ \\
       & 7 & 566 & 1.194 $[1.001,1.305]$ & 0.023 $[-0.040,0.093]$ & -0.025 $[-0.085,0.002]$ & 0.1204 $[0.0007,0.2431]$ & 0.1729 $[0.0009,0.3360]$ \\
       & 20 & 198 & 1.185 $[1.026,1.330]$ & 0.020 $[-0.097,0.137]$ & -0.032 $[-0.124,0.021]$ & 0.1023 $[0.0070,0.2776]$ & 0.1554 $[0.0145,0.3897]$ \\
      \midrule
      \multirow{3}{*}{BRENT} & 1 & 9212 & 1.449 $[1.426,1.471]$ & 0.023 $[0.009,0.037]$ & -0.076 $[-0.105,-0.045]$ & 0.5990 $[0.5313,0.6719]$ & 0.7533 $[0.6743,0.8392]$ \\
       & 7 & 1316 & 1.355 $[1.285,1.410]$ & 0.045 $[0.008,0.084]$ & -0.109 $[-0.172,-0.053]$ & 0.2999 $[0.2075,0.4140]$ & 0.4992 $[0.3432,0.6679]$ \\
       & 20 & 460 & 1.331 $[1.183,1.443]$ & 0.027 $[-0.043,0.091]$ & -0.049 $[-0.152,0.029]$ & 0.3105 $[0.1299,0.5573]$ & 0.4042 $[0.1745,0.7045]$ \\
      \midrule
      \multirow{3}{*}{BTC} & 1 & 4152 & 1.616 $[1.586,1.648]$ & 0.005 $[-0.012,0.021]$ & 0.065 $[-0.019,0.144]$ & 1.6135 $[1.3741,1.9383]$ & 1.4576 $[1.2148,1.7262]$ \\
       & 7 & 593 & 1.403 $[1.303,1.498]$ & -0.010 $[-0.067,0.048]$ & 0.067 $[-0.033,0.167]$ & 0.5499 $[0.3292,0.8617]$ & 0.4239 $[0.2485,0.6797]$ \\
       & 20 & 207 & 1.305 $[1.049,1.449]$ & -0.000 $[-0.103,0.101]$ & 0.018 $[-0.074,0.169]$ & 0.3013 $[0.0265,0.7109]$ & 0.2589 $[0.0370,0.5414]$ \\
      \midrule
      \multirow{3}{*}{ETH} & 1 & 2877 & 1.543 $[1.504,1.586]$ & 0.012 $[-0.010,0.035]$ & 0.008 $[-0.069,0.086]$ & 1.0806 $[0.8554,1.3592]$ & 1.0537 $[0.8681,1.3018]$ \\
       & 7 & 411 & 1.509 $[1.394,1.600]$ & 0.000 $[-0.062,0.061]$ & 0.108 $[-0.066,0.280]$ & 0.9322 $[0.5010,1.5289]$ & 0.6970 $[0.3959,1.1568]$ \\
       & 20 & 143 & 1.364 $[1.088,1.575]$ & -0.005 $[-0.138,0.119]$ & 0.060 $[-0.082,0.356]$ & 0.4245 $[0.0396,1.4538]$ & 0.3191 $[0.0514,0.9200]$ \\
      \bottomrule
    \end{tabular}}
  \caption{Real markets, all horizons: structural estimates, bootstrap median $[p_5,p_{95}]$.}
  \label{tab:structural_markets}
\end{table}

\FloatBarrier

{\footnotesize
\bibliographystyle{apacite}
\bibliography{Bib_CCGR_TsallisEntropy}
}

\clearpage
\allowdisplaybreaks
\makeatletter
\let\section@cntformat\relax   % S1, S2, S3: no "Appendix" prefix
\makeatother
\setcounter{section}{0}
\setcounter{table}{0}
\setcounter{figure}{0}
\setcounter{equation}{0}
\renewcommand{\thesection}{S\arabic{section}}
\renewcommand{\thetable}{S\arabic{table}}
\renewcommand{\thefigure}{S\arabic{figure}}
\renewcommand{\theequation}{S\arabic{equation}}
\renewcommand{\theHsection}{supp.\arabic{section}}
\renewcommand{\theHtable}{supp.\arabic{table}}
\renewcommand{\theHfigure}{supp.\arabic{figure}}
\renewcommand{\theHequation}{supp.\arabic{equation}}

\begin{center}
  {\Large\bfseries Supplementary Material}
\end{center}

\noindent
This Supplementary Material collects the material announced in the main text: the proof of
Proposition~5 deferred from Section~3, the multistart robustness check of the
estimation procedure of Section~4.2, and the market-by-market graphical
diagnostics of Section~4.5. Section, table, figure and equation
numbers, and the numbers of the lemmas, theorems and corollaries of
Section~\ref{sec:proof}, are prefixed by \textsf{S}; references to unprefixed
numbers point to the main text.

\section{Proof of Proposition 5}
\label{sec:proof}

\noindent
Proposition~5 of the main text gives a sharp characterization of the existence
of Lagrange multipliers $(\beta,\gamma_-,\gamma_+)$ reproducing a prescribed
escort mean and a prescribed pair of escort semi-variances. The statement covers
existence, the sharp compatibility bounds on the admissible semi-variance ratio,
the characterization of the endpoints of the admissible range, and the
non-existence regime of the centered problem.

Throughout this section we write
\[
  m=\frac{q}{q-1},
  \qquad
  s_-=\sigma^2_{q,-},
  \qquad
  s_+=\sigma^2_{q,+},
\]
so that $-m$ is the exponent of the kernel $H$ in the stationary escort
density $\hat p$ of the main text and $s_\pm$ are the prescribed escort
semi-variances. The argument is a direct variational one: the multipliers are
recovered as the minimizer of an auxiliary convex function $\Phi$ of the three
kernel coefficients, and the compatibility bounds arise from the behaviour of
$\Phi$ on the two faces on which one quadratic coefficient vanishes.
Corollary~\ref{cor:complete-x-statement} at the end of the section is
Proposition~5 of the main text.

Fix
\[
    1<q<3,
    \qquad
    m=\frac{q}{q-1}>\frac32,
    \qquad
    s_+>0,
    \qquad
    s_->0,
\]
and use the notation
\[
    y_+=\max\{y,0\},
    \qquad
    y_-=\min\{y,0\}.
\]
The purpose of this note is to determine when parameters
\(\beta\in\R\), \(\gamma_+\geq0\), and \(\gamma_-\geq0\) can be
chosen so that the normalized density proportional to
\[
    \bigl[1+\beta y+\gamma_+y_+^2+\gamma_-y_-^2\bigr]^{-m}
\]
has mean zero and prescribed semi-variances \(s_+\) and \(s_-\).
Only ordinary finite-dimensional calculus is used below. Global
optimality and uniqueness for the entropy-maximization problem are
established in the main paper.

\subsection{Reduction to the centered problem}\label{sec:translation}

Let \(\mu_q\in\R\) be prescribed. Suppose first that a density \(g\)
on \(\R\) satisfies
\begin{equation}\label{eq:centered-moments}
    \int_{\R} y g(y)\,dy=0,
    \qquad
    \int_{\R} y_+^2 g(y)\,dy=s_+,
    \qquad
    \int_{\R} y_-^2 g(y)\,dy=s_-.
\end{equation}
Define
\begin{equation}\label{eq:translation-density}
    p_{\mu_q}(x)=g(x-\mu_q).
\end{equation}
With the change of variable \(y=x-\mu_q\),
\[
    \int_{\R}p_{\mu_q}(x)\,dx=1,
\]
\[
    \int_{\R}(x-\mu_q)p_{\mu_q}(x)\,dx
    =\int_{\R}y g(y)\,dy=0,
\]
and
\[
    \int_{\R}\bigl[(x-\mu_q)_\pm\bigr]^2p_{\mu_q}(x)\,dx
    =\int_{\R}y_\pm^2g(y)\,dy=s_\pm.
\]
Conversely, every density satisfying the constraints centered at
\(\mu_q\) gives a density satisfying \eqref{eq:centered-moments} by
setting \(g(y)=p_{\mu_q}(y+\mu_q)\).

Consequently, for every fixed \(\mu_q\in\R\), the existence and
uniqueness problem in the variable \(x\) is equivalent to the same
centered problem in the variable \(y=x-\mu_q\). The numerical value of
\(\mu_q\) creates no additional compatibility condition. We may
therefore work throughout with the centered variable \(y\).

\subsection{A reduced finite-dimensional parametrization}

Consider the homogeneous kernel
\[
    Q(y)=a+by+c_+y_+^2+c_-y_-^2.
\]
Multiplying \((a,b,c_+,c_-)\) by a common positive constant does not
change the normalized density proportional to \(Q^{-m}\). We fix this
scale by imposing
\begin{equation}\label{eq:gauge}
    a+c_+s_+ +c_-s_-=1.
\end{equation}
Eliminating \(a\) from \eqref{eq:gauge}, set
\begin{equation}\label{eq:a-definition}
    a=1-c_+s_+-c_-s_-
\end{equation}
and define the kernel, now as an ordinary function of three real
parameters,
\begin{align}
    Q_{b,c_+,c_-}(y)
    & =a+by+c_+y_+^2+c_-y_-^2 \notag\\
    & =1+by+c_+(y_+^2-s_+)+c_-(y_-^2-s_-).
    \label{eq:reduced-kernel}
\end{align}
For brevity, write \(\theta=(b,c_+,c_-)\), \(a(\theta)
=1-c_+s_+-c_-s_-\), and \(Q_\theta=Q_{b,c_+,c_-}\).

Let
\begin{equation}\label{eq:closed-domain}
    \overline{\mathcal A}
    =\bigl\{\theta:\ c_+\geq0,\ c_-\geq0,
       \ Q_\theta(y)\geq0\ \text{for every }y\in\R\bigr\}.
\end{equation}
For \(\theta\in\overline{\mathcal A}\), define the extended-valued
function
\begin{equation}\label{eq:Phi}
    \Phi(\theta)=\int_{\R}Q_\theta(y)^{1-m}\,dy,
\end{equation}
with \(\Phi(\theta)=+\infty\) whenever the integral is not finite.
If \(a>0\) and \(c_+,c_->0\), positivity of the kernel is equivalent
to
\begin{equation}\label{eq:positivity-homogeneous}
    -2\sqrt{ac_+}<b<2\sqrt{ac_-}.
\end{equation}
The corresponding non-strict inequalities hold on
\(\overline{\mathcal A}\).

\begin{slemma}[Compactness and boundary singularities]
\label{lem:compactness}
The set \(\overline{\mathcal A}\) is compact and convex, and \(\Phi\)
attains a finite minimum on it. Moreover, if \(Q_\theta\) vanishes at
a finite point, then \(\Phi(\theta)=+\infty\).
\end{slemma}

\begin{proof}
Since \(Q_\theta(0)=a\), non-negativity of the kernel gives
\(a\geq0\). Together with \eqref{eq:a-definition}, this yields
\[
    0\leq c_+\leq\frac1{s_+},
    \qquad
    0\leq c_-\leq\frac1{s_-},
    \qquad
    0\leq a\leq1.
\]
The non-strict form of \eqref{eq:positivity-homogeneous} bounds \(b\).
Thus \(\overline{\mathcal A}\) is bounded. It is closed by its
definition and convex because each condition
\(Q_\theta(y)\geq0\), for fixed \(y\), is an affine inequality in
\(\theta\). Hence it is compact and convex.

There are points with finite value: choose \(b=0\) and
\(c_+=c_-=\varepsilon>0\) with
\(\varepsilon(s_++s_-)<1\). Then
\(Q_\theta(y)=a+\varepsilon y^2\) with \(a>0\), and
\(\Phi(\theta)<+\infty\) because \(m>3/2\).

If \(\theta_n\to\theta\), then \(Q_{\theta_n}(y)\to Q_\theta(y)\)
for every \(y\). Fatou's lemma shows that \(\Phi\) is lower
semicontinuous on \(\overline{\mathcal A}\); therefore it attains a
finite minimum on this compact set.

Finally, if a non-negative piecewise quadratic kernel vanishes at a
finite point \(y_0\neq0\), then either it vanishes identically on the
relevant half-line, in which case the integral is immediately infinite,
or the zero is a double zero and locally
\(Q_\theta(y)\asymp (y-y_0)^2\). If it vanishes at \(y_0=0\), then
\(a=0\), and non-negativity on both sides forces \(b=0\); the kernel
is then quadratic (or identically zero) on each adjacent half-line.
In every case, \(Q_\theta^{1-m}\) has a non-integrable local
singularity because \(m>3/2\). Hence \(\Phi(\theta)=+\infty\).
\end{proof}

\begin{slemma}[Ordinary partial derivatives]
\label{lem:derivatives}
At every non-degenerate interior point of
\(\overline{\mathcal A}\), that is, at every point with
\(a>0\), \(c_+>0\), \(c_->0\), and strict inequalities in
\eqref{eq:positivity-homogeneous}, the function \(\Phi\) is
continuously differentiable and
\begin{align}
    \frac{\partial\Phi}{\partial b}
    &=(1-m)\int_{\R}yQ_\theta(y)^{-m}\,dy,
    \label{eq:partial-b}\\
    \frac{\partial\Phi}{\partial c_+}
    &=(1-m)\int_{\R}(y_+^2-s_+)Q_\theta(y)^{-m}\,dy,
    \label{eq:partial-cplus}\\
    \frac{\partial\Phi}{\partial c_-}
    &=(1-m)\int_{\R}(y_-^2-s_-)Q_\theta(y)^{-m}\,dy.
    \label{eq:partial-cminus}
\end{align}
Let
\begin{equation}\label{eq:centered-density}
    Z_\theta=\int_{\R}Q_\theta(y)^{-m}\,dy,
    \qquad
    p_\theta(y)=\frac{Q_\theta(y)^{-m}}{Z_\theta}.
\end{equation}
Then \eqref{eq:partial-b}--\eqref{eq:partial-cminus} can be written as
\begin{align}
    \frac{\partial\Phi}{\partial b}
    &=(1-m)Z_\theta\E_{p_\theta}[y],
    \label{eq:partial-b-expectation}\\
    \frac{\partial\Phi}{\partial c_+}
    &=(1-m)Z_\theta
      \bigl(\E_{p_\theta}[y_+^2]-s_+\bigr),
    \label{eq:partial-cplus-expectation}\\
    \frac{\partial\Phi}{\partial c_-}
    &=(1-m)Z_\theta
      \bigl(\E_{p_\theta}[y_-^2]-s_-\bigr).
    \label{eq:partial-cminus-expectation}
\end{align}
\end{slemma}

\begin{proof}
On a sufficiently small closed neighborhood of an interior point,
there is a constant \(C>0\) such that
\[
    Q_\theta(y)\geq C(1+y^2)
    \qquad (y\in\R).
\]
The derivatives of the integrand with respect to \(b,c_+,c_-\) are
therefore dominated by an integrable multiple of
\((1+y^2)^{1-m}\). The standard theorem for differentiating an
integral depending on finitely many parameters gives
\eqref{eq:partial-b}--\eqref{eq:partial-cminus}. Dividing by
\(Z_\theta\) gives
\eqref{eq:partial-b-expectation}--\eqref{eq:partial-cminus-expectation}.
\end{proof}

\begin{slemma}[Strict convexity and the first-order equations]
\label{lem:convexity}
The function \(\Phi\) is strictly convex on every convex subset of
its finite domain. Consequently, if an interior point \(\theta^*\)
satisfies
\begin{equation}\label{eq:three-first-order-conditions}
    \frac{\partial\Phi}{\partial b}(\theta^*)=0,
    \qquad
    \frac{\partial\Phi}{\partial c_+}(\theta^*)=0,
    \qquad
    \frac{\partial\Phi}{\partial c_-}(\theta^*)=0,
\end{equation}
then it is the unique global minimizer of \(\Phi\). The density
\(p_{\theta^*}\) satisfies exactly the three equations in
\eqref{eq:centered-moments}.
\end{slemma}

\begin{proof}
The map \(t\mapsto t^{1-m}\) is strictly convex on \((0,+\infty)\),
and \(Q_\theta(y)\) is affine in \(\theta\). If two parameter triples
produced the same kernel almost everywhere, their difference would be
a piecewise quadratic polynomial that vanishes almost everywhere;
all three coefficients would therefore be zero. Integration thus
preserves strict convexity.

Let \(\theta\neq\theta^*\) be any other point with finite value and
set
\[
    \varphi(t)=\Phi\bigl((1-t)\theta^*+t\theta\bigr),
    \qquad 0\leq t\leq1.
\]
The function \(\varphi\) is strictly convex. By the ordinary chain
rule and \eqref{eq:three-first-order-conditions},
\(\varphi'(0)=0\). Hence \(\varphi(1)>\varphi(0)\), which proves
unique global minimality. The final assertion follows immediately
from
\eqref{eq:partial-b-expectation}--\eqref{eq:partial-cminus-expectation}.
\end{proof}

\subsection{The full-ratio regime}

\begin{stheorem}[Existence for arbitrary positive semi-variances]
\label{thm:full-ratio}
Let
\[
    \frac32\leq q<3,
    \qquad
    m=\frac{q}{q-1},
\]
so that \(3/2<m\leq3\). For every \(s_+>0\) and \(s_->0\), there
exist unique parameters
\[
    \beta\in\R,
    \qquad
    \gamma_+>0,
    \qquad
    \gamma_->0,
\]
such that
\begin{equation}\label{eq:positive-normalized-kernel}
    -2\sqrt{\gamma_+}<\beta<2\sqrt{\gamma_-},
\end{equation}
and the associated centered escort density satisfies
\[
    \E[y]=0,
    \qquad
    \E[y_+^2]=s_+,
    \qquad
    \E[y_-^2]=s_-.
\]
\end{stheorem}

\begin{proof}
By Lemma~\ref{lem:compactness}, \(\Phi\) has a finite minimizer
\(\theta^*=(b^*,c_+^*,c_-^*)\) on
\(\overline{\mathcal A}\), and its kernel cannot vanish at a finite
point.

Suppose that \(c_+^*=0\). If \(\Phi(\theta^*)\) is finite, positivity
of the positive tail requires \(b^*>0\). For sufficiently small
\(t>0\), the point
\[
    \theta_t=(b^*,t,c_-^*)
\]
is admissible and
\[
    Q_{\theta_t}(y)=Q_{\theta^*}(y)+t(y_+^2-s_+).
\]
Differentiating first on a bounded interval and then letting the
interval expand to \(\R\) gives the inward right derivative
\begin{equation}\label{eq:full-regime-boundary-derivative}
    \left.\frac{d}{dt}\Phi(\theta_t)\right|_{t=0+}
    =(1-m)\int_{\R}(y_+^2-s_+)Q_{\theta^*}(y)^{-m}\,dy.
\end{equation}
On the positive half-line, \(Q_{\theta^*}(y)\sim b^*y\), so the
leading term in the integral is proportional to
\[
    \int^\infty y^{2-m}\,dy,
\]
which diverges for \(m\leq3\). Since \(1-m<0\), the derivative in
\eqref{eq:full-regime-boundary-derivative} equals \(-\infty\). Thus a
small positive value of \(c_+\) decreases \(\Phi\), contradicting
minimality. If the value on that face is infinite, it cannot be the
finite minimum in the first place. Hence \(c_+^*>0\). Reflection gives
\(c_-^*>0\).

The remaining boundary possibilities would make the kernel vanish at
a finite point and are excluded by Lemma~\ref{lem:compactness}.
Therefore \(\theta^*\) is a non-degenerate interior point. The three
ordinary first-order conditions
\eqref{eq:three-first-order-conditions} hold, and
Lemma~\ref{lem:convexity} shows that the density
\(p_{\theta^*}\) has the prescribed centered moments.

Let
\[
    a^*=1-c_+^*s_+-c_-^*s_->0
\]
and set
\begin{equation}\label{eq:normalized-parameters}
    \beta=\frac{b^*}{a^*},
    \qquad
    \gamma_+=\frac{c_+^*}{a^*},
    \qquad
    \gamma_-=\frac{c_-^*}{a^*}.
\end{equation}
Dividing the kernel by \(a^*\) leaves the normalized density
unchanged, and the strict positivity inequalities become
\eqref{eq:positive-normalized-kernel}. Strict convexity gives
uniqueness on the gauge section, and therefore uniqueness of
\((\beta,\gamma_+,\gamma_-)\).
\end{proof}

\begin{sremark}[Finite ordinary variance]
Theorem~\ref{thm:full-ratio} concerns the escort density and its escort
semi-variances. In the non-degenerate case, finite ordinary variance
of the corresponding Tsallis density requires the additional
restriction
\[
    \frac32\leq q<\frac53.
\]
\end{sremark}

\subsection{The lower-\texorpdfstring{$q$}{q} regime}

Assume now
\[
    1<q<\frac32,
    \qquad
    m=\frac{q}{q-1}>3,
    \qquad
    \rho=\frac{s_+}{s_-}.
\]
In this range, a linearly confining tail has a finite second moment
under the escort density, so a minimizer may lie on a face on which
one quadratic coefficient vanishes.

\begin{slemma}[Boundary ratios and one-sided partial derivatives]
\label{lem:boundary-ratios}
Let \(m>3\). On the face \(c_+=0\), the function \(\Phi\) has a
unique relative minimizer
\[
    \bar\theta_+=(\bar b,0,\bar c_-).
\]
Writing \(\bar p_+\) for its normalized density, one has
\begin{equation}\label{eq:face-plus-moments}
    \E_{\bar p_+}[y]=0,
    \qquad
    \E_{\bar p_+}[y_-^2]=s_-,
\end{equation}
and
\begin{equation}\label{eq:face-plus-ratio}
    \frac{\E_{\bar p_+}[y_+^2]}
         {\E_{\bar p_+}[y_-^2]}
    =\frac{2m-3}{m-3}.
\end{equation}
The inward one-sided partial derivative with respect to \(c_+\) is
\begin{equation}\label{eq:face-plus-derivative}
    \frac{\partial^+\Phi}{\partial c_+}(\bar\theta_+)
    =(1-m)\bar Z_+
      \left(\frac{2m-3}{m-3}s_- -s_+\right),
\end{equation}
where
\(\bar Z_+=\int_{\R}Q_{\bar\theta_+}(y)^{-m}\,dy\).

On the reflected face \(c_-=0\), the unique relative minimizer
\(\bar\theta_-=(\bar b_-,\bar c_+,0)\) satisfies
\begin{equation}\label{eq:face-minus-moments}
    \E_{\bar p_-}[y]=0,
    \qquad
    \E_{\bar p_-}[y_+^2]=s_+,
\end{equation}
\begin{equation}\label{eq:face-minus-ratio}
    \frac{\E_{\bar p_-}[y_+^2]}
         {\E_{\bar p_-}[y_-^2]}
    =\frac{m-3}{2m-3},
\end{equation}
and
\begin{equation}\label{eq:face-minus-derivative}
    \frac{\partial^+\Phi}{\partial c_-}(\bar\theta_-)
    =(1-m)\bar Z_-
      \left(\frac{2m-3}{m-3}s_+ -s_-\right).
\end{equation}
\end{slemma}

\begin{proof}
It is enough to consider the face \(c_+=0\); the other face follows
from the reflection \(y\mapsto-y\) together with the interchange of
\(+\) and \(-\).

On \(c_+=0\), the kernel is
\[
    Q(y)=a+by \quad (y>0),
    \qquad
    Q(y)=a+by+c_-y^2 \quad (y<0),
\]
with \(a=1-c_-s_-\). A finite positive tail requires \(b>0\).
Compactness, strict convexity, and the divergence at kernels with a
finite zero give a unique relative minimizer in the interior of this
face. At that point the ordinary partial derivatives of the restricted
two-variable function \( (b,c_-)\mapsto\Phi(b,0,c_-) \) vanish.
Since \(m>3\), differentiating under the integral sign is justified on
this face and gives
\[
    \frac{\partial\Phi}{\partial b}
    =(1-m)\int_{\R}yQ^{-m}\,dy,
    \qquad
    \frac{\partial\Phi}{\partial c_-}
    =(1-m)\int_{\R}(y_-^2-s_-)Q^{-m}\,dy.
\]
Their vanishing gives \eqref{eq:face-plus-moments}.

For the rest of the proof, suppress the bars and let
\(Z=\int_{\R}Q^{-m}\,dy\). Integrating the ordinary derivative
\[
    \frac{d}{dy}\bigl(yQ(y)^{1-m}\bigr)
\]
over the two half-lines, using the vanishing boundary terms and
\(\E[y]=0\), yields
\begin{equation}\label{eq:integration-identity-one}
    \int_{\R}Q^{1-m}\,dy
    =2(m-1)c_-\int_{\R}y_-^2Q^{-m}\,dy.
\end{equation}
On the other hand, the gauge and the two moment equations give
\begin{align}
    \int_{\R}Q^{1-m}\,dy
    &=\int_{\R}QQ^{-m}\,dy \notag\\
    &=Z(a+c_-s_-)=Z.
    \label{eq:integration-identity-two}
\end{align}
Combining \eqref{eq:integration-identity-one} and
\eqref{eq:integration-identity-two} gives
\begin{equation}\label{eq:a-cminus-values}
    c_-s_-=\frac{1}{2(m-1)},
    \qquad
    a=\frac{2m-3}{2(m-1)}.
\end{equation}

On the positive half-line, define
\[
    I_1^+=\int_0^\infty y(a+by)^{-m}\,dy,
    \qquad
    I_2^+=\int_0^\infty y^2(a+by)^{-m}\,dy.
\]
Elementary beta integrals give
\[
    I_1^+=\frac{a^{2-m}}{b^2(m-1)(m-2)},
\]
\[
    I_2^+=
    \frac{2a^{3-m}}{b^3(m-1)(m-2)(m-3)},
\]
and hence
\begin{equation}\label{eq:I2-I1-ratio}
    \frac{I_2^+}{I_1^+}=\frac{2a}{b(m-3)}.
\end{equation}
Integrating the ordinary derivative of \(Q^{1-m}\) over the two
half-lines gives
\[
    bZ+2c_-\int_{-\infty}^0 yQ^{-m}\,dy=0.
\]
The centering equation therefore implies
\begin{equation}\label{eq:bZ-I1}
    bZ=2c_-I_1^+.
\end{equation}
Using \eqref{eq:I2-I1-ratio}, \eqref{eq:bZ-I1}, and
\eqref{eq:a-cminus-values}, we obtain
\[
    \E[y_+^2]
    =\frac{I_2^+}{Z}
    =\frac{a}{c_-(m-3)}
    =\frac{2m-3}{m-3}s_-.
\]
This proves \eqref{eq:face-plus-ratio}.

Because \(m>3\), the function
\((y_+^2-s_+)Q^{-m}\) is integrable even on the linear tail. The
ordinary inward one-sided partial derivative is therefore
\[
    \frac{\partial^+\Phi}{\partial c_+}
    =(1-m)\int_{\R}(y_+^2-s_+)Q^{-m}\,dy,
\]
which, together with \eqref{eq:face-plus-ratio}, gives
\eqref{eq:face-plus-derivative}. Reflection proves
\eqref{eq:face-minus-moments}--\eqref{eq:face-minus-derivative}.
\end{proof}

\begin{stheorem}[Sharp compatibility for \(1<q<3/2\)]
\label{thm:lower-q}
Let
\[
    1<q<\frac32,
    \qquad
    m=\frac{q}{q-1}>3.
\]
Non-degenerate parameters
\[
    \beta\in\R,
    \qquad
    \gamma_+>0,
    \qquad
    \gamma_->0
\]
exist if and only if
\begin{equation}\label{eq:compatibility-m}
    \frac{m-3}{2m-3}
    <\frac{s_+}{s_-}
    <\frac{2m-3}{m-3}.
\end{equation}
Equivalently,
\begin{equation}\label{eq:compatibility-q}
    \frac{3-2q}{3-q}
    <\frac{s_+}{s_-}
    <\frac{3-q}{3-2q}.
\end{equation}
Whenever they exist, the parameters and the normalized centered
density are unique.
\end{stheorem}

\begin{proof}
Assume first that \eqref{eq:compatibility-m} holds. Let \(\theta^*\)
be a minimizer of \(\Phi\) on \(\overline{\mathcal A}\). Its kernel
cannot vanish at a finite point. If a minimizer lay on \(c_+=0\), it
would have to be the unique relative minimizer \(\bar\theta_+\) of
Lemma~\ref{lem:boundary-ratios}. The upper inequality in
\eqref{eq:compatibility-m} makes
\eqref{eq:face-plus-derivative} strictly negative, so increasing
\(c_+\) slightly would decrease \(\Phi\). Hence no minimizer lies on
that face. By reflection, the lower inequality excludes the face
\(c_-=0\). Therefore the minimizer is non-degenerate and interior.
The three ordinary first-order equations give the prescribed moments
by Lemma~\ref{lem:convexity}, and normalization by
\eqref{eq:normalized-parameters} produces
\((\beta,\gamma_+,\gamma_-)\). Strict convexity gives uniqueness.

Conversely, suppose that non-degenerate parameters satisfying the
three prescribed moments exist. Their gauge-normalized representative
\(\theta^\circ\) is an interior point at which all three partial
derivatives vanish. By Lemma~\ref{lem:convexity}, it is the unique
global minimizer of \(\Phi\).

Assume, toward a contradiction, that
\begin{equation}\label{eq:upper-bound-failure}
    \frac{s_+}{s_-}\geq\frac{2m-3}{m-3}.
\end{equation}
Then \eqref{eq:face-plus-derivative} is non-negative at
\(\bar\theta_+\). Let \(\theta=(b,c_+,c_-)\) be any admissible point
with finite value, and consider the ordinary one-variable function
\[
    \varphi(t)
    =\Phi\bigl((1-t)\bar\theta_+ +t\theta\bigr),
    \qquad 0\leq t\leq1.
\]
At the relative minimizer on the face,
\[
    \frac{\partial\Phi}{\partial b}(\bar\theta_+)=0,
    \qquad
    \frac{\partial\Phi}{\partial c_-}(\bar\theta_+)=0.
\]
The ordinary chain rule and the one-sided derivative in the
\(c_+\)-coordinate therefore give
\begin{align*}
    \varphi'(0+)
    &=(b-\bar b)
      \frac{\partial\Phi}{\partial b}(\bar\theta_+)
      +c_+\frac{\partial^+\Phi}{\partial c_+}(\bar\theta_+)\\
    &\quad +(c_--\bar c_-)
      \frac{\partial\Phi}{\partial c_-}(\bar\theta_+)\\
    &=c_+\frac{\partial^+\Phi}{\partial c_+}(\bar\theta_+)\geq0.
\end{align*}
Since \(\varphi\) is convex, \(\varphi(1)\geq\varphi(0)\). Thus
\(\bar\theta_+\) is a global minimizer of \(\Phi\), contradicting the
existence of the distinct interior minimizer \(\theta^\circ\). Hence
the upper inequality in \eqref{eq:compatibility-m} must be strict.
Reflection gives the strict lower inequality. Substituting
\(m=q/(q-1)\) yields \eqref{eq:compatibility-q}.
\end{proof}

\begin{sremark}[Endpoint solutions]
\label{rem:endpoints}
At
\[
    \frac{s_+}{s_-}=\frac{2m-3}{m-3},
\]
the unique global minimizer is the boundary solution on \(c_+=0\),
which yields \(\gamma_+=0\) and \(\gamma_->0\). At
\[
    \frac{s_+}{s_-}=\frac{m-3}{2m-3},
\]
the reflected solution has \(\gamma_-=0\) and \(\gamma_+>0\). Outside
the corresponding closed interval, no centered kernel of the stated
linear--quadratic form with \(\gamma_+,\gamma_-\geq0\) satisfies both
prescribed semi-variance constraints.
\end{sremark}

\begin{sremark}[Ordinary variance at a degenerate endpoint]
At an endpoint, one tail of the corresponding Tsallis density is
linearly rather than quadratically confined. On that tail,
\(f(y)\asymp |y|^{-1/(q-1)}\), so its ordinary second moment is finite
only when \(q<4/3\). This extra qualification concerns only the
degenerate endpoint solutions; it does not affect the non-degenerate
case \(\gamma_+,\gamma_->0\).
\end{sremark}

\subsection{Complete statement for a prescribed \texorpdfstring{$\mu_q$}{muq}}

Combining Subsection~\ref{sec:translation} with
Theorems~\ref{thm:full-ratio} and~\ref{thm:lower-q} gives the form used
in the main paper.

\begin{scorollary}\label{cor:complete-x-statement}
Fix \(\mu_q\in\R\), \(1<q<3\), and prescribed values
\(s_+,s_->0\). A non-degenerate triple
\[
    \beta\in\R,
    \qquad
    \gamma_+>0,
    \qquad
    \gamma_->0,
    \qquad
    -2\sqrt{\gamma_+}<\beta<2\sqrt{\gamma_-},
\]
exists such that
\begin{equation}\label{eq:escort-density-x}
    \widehat p(x)
    =\frac{
      \bigl[1+\beta(x-\mu_q)
      +\gamma_+((x-\mu_q)_+)^2
      +\gamma_-((x-\mu_q)_-)^2\bigr]^{-m}}
    {\displaystyle
      \int_{\R}
      \bigl[1+\beta(u-\mu_q)
      +\gamma_+((u-\mu_q)_+)^2
      +\gamma_-((u-\mu_q)_-)^2\bigr]^{-m}\,du}
\end{equation}
satisfies
\[
    \int_{\R}(x-\mu_q)\widehat p(x)\,dx=0,
    \qquad
    \int_{\R}[(x-\mu_q)_\pm]^2\widehat p(x)\,dx=s_\pm,
\]
precisely in the following cases:
\begin{enumerate}
    \item if \(3/2\leq q<3\), for every \(s_+,s_->0\);
    \item if \(1<q<3/2\), if and only if
    \[
        \frac{3-2q}{3-q}
        <\frac{s_+}{s_-}
        <\frac{3-q}{3-2q}.
    \]
\end{enumerate}
The endpoint ratios in the second case give the unique degenerate
solutions described in Remark~\ref{rem:endpoints}. The existence
criterion is independent of the prescribed value of \(\mu_q\), and
whenever a solution exists, the parameter triple and the normalized
density \(\widehat p\) are unique.
\end{scorollary}

\begin{proof}
Apply the centered results to \(y=x-\mu_q\), and then use the
translation \eqref{eq:translation-density}. The change of variables
leaves the normalization and all three centered moments unchanged.
Uniqueness is preserved by this bijective translation.
\end{proof}

\section{Warm start versus multistart}
\label{sec:multistart}

The estimation procedure of the main text uses a single starting point,
obtained by first solving the two-parameter symmetric restriction
($\beta=0$, $\gamma_-=\gamma_+$) and then using its solution as the initial
point of the full five-parameter problem. Since the likelihood is not concave,
this raises the question of whether the warm start can trap the optimiser in a
local maximum, and in particular whether it can spuriously return
$\widehat\beta\approx0$ in a series that is in fact asymmetric.

The check is run on the least favourable configuration in the sample. BOVESPA
is the series whose fitted asymmetry is the weakest at every horizon
($|\widehat\beta|\le0.045$, $\Delta\ell\le0.82$, never significant), so it is
the one for which a warm start biased towards the symmetric solution would be
hardest to detect. The full five-parameter problem is re-solved from $74$
dispersed asymmetric starting points, with $\beta$ of both signs and
$\gamma_-\neq\gamma_+$, and the best multistart solution is compared with the
warm-start one.

Table~\ref{tab:bovespa_multistart} reports, at each horizon $\Delta$ and for
each of the two procedures, the sample size $n$, the entropic index
$\widehat q$, the asymmetry parameter $\widehat\beta$, the log-likelihood
difference $\Delta\ell$ against the symmetric Tsallis benchmark and the
corresponding likelihood-ratio statistic $\mathrm{LR}=2\Delta\ell$; the last
column is the log-likelihood gain $\ell_{MS}-\ell_{W}$ of the multistart
solution over the warm-start one.

\begin{table}[h]
  \centering
  \adjustbox{max width=\textwidth}{\scriptsize
    \begin{tabular}{rr rrr r rrr r r}
      \toprule
      & & \multicolumn{4}{c}{Warm-only (paper)} & \multicolumn{4}{c}{Multistart} & \\
      \cmidrule(lr){3-6}\cmidrule(lr){7-10}
      $\Delta$ & $n$ & $\widehat{q}$ & $\widehat{\beta}$ & $\Delta\ell$ & LR & $\widehat{q}$ & $\widehat{\beta}$ & $\Delta\ell$ & LR & $\ell_{MS}-\ell_{W}$ \\
      \midrule
      1 & 3968 & 1.337 & $-0.0205$ & 0.65 & $1.29$ & 1.337 & $-0.0205$ & 0.65 & $1.29$ & 2.60e-10 \\
      7 & 566 & 1.207 & $-0.0330$ & 0.82 & $1.63$ & 1.207 & $-0.0330$ & 0.82 & $1.63$ & 4.68e-08 \\
      20 & 198 & 1.206 & $-0.0446$ & 0.55 & $1.11$ & 1.206 & $-0.0446$ & 0.55 & $1.11$ & 9.93e-09 \\
      \bottomrule
    \end{tabular}}
  \caption{BOVESPA, warm start versus multistart: asymmetric Tsallis fit on standardised logarithmic returns.}
  \label{tab:bovespa_multistart}
\end{table}

The two solutions coincide at every horizon: the largest log-likelihood
improvement over the three frequencies is $4.7\times10^{-8}$, that is,
optimizer tolerance, and all $74$ asymmetric starts converge to the same
solution as the warm start. The small $\widehat\beta$ of BOVESPA, whose
bootstrap intervals cover zero at every horizon, is therefore a property of
the data and not an artefact of the initialisation.
The same check was repeated within the bootstrap ($6$ starting points for each
of the $999$ resamples) with no failures and no improvement over the warm
start.

\section{Market-by-market graphical diagnostics}
\label{sec:market_figures}

Figure~2 of the main text reports the two configurations with the strongest
evidence of asymmetry. This section reports the corresponding diagnostics for
all the market--horizon combinations, so that the cases in which the two
models are indistinguishable can be inspected alongside those in which they
are not. All the series are the standardised non-overlapping logarithmic
returns described in Section~4.1 of the main text, and all the fits are those
underlying the tables of Appendix~B. Each figure refers to a single market at
a single horizon and pairs the two diagnostics: the left panel reports the
binned empirical density on a semi-logarithmic scale, with the fitted
asymmetric Tsallis density as a solid curve and the symmetric Tsallis
benchmark as a dashed one; the right panel reports the QQ-plots of the same
two fitted models against the empirical quantiles, together with the
$45^{\circ}$ reference line.

\subsection{Horizon $\Delta=1$}

\begin{figure}[!h]
  \centering
  \includegraphics[width=\textwidth]{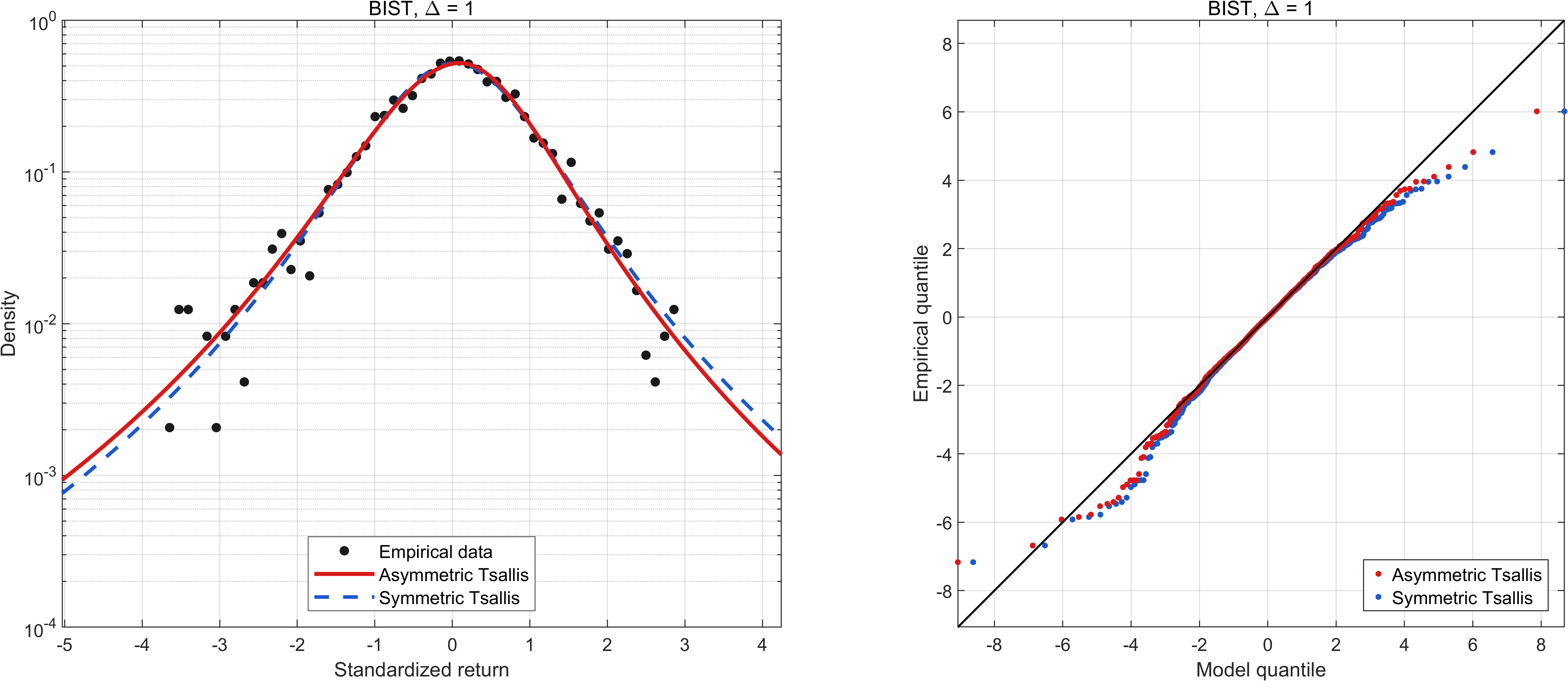}
  \caption{BIST, $\Delta=1$.}
  \label{fig:pdfqq_BIST_d1}
\end{figure}

\begin{figure}[!h]
  \centering
  \includegraphics[width=\textwidth]{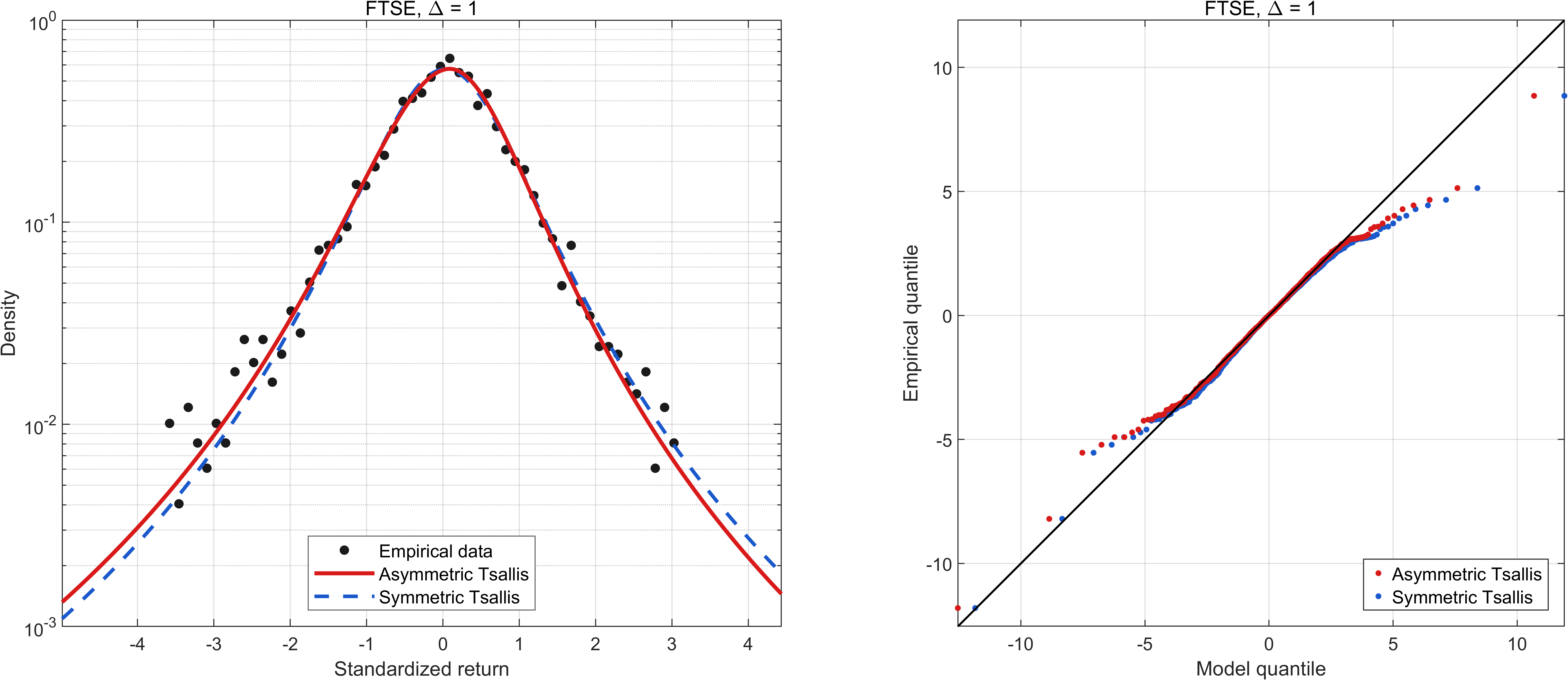}
  \caption{FTSE, $\Delta=1$.}
  \label{fig:pdfqq_FTSE_d1}
\end{figure}

\begin{figure}[!h]
  \centering
  \includegraphics[width=\textwidth]{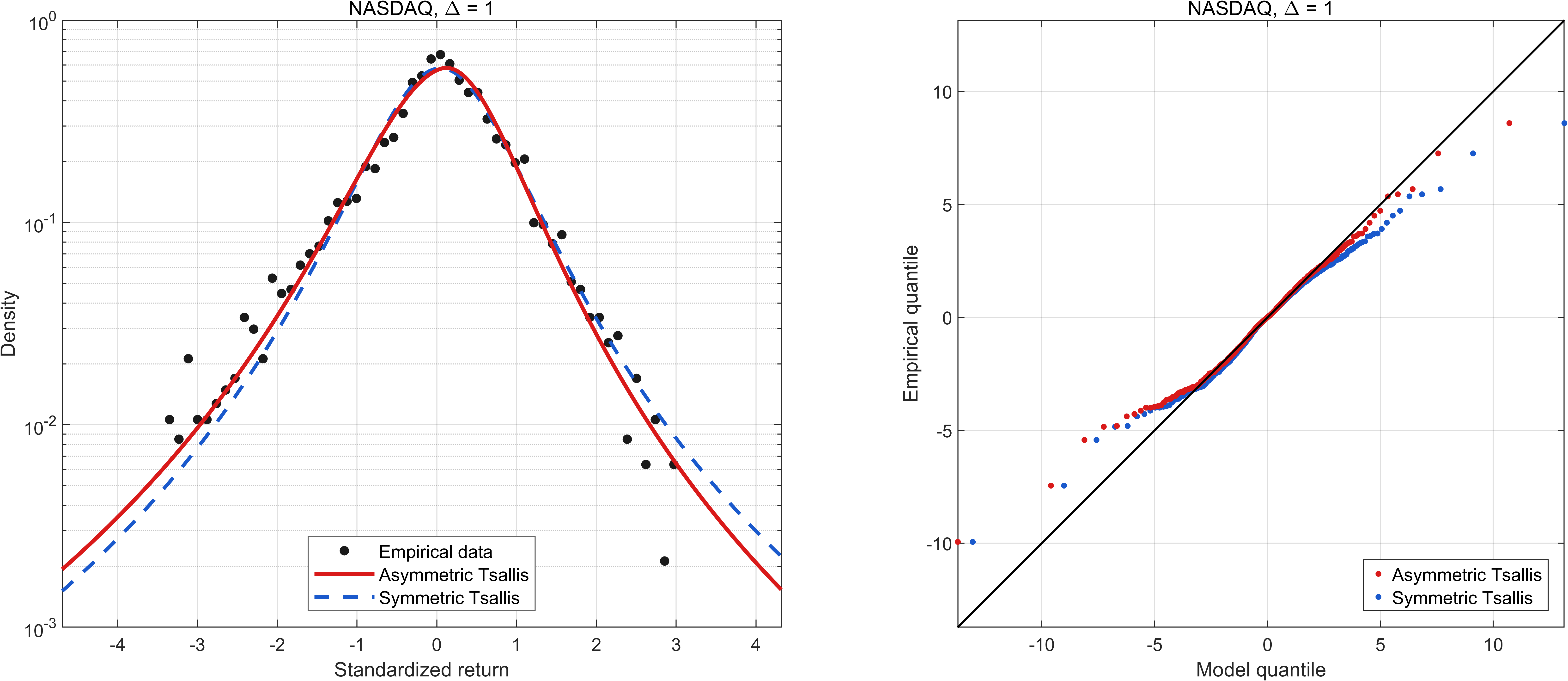}
  \caption{NASDAQ, $\Delta=1$.}
  \label{fig:pdfqq_NASDAQ_d1}
\end{figure}

\begin{figure}[!h]
  \centering
  \includegraphics[width=\textwidth]{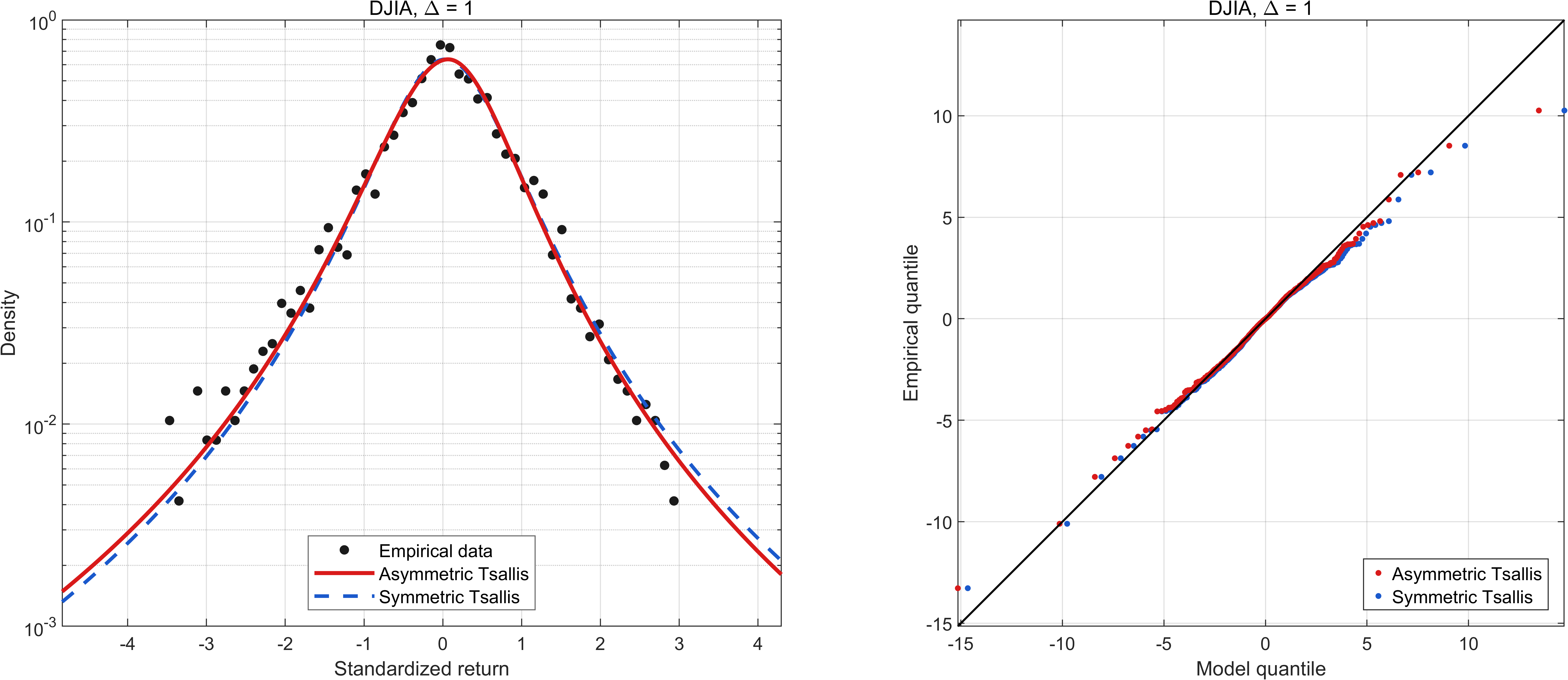}
  \caption{DJIA, $\Delta=1$.}
  \label{fig:pdfqq_DOWJONES_d1}
\end{figure}

\begin{figure}[!h]
  \centering
  \includegraphics[width=\textwidth]{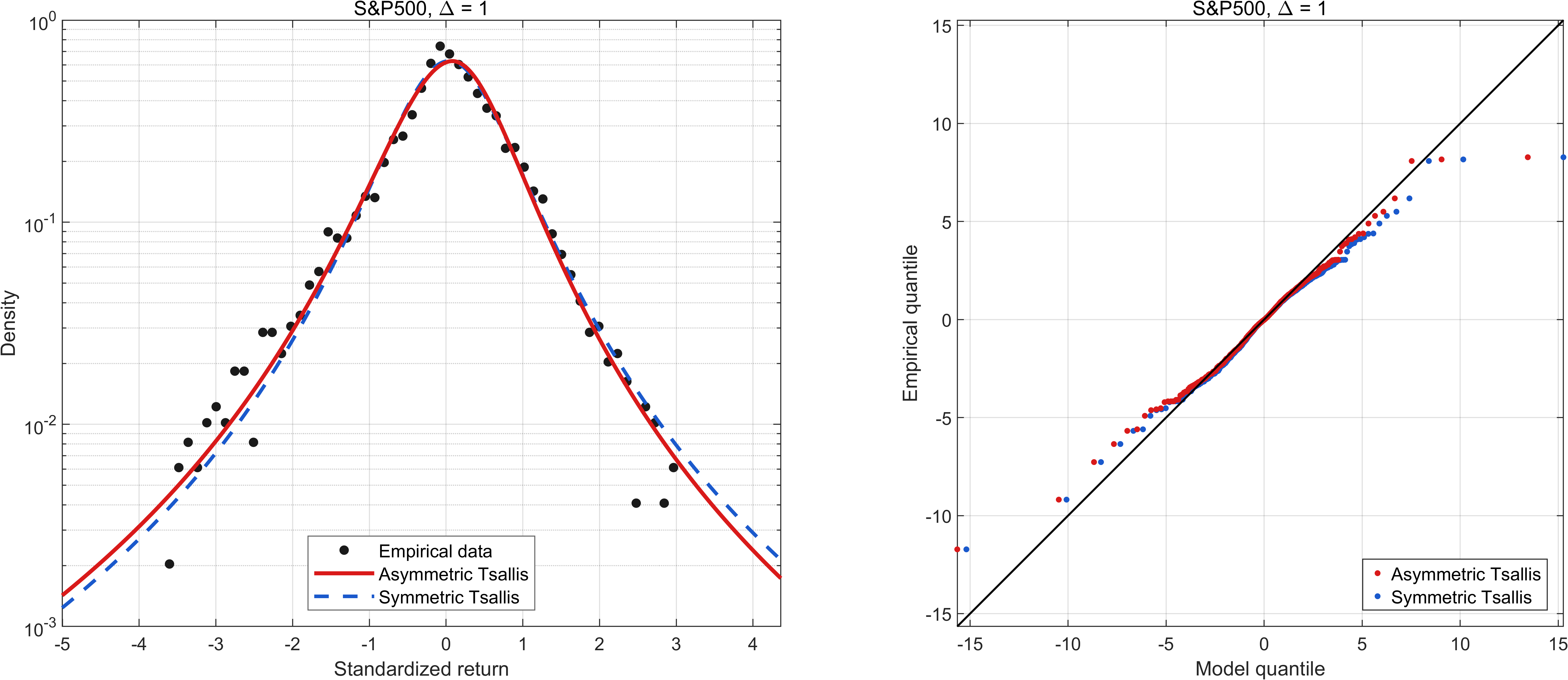}
  \caption{S\&P500, $\Delta=1$.}
  \label{fig:pdfqq_SP500_d1}
\end{figure}

\begin{figure}[!h]
  \centering
  \includegraphics[width=\textwidth]{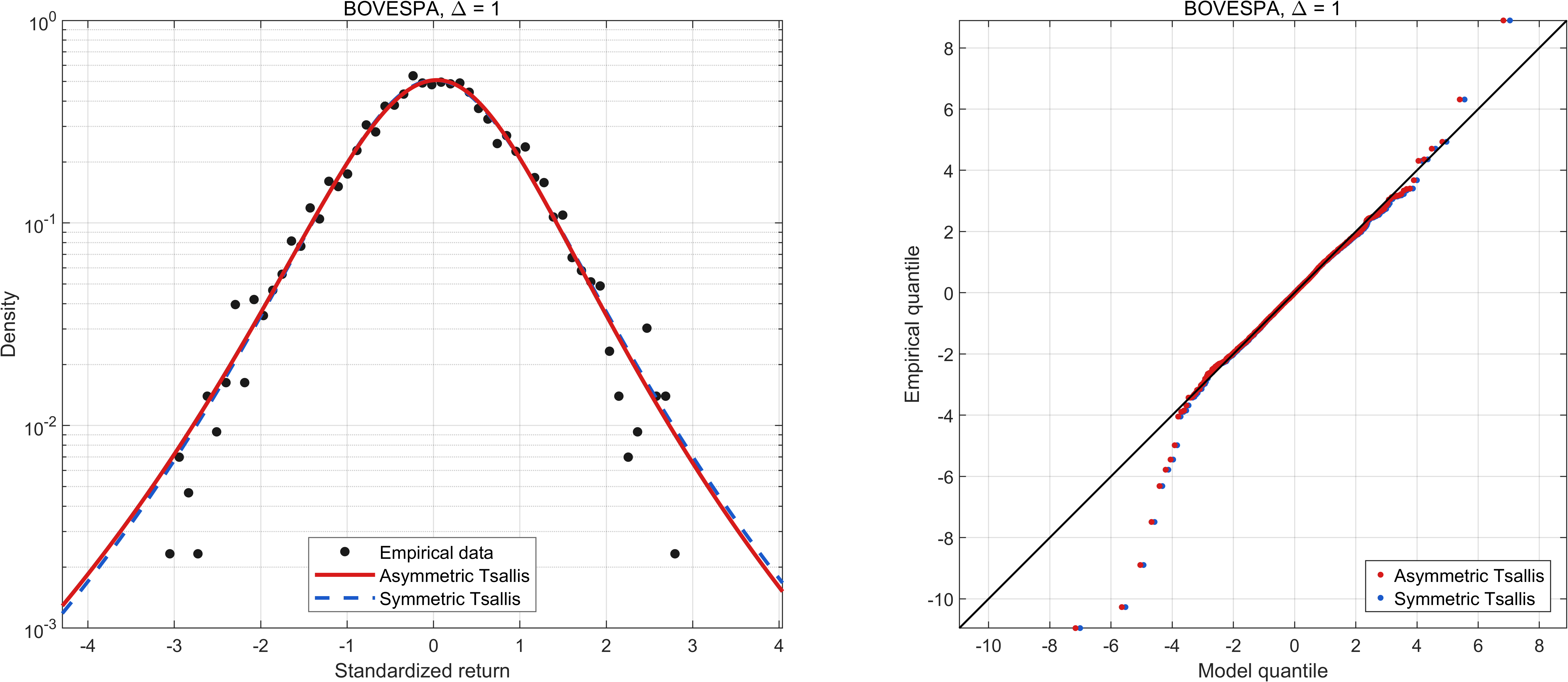}
  \caption{BOVESPA, $\Delta=1$.}
  \label{fig:pdfqq_BOVESPA_d1}
\end{figure}

\begin{figure}[!h]
  \centering
  \includegraphics[width=\textwidth]{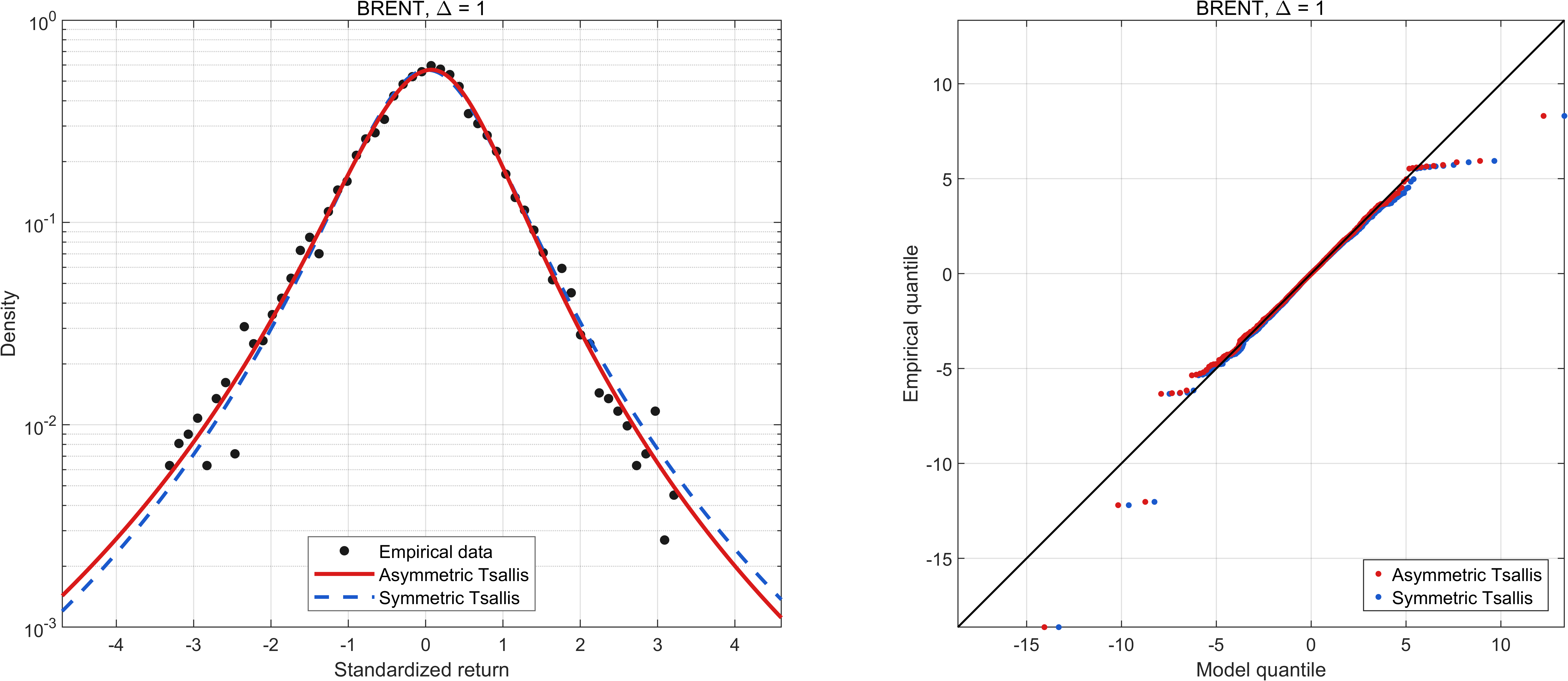}
  \caption{BRENT, $\Delta=1$.}
  \label{fig:pdfqq_BRENT_d1}
\end{figure}

\begin{figure}[!h]
  \centering
  \includegraphics[width=\textwidth]{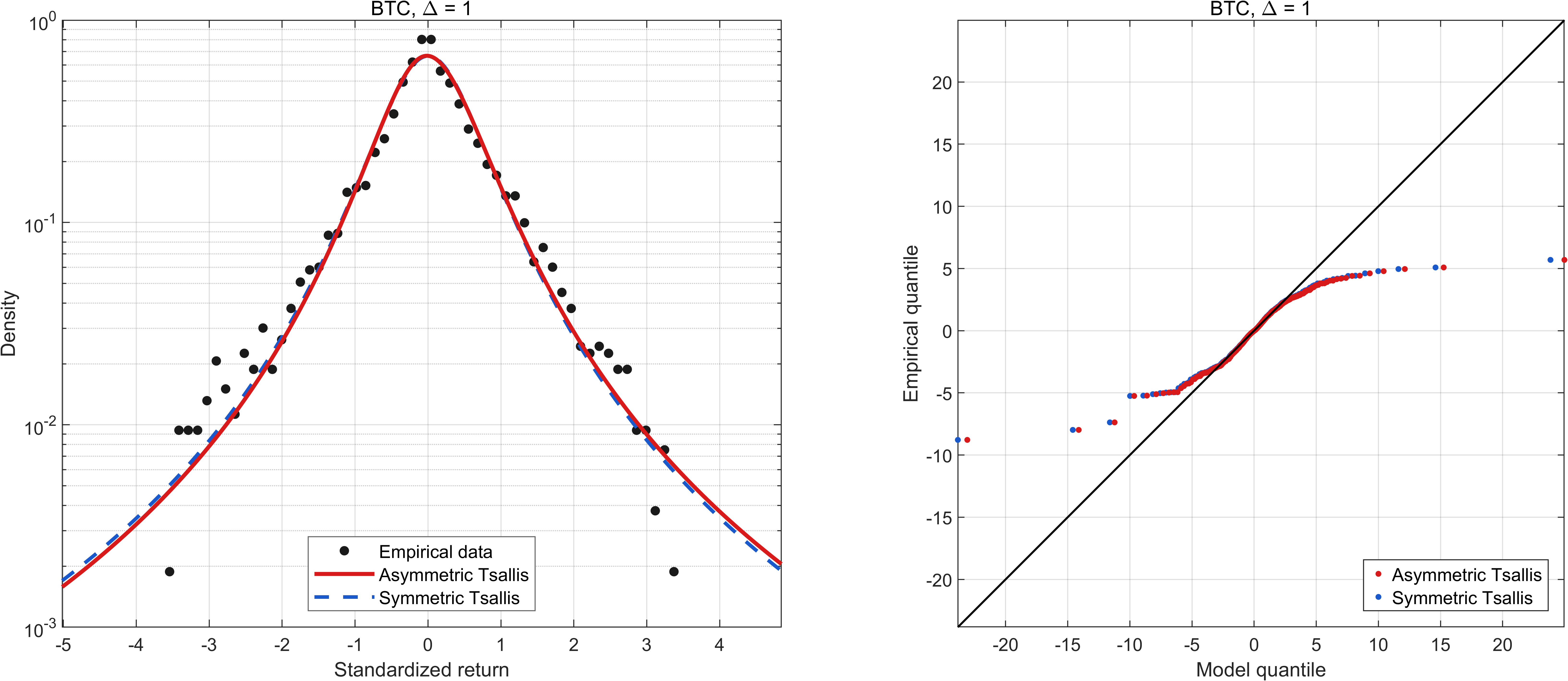}
  \caption{BTC, $\Delta=1$.}
  \label{fig:pdfqq_BITCOIN_d1}
\end{figure}

\begin{figure}[!h]
  \centering
  \includegraphics[width=\textwidth]{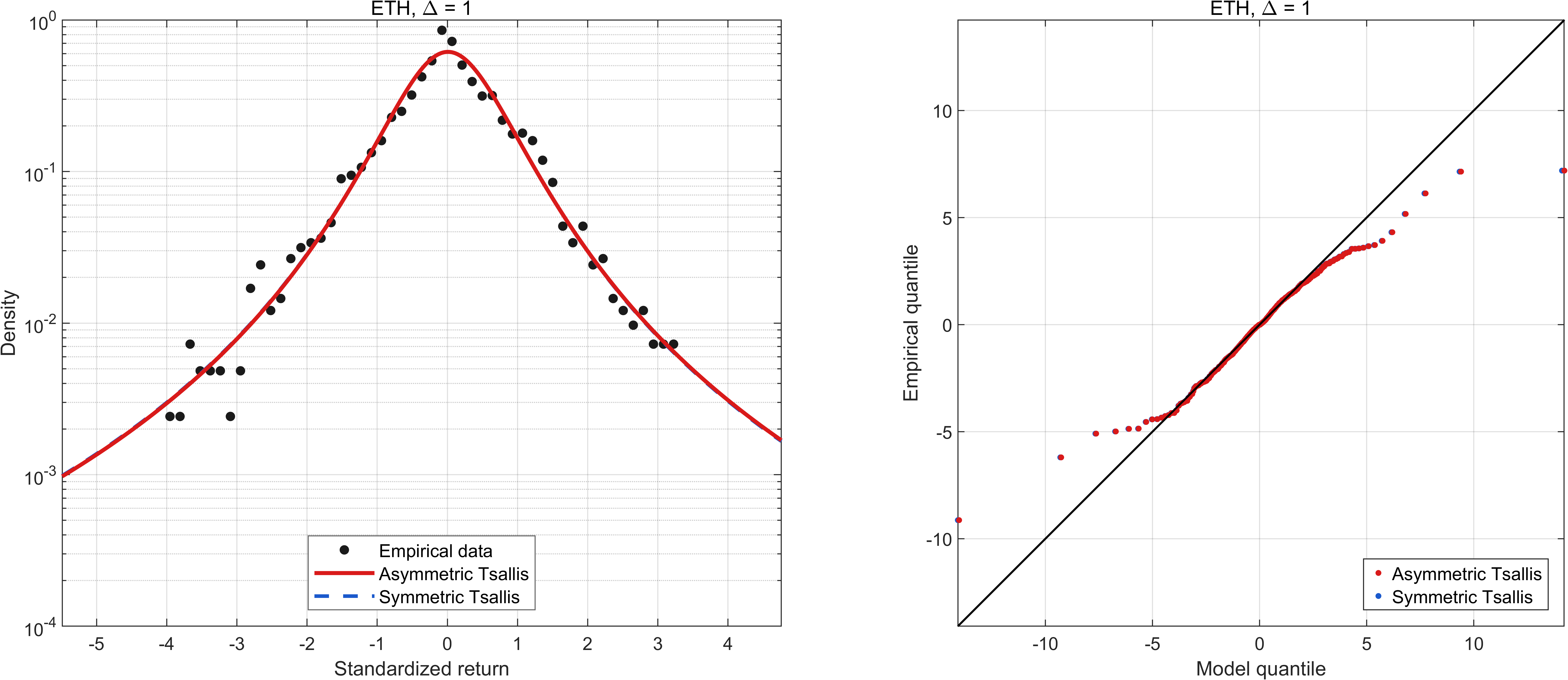}
  \caption{ETH, $\Delta=1$.}
  \label{fig:pdfqq_ETHEREUM_d1}
\end{figure}

\FloatBarrier

\subsection{Horizon $\Delta=7$}

\begin{figure}[!h]
  \centering
  \includegraphics[width=\textwidth]{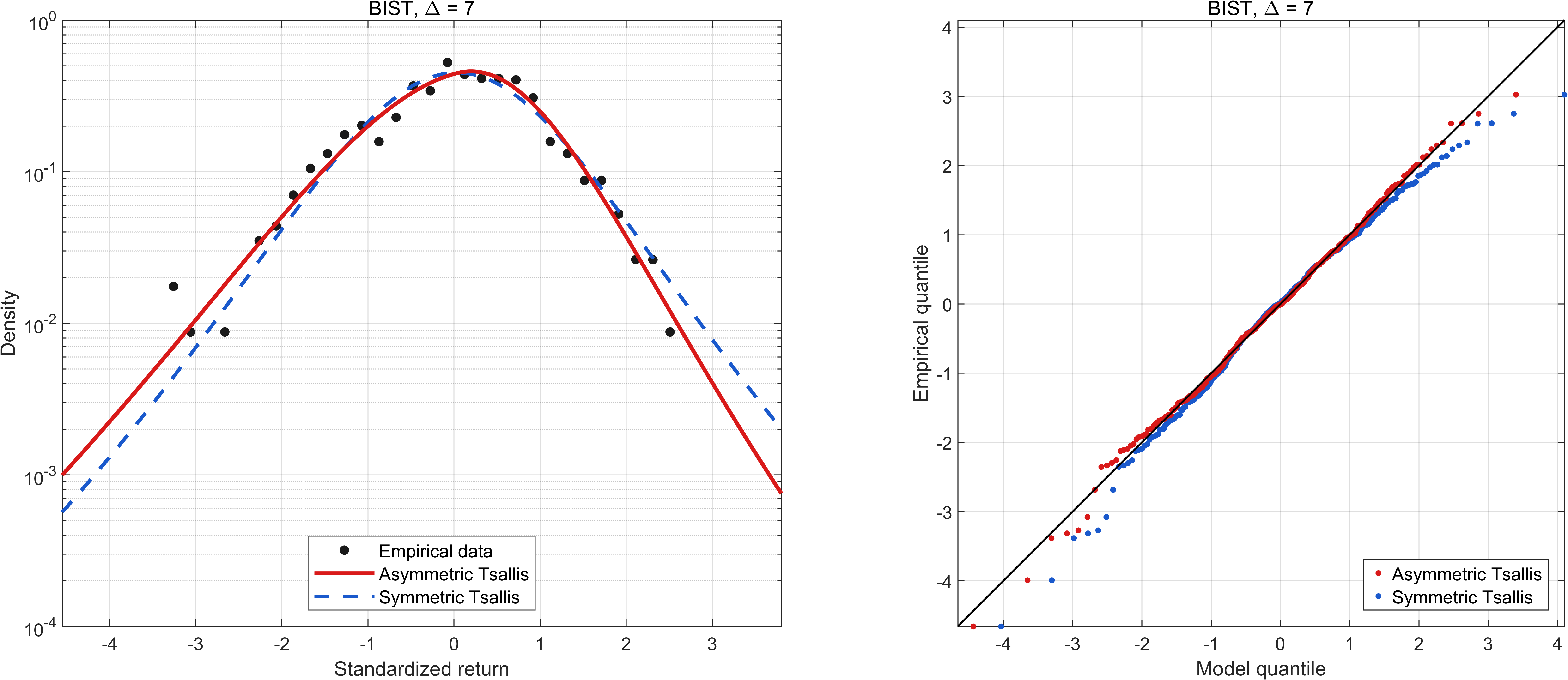}
  \caption{BIST, $\Delta=7$.}
  \label{fig:pdfqq_BIST_d7}
\end{figure}

\begin{figure}[!h]
  \centering
  \includegraphics[width=\textwidth]{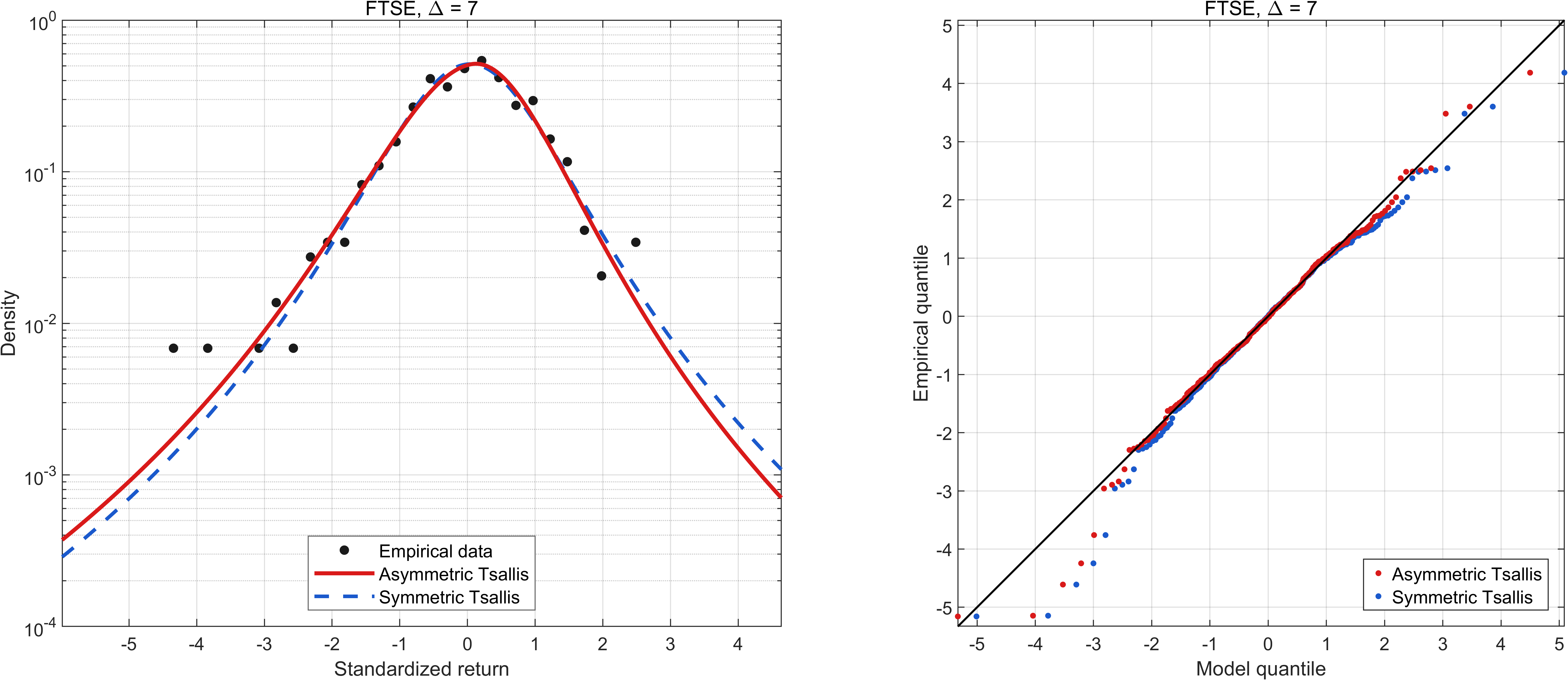}
  \caption{FTSE, $\Delta=7$.}
  \label{fig:pdfqq_FTSE_d7}
\end{figure}

\begin{figure}[!h]
  \centering
  \includegraphics[width=\textwidth]{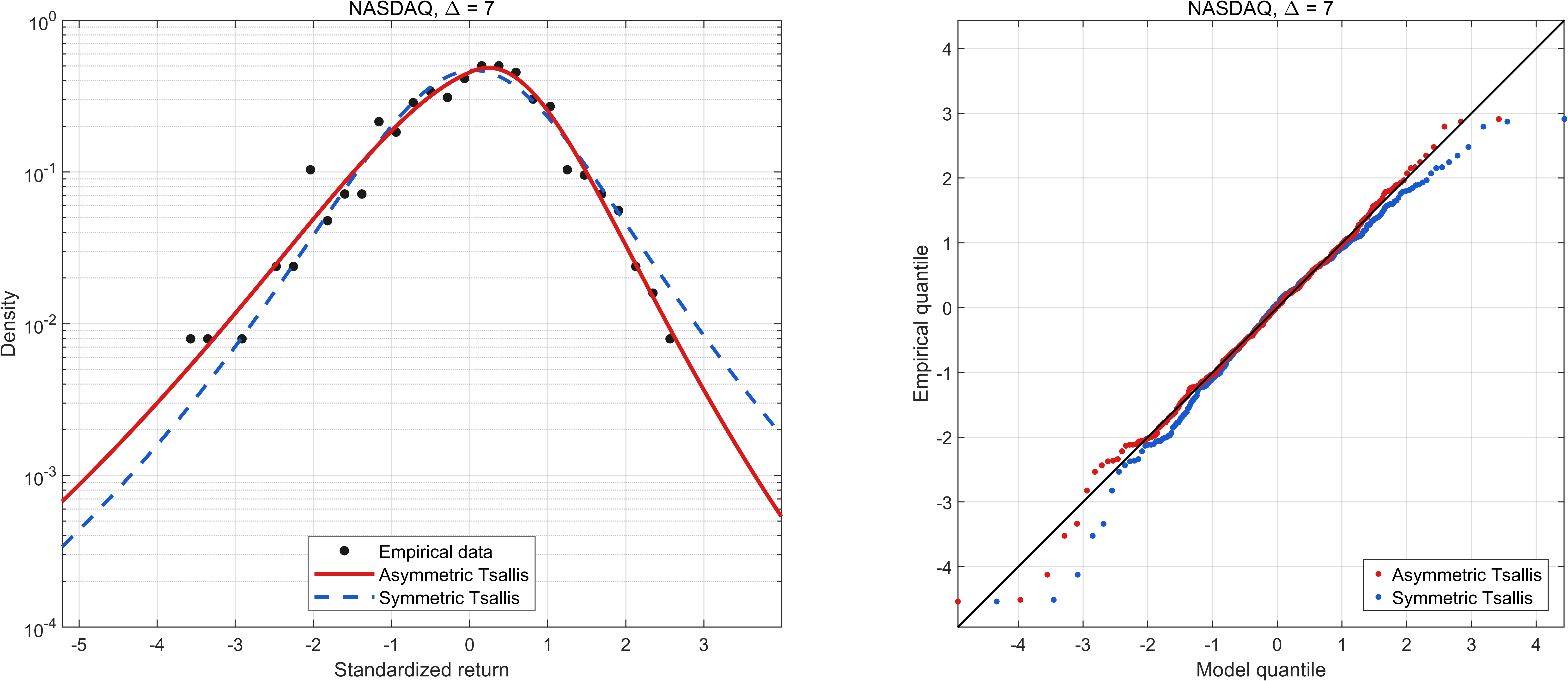}
  \caption{NASDAQ, $\Delta=7$.}
  \label{fig:pdfqq_NASDAQ_d7}
\end{figure}

\begin{figure}[!h]
  \centering
  \includegraphics[width=\textwidth]{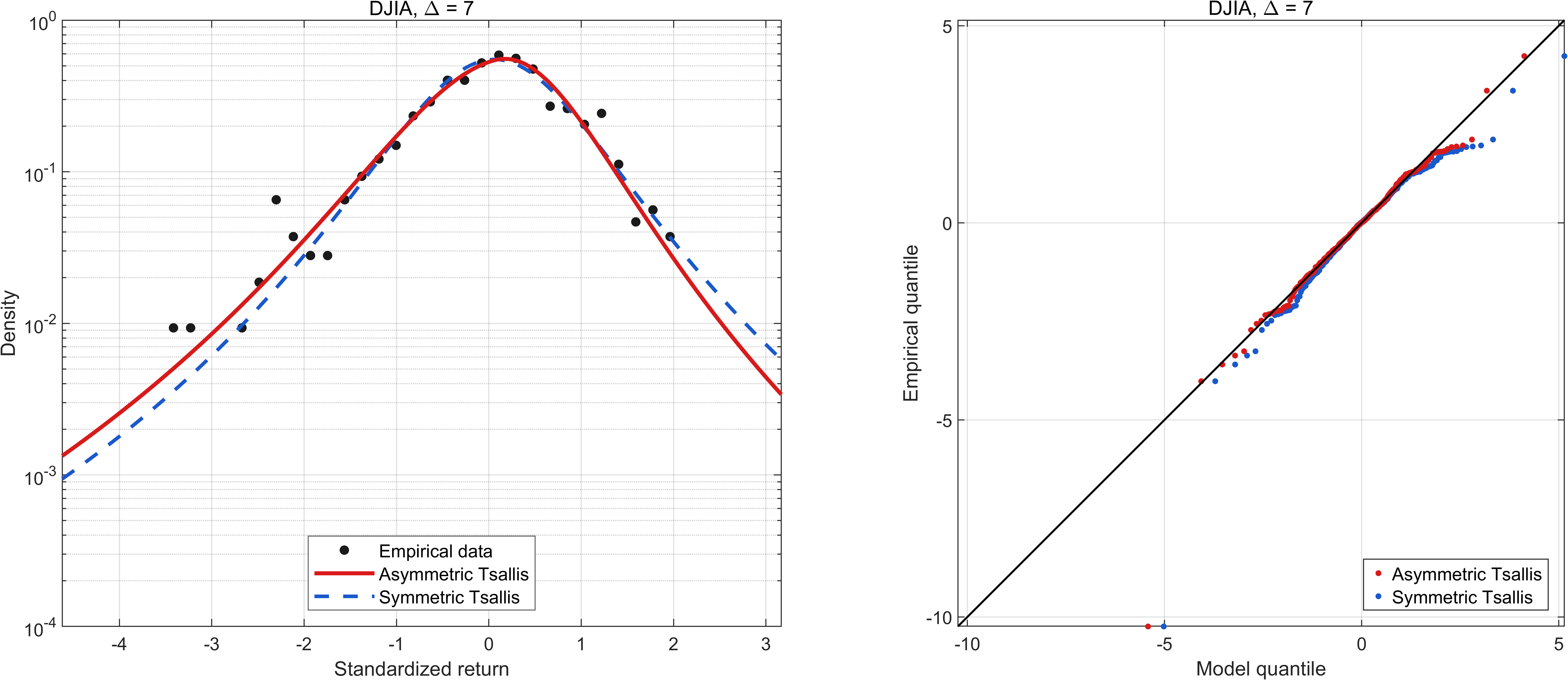}
  \caption{DJIA, $\Delta=7$.}
  \label{fig:pdfqq_DOWJONES_d7}
\end{figure}

\begin{figure}[!h]
  \centering
  \includegraphics[width=\textwidth]{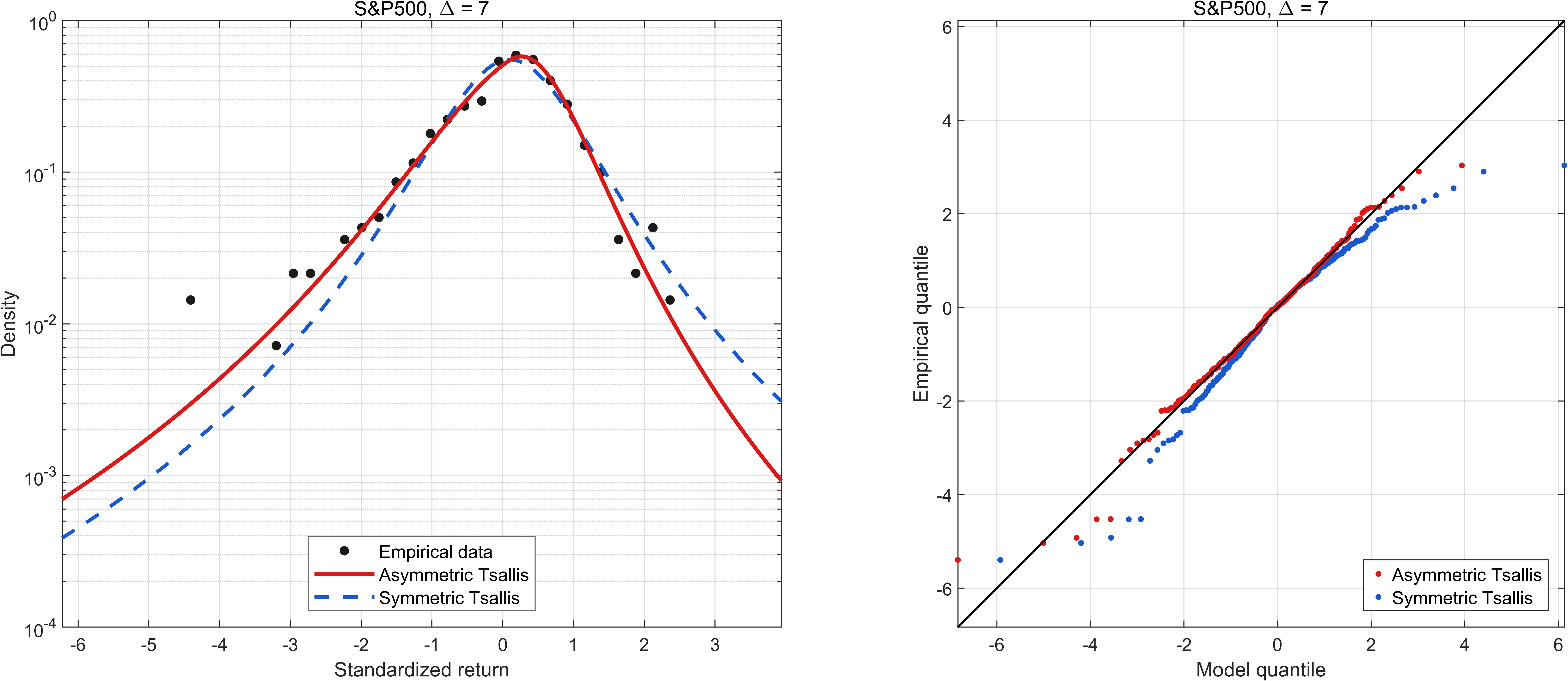}
  \caption{S\&P500, $\Delta=7$.}
  \label{fig:pdfqq_SP500_d7}
\end{figure}

\begin{figure}[!h]
  \centering
  \includegraphics[width=\textwidth]{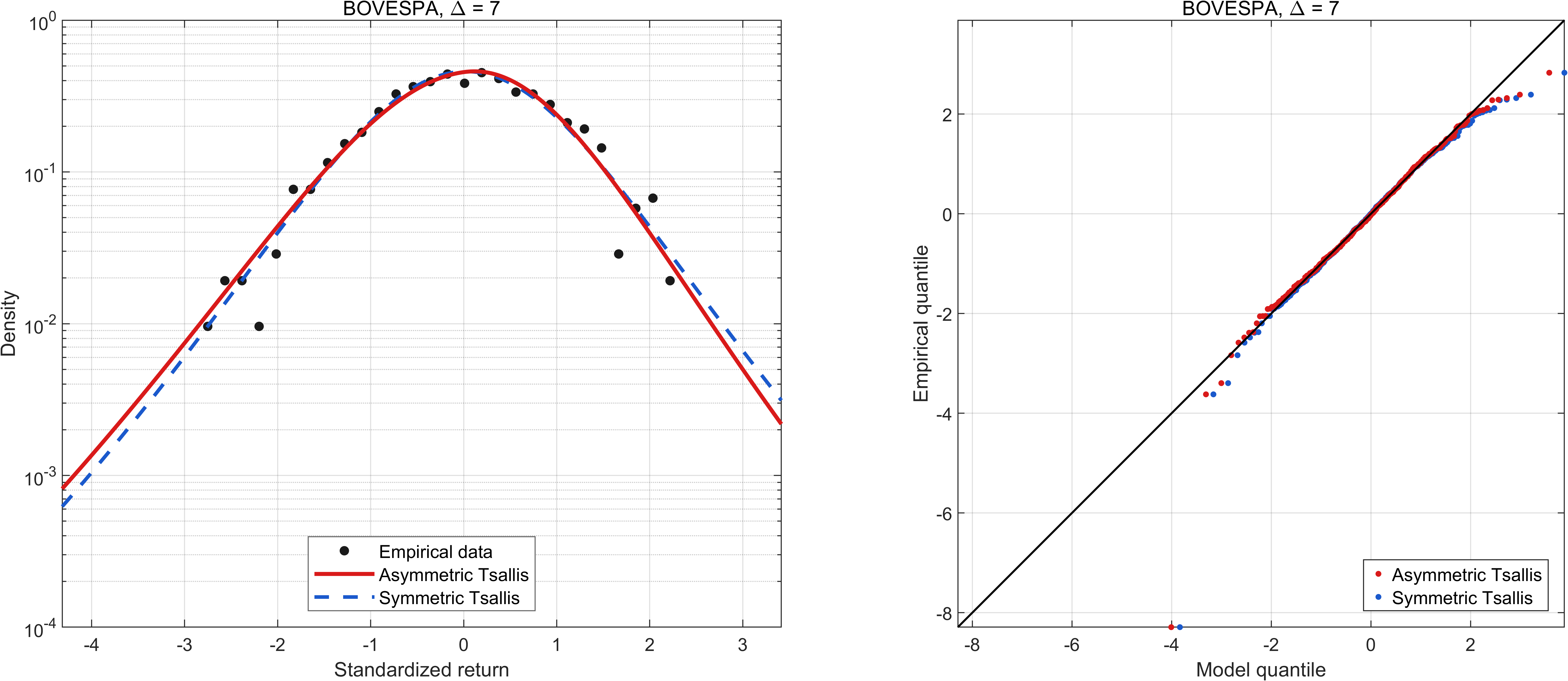}
  \caption{BOVESPA, $\Delta=7$.}
  \label{fig:pdfqq_BOVESPA_d7}
\end{figure}

\begin{figure}[!h]
  \centering
  \includegraphics[width=\textwidth]{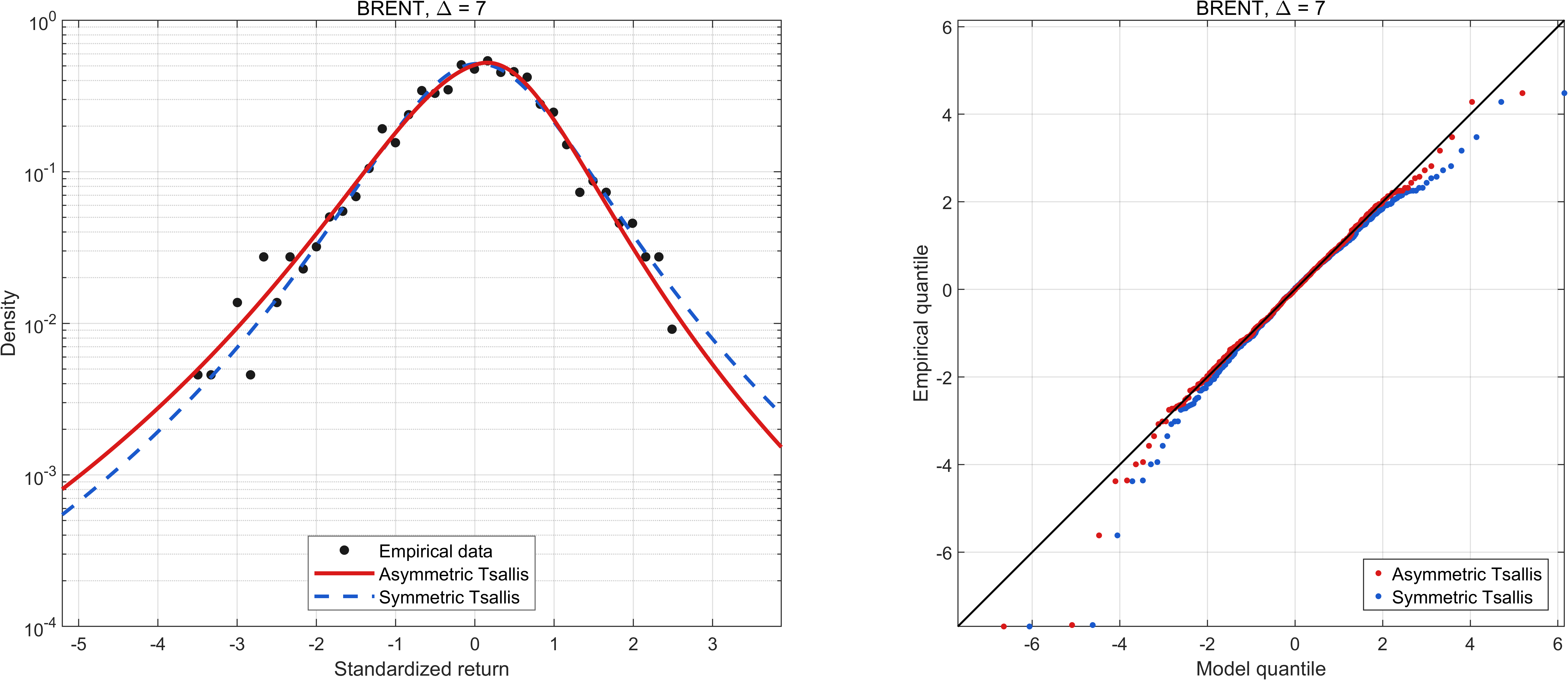}
  \caption{BRENT, $\Delta=7$.}
  \label{fig:pdfqq_BRENT_d7}
\end{figure}

\begin{figure}[!h]
  \centering
  \includegraphics[width=\textwidth]{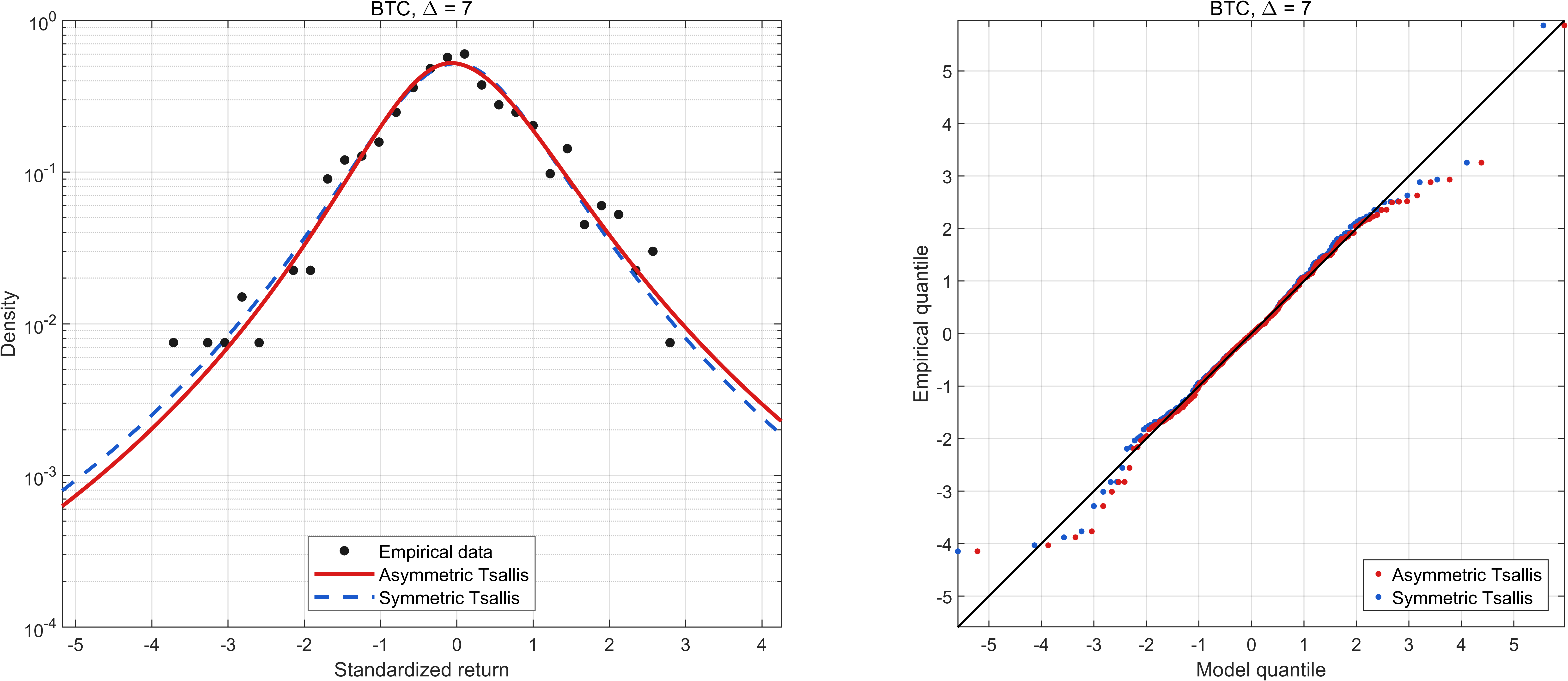}
  \caption{BTC, $\Delta=7$.}
  \label{fig:pdfqq_BITCOIN_d7}
\end{figure}

\begin{figure}[!h]
  \centering
  \includegraphics[width=\textwidth]{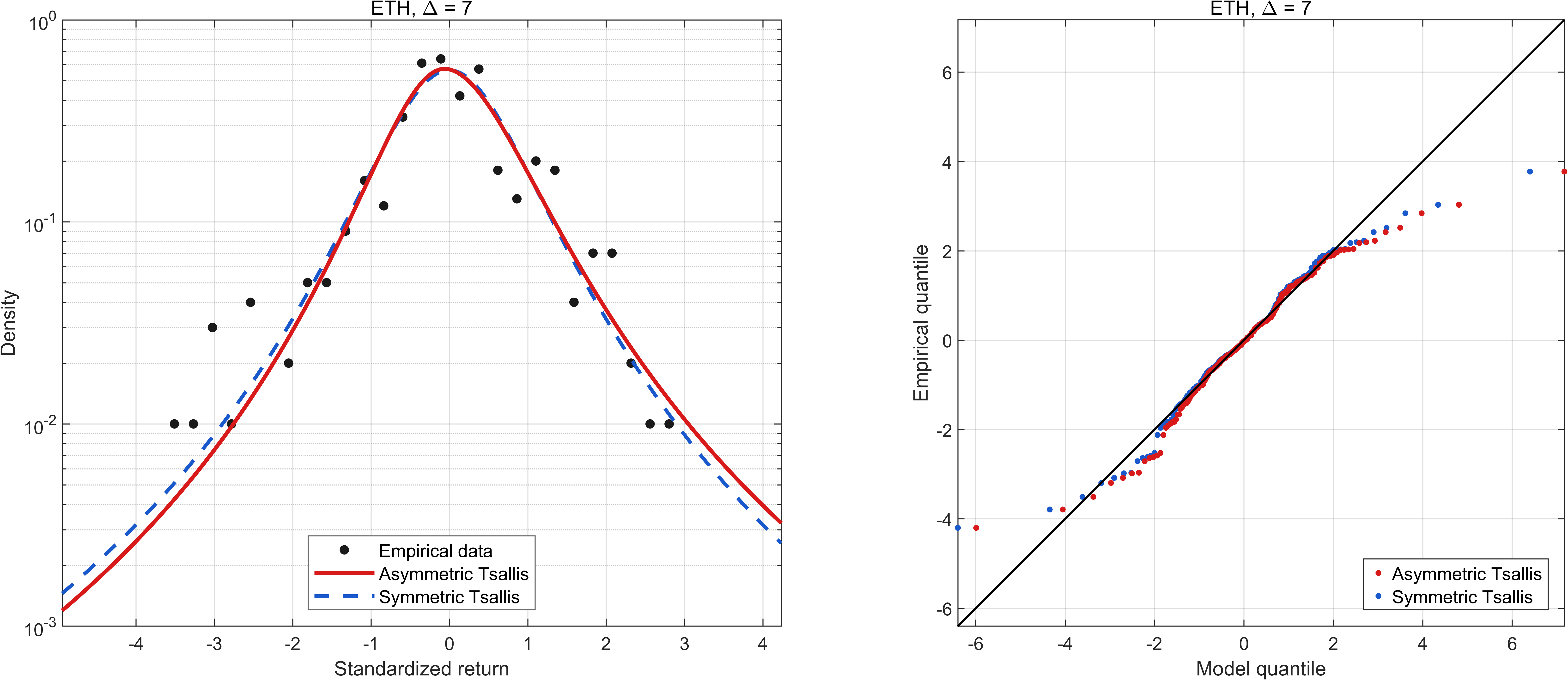}
  \caption{ETH, $\Delta=7$.}
  \label{fig:pdfqq_ETHEREUM_d7}
\end{figure}

\FloatBarrier

\subsection{Horizon $\Delta=20$}

\begin{figure}[!h]
  \centering
  \includegraphics[width=\textwidth]{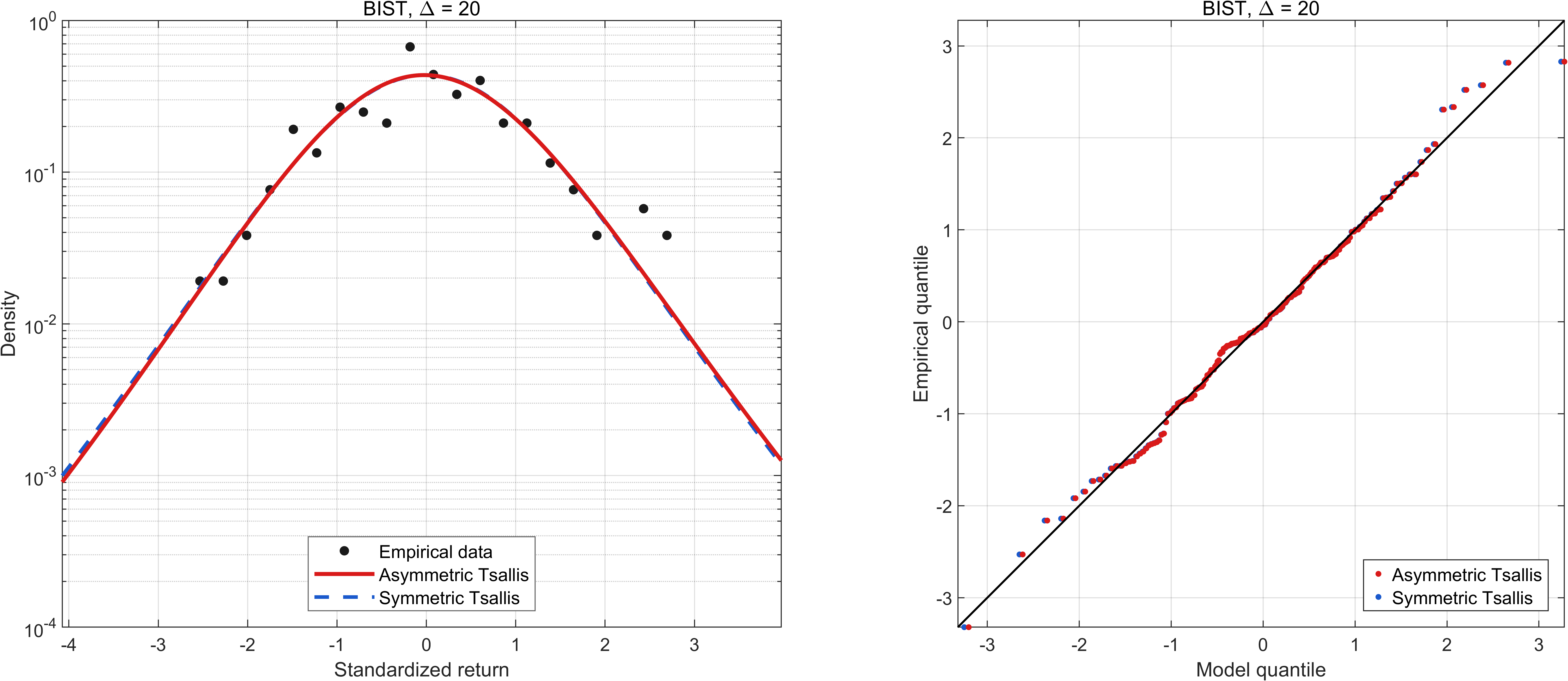}
  \caption{BIST, $\Delta=20$.}
  \label{fig:pdfqq_BIST_d20}
\end{figure}

\begin{figure}[!h]
  \centering
  \includegraphics[width=\textwidth]{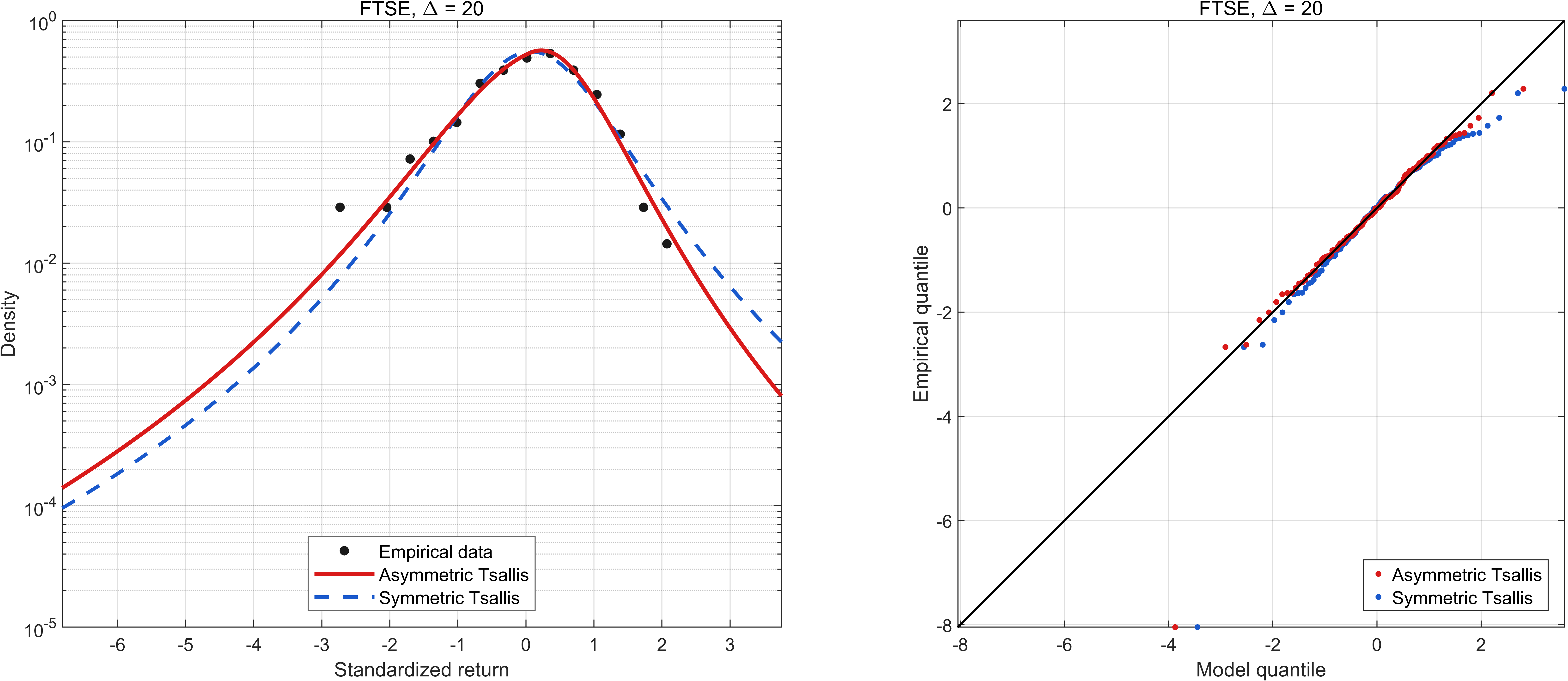}
  \caption{FTSE, $\Delta=20$.}
  \label{fig:pdfqq_FTSE_d20}
\end{figure}

\begin{figure}[!h]
  \centering
  \includegraphics[width=\textwidth]{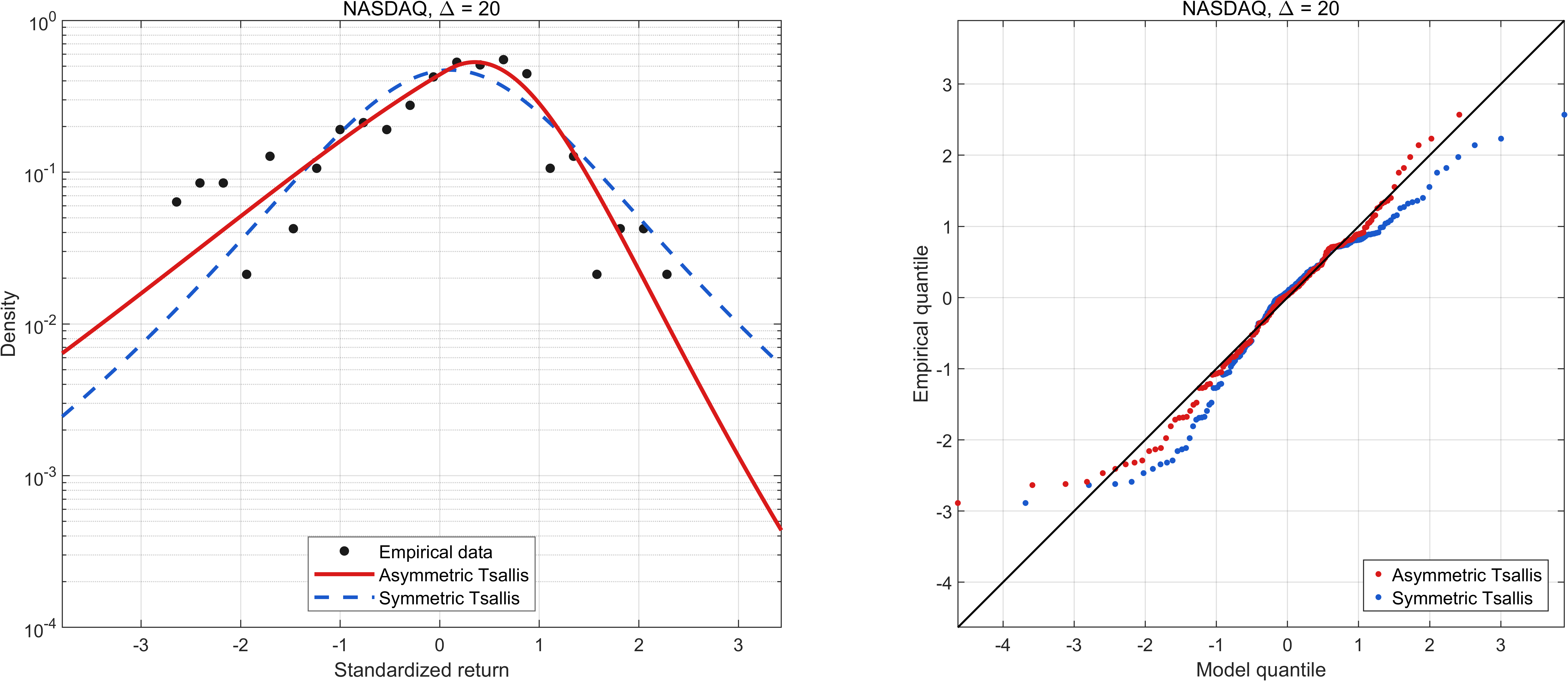}
  \caption{NASDAQ, $\Delta=20$.}
  \label{fig:pdfqq_NASDAQ_d20}
\end{figure}

\begin{figure}[!h]
  \centering
  \includegraphics[width=\textwidth]{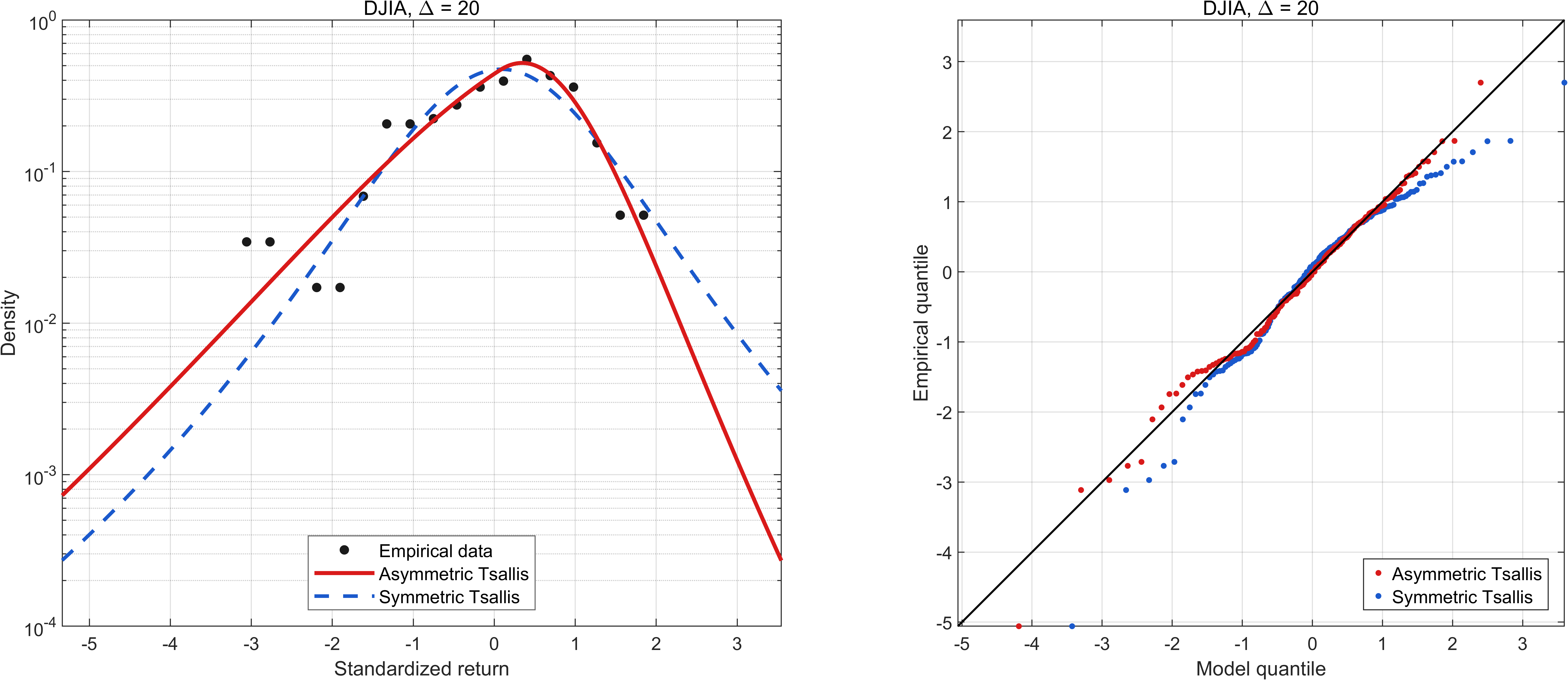}
  \caption{DJIA, $\Delta=20$.}
  \label{fig:pdfqq_DOWJONES_d20}
\end{figure}

\begin{figure}[!h]
  \centering
  \includegraphics[width=\textwidth]{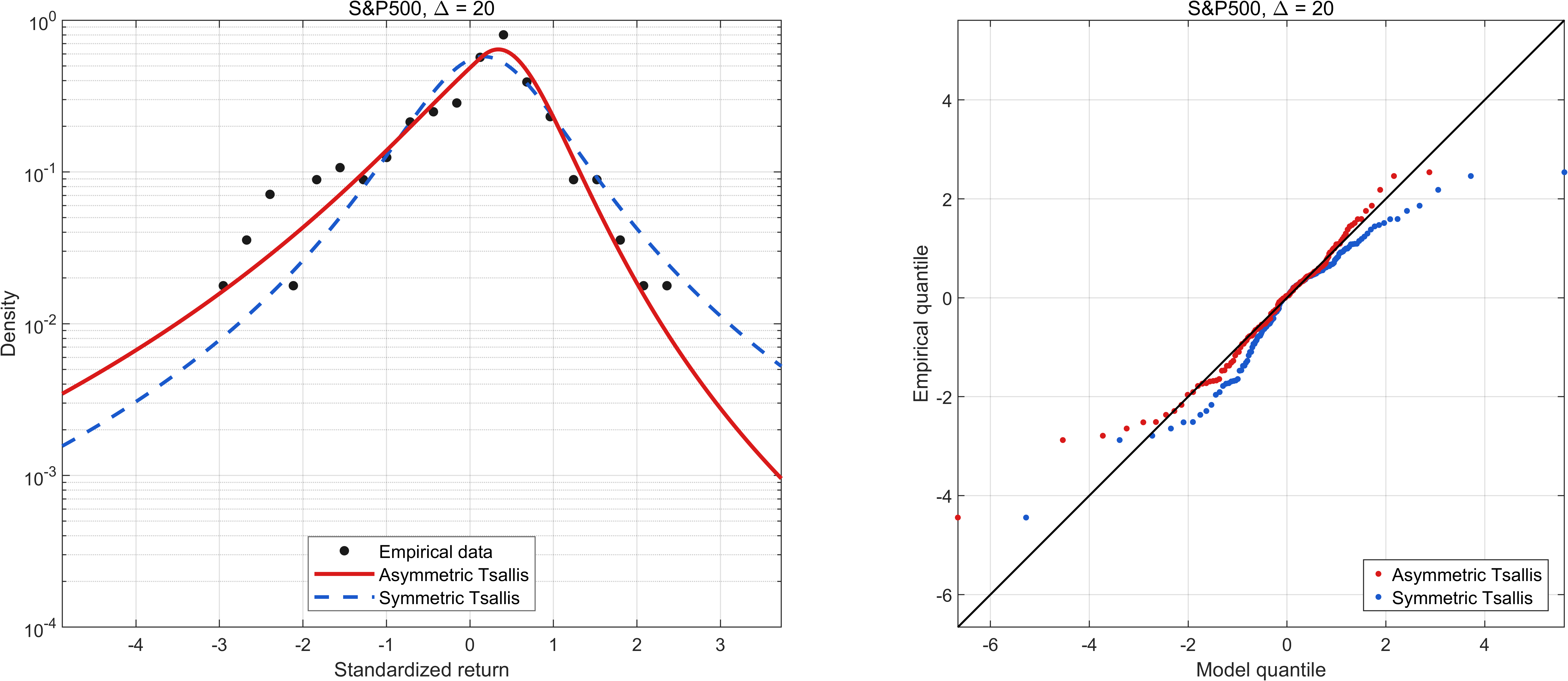}
  \caption{S\&P500, $\Delta=20$.}
  \label{fig:pdfqq_SP500_d20}
\end{figure}

\begin{figure}[!h]
  \centering
  \includegraphics[width=\textwidth]{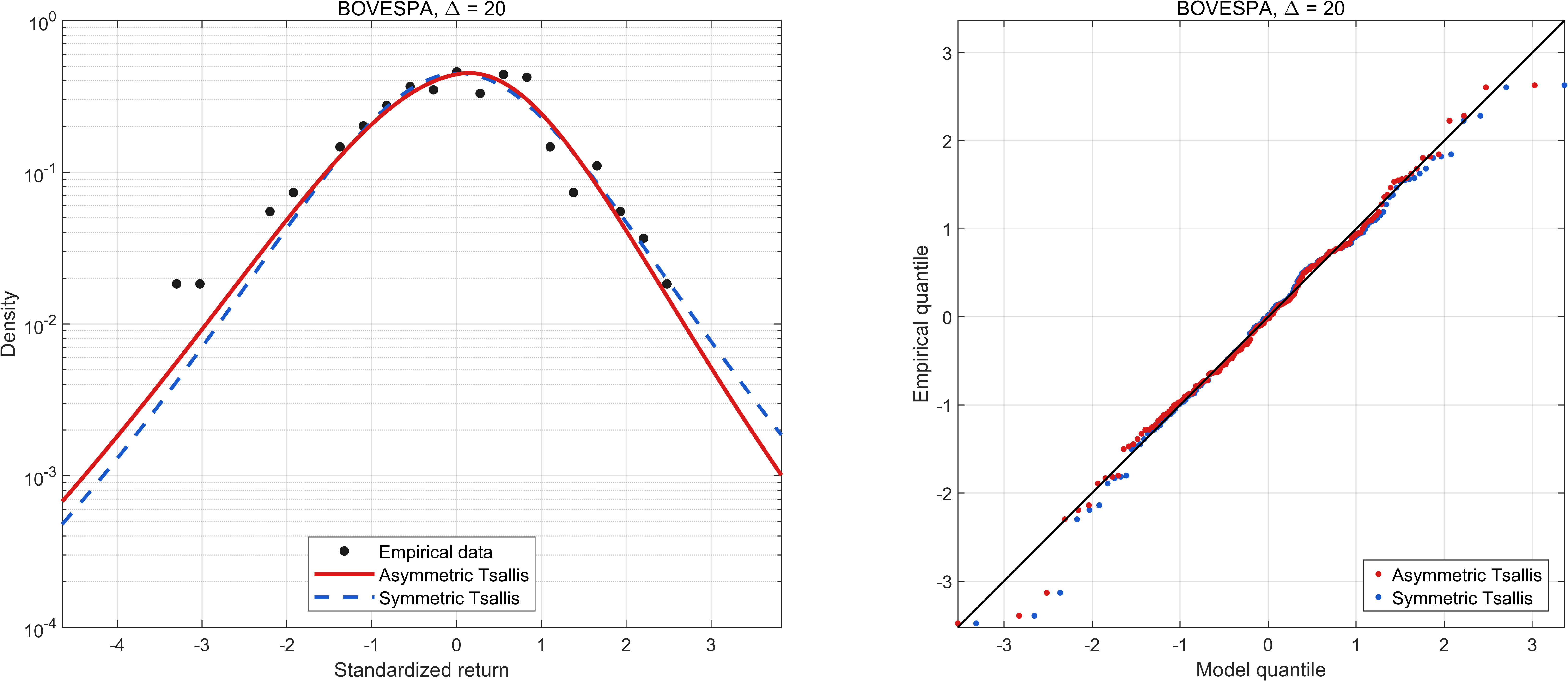}
  \caption{BOVESPA, $\Delta=20$.}
  \label{fig:pdfqq_BOVESPA_d20}
\end{figure}

\begin{figure}[!h]
  \centering
  \includegraphics[width=\textwidth]{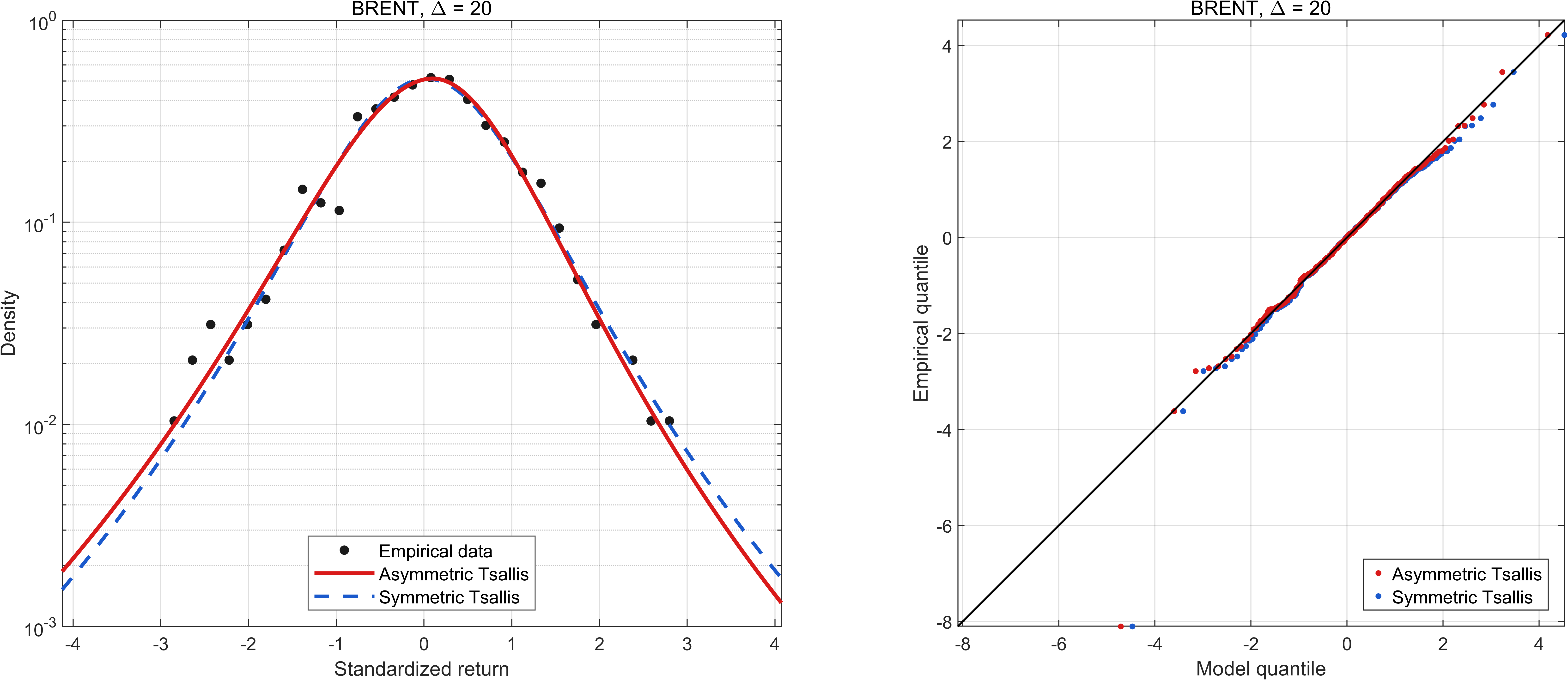}
  \caption{BRENT, $\Delta=20$.}
  \label{fig:pdfqq_BRENT_d20}
\end{figure}

\begin{figure}[!h]
  \centering
  \includegraphics[width=\textwidth]{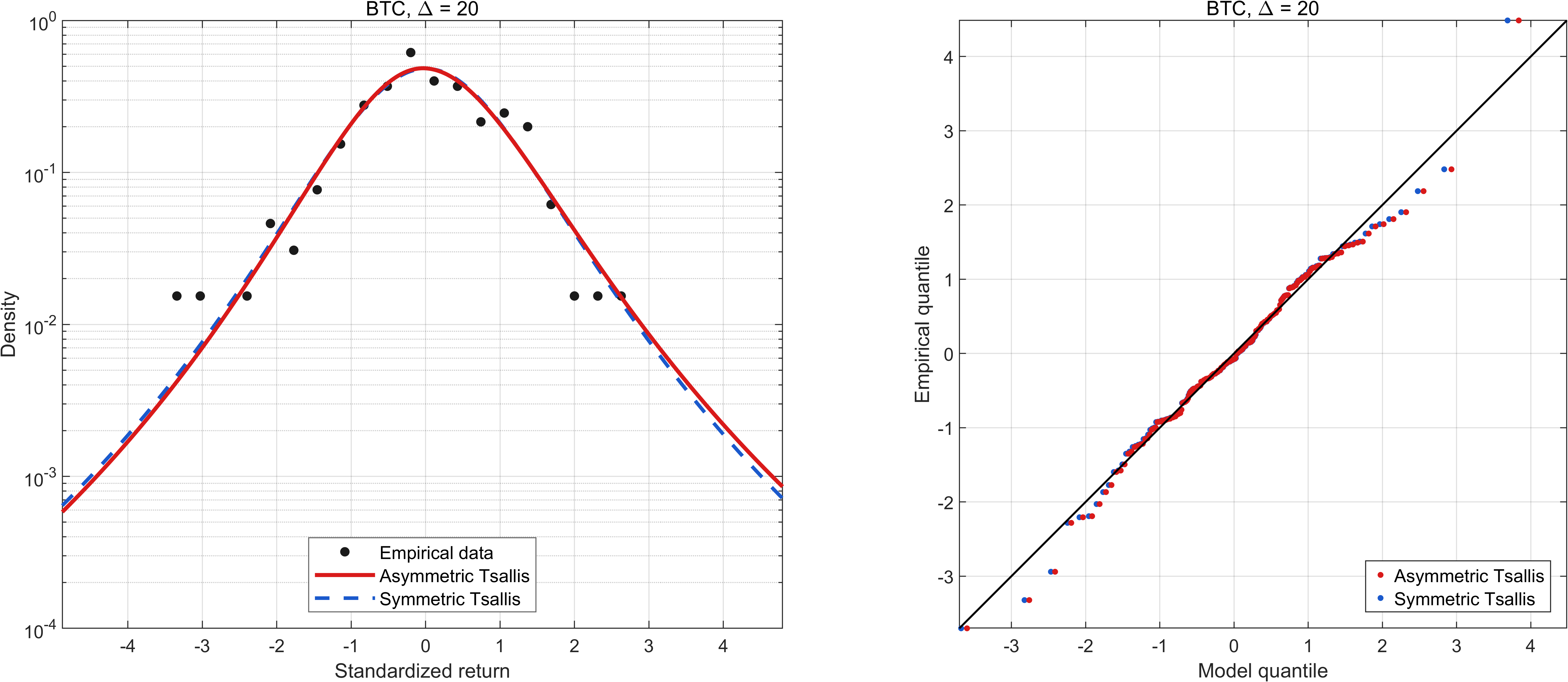}
  \caption{BTC, $\Delta=20$.}
  \label{fig:pdfqq_BITCOIN_d20}
\end{figure}

\begin{figure}[!h]
  \centering
  \includegraphics[width=\textwidth]{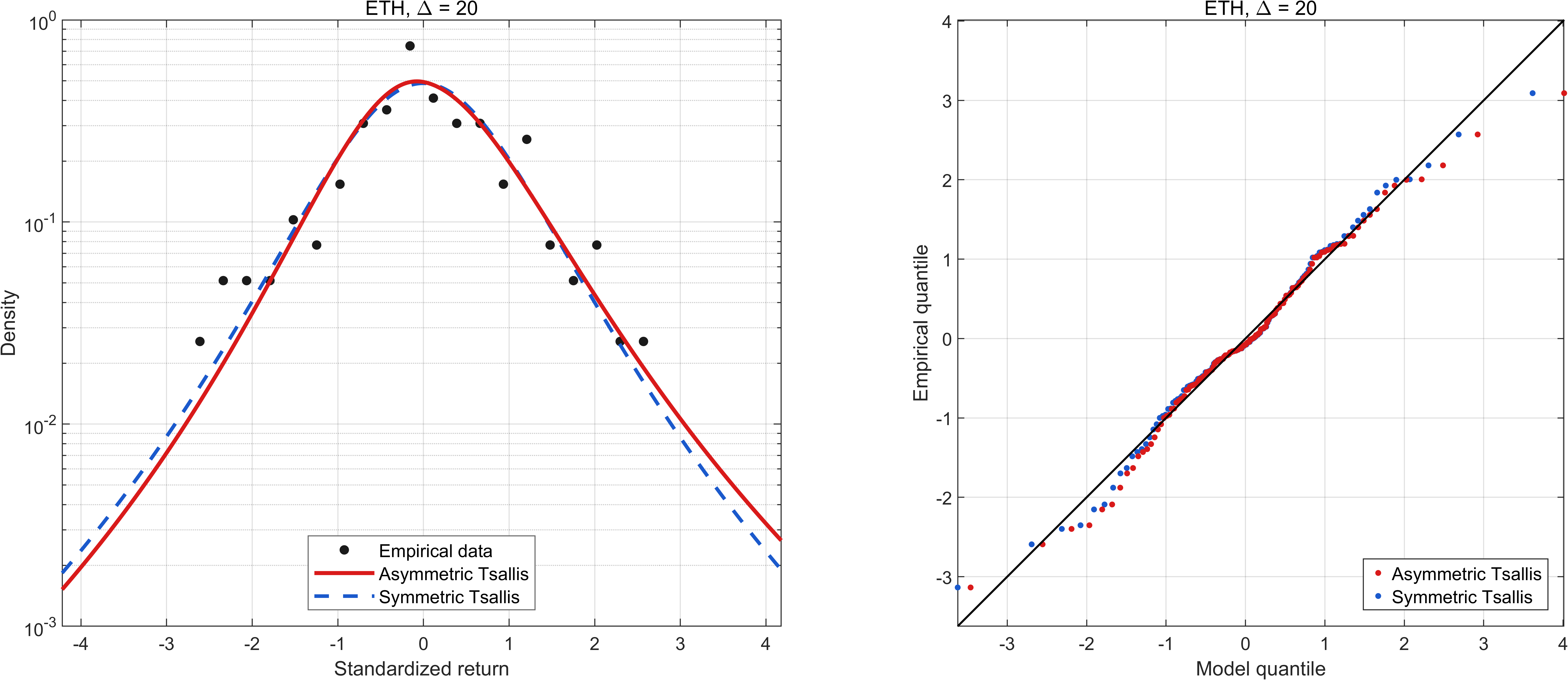}
  \caption{ETH, $\Delta=20$.}
  \label{fig:pdfqq_ETHEREUM_d20}
\end{figure}

\FloatBarrier

\end{document}